\documentclass[11pt,letterpaper]{article}

\usepackage{slivkins-setup,slivkins-theorems,my-acm} 
\usepackage{xspace,hyperref}
\usepackage{algorithm}
\usepackage[noend]{algorithmic-hacked}

\usepackage[round]{natbib}

	\advance \topmargin by -\headheight
	\advance \topmargin by -\headsep

	\evensidemargin \oddsidemargin
\renewcommand{\eqref}[1]{(\ref{#1})}
\renewcommand{\refeq}[1]{(\ref{#1})}

\newcommand{\A}{\ensuremath{\mA}}
\newcommand{\F}{\ensuremath{\mF}}

\newcommand{\mS}{\ensuremath{\mathcal{S}}}

\newcommand{\myball}{extensive-form ball\xspace}
\newcommand{\myballs}{extensive-form balls\xspace}

\newcommand{\NaiveAlg}{\ensuremath{\mathtt{UniformMesh}}\xspace}
\newcommand{\NaiveExp}{\ensuremath{\mathtt{UniformMeshExp}}\xspace}
\newcommand{\UCB}{\ensuremath{\mathtt{UCB1}}\xspace}
\newcommand{\EXP}{\ensuremath{\mathtt{EXP3}}\xspace}

\newcommand{\problem}{Lipschitz MAB problem\xspace}
\newcommand{\FFproblem}{Lipschitz experts problem\xspace}
\newcommand{\ULproblem}{uniformly Lipschitz experts problem\xspace}

\newcommand{\quasimetric}{quasi-distance}

\newcommand{\Cdbl}{\ensuremath{c_{\textsc{dbl}}}} 

\newcommand{\COV}{\mathtt{COV}}		
\newcommand{\MinCOV}{\mathtt{MinCOV}}	
\newcommand{\RegretDim}{\mathtt{DIM}}	
\newcommand{\MaxMinCOV}{\mathtt{MaxMinCOV}}	

\newcommand{\Bobby}[1]{\sidenote{Bobby}{#1}}

\newcommand{\payoff}{{\pi}}

\newcommand{\prob}{{\mathbb{P}}}

\renewcommand{\Bobby}[1]{}

\newcommand{\LCD}{\ensuremath{\mathtt{LCD}}} 
\newcommand{\MaxMinLCD}{\ensuremath{\mathtt{MaxMinLCD}}} 

\begin{document}


\title{{\bf Bandits and experts in metric spaces}%
\footnote{
This manuscript is a merged and definitive version of \citet*{LipschitzMAB-stoc08} and \citet*{DichotomyMAB-soda10}, which supersedes their full versions (published as technical reports on {\tt arxiv.org} in September 2008 and November 2009, respectively).
 Compared to the conference publications, the manuscript contains full proofs and a significantly revised presentation. In particular, it develops new terminology and modifies the proof outlines to unify the technical exposition of the two papers. The manuscript also features an updated discussion of the follow-up work and open questions. All results on the zooming algorithm and the max-min-covering dimension are from \citet{LipschitzMAB-stoc08}; all results on regret dichotomies and on Lipschitz experts are from \citet{DichotomyMAB-soda10}.}}


\author{
{\Large Robert Kleinberg}\thanks{
Computer Science Department, Cornell University (Ithaca NY, USA). Email:~{\tt rdk at cs.cornell.edu}. Supported in part by NSF awards CCF-0643934 and  CCF-0729102.
Some of the work was done while the author was visiting Microsoft Research.
}\and
{\Large Aleksandrs Slivkins}\thanks{
Aleksandrs Slivkins, Microsoft Research (New York NY, USA).
Email:~{\tt slivkins at microsoft.com}.
Some of the work was done while the author was a postdoctoral research associate at Brown University.
}\and
{\Large Eli Upfal}\thanks{
Eli Upfal, Computer Science Department, Brown University (Providence RI, USA). Email:~{\tt eli at cs.brown.edu}.
Supported in part by NSF awards CCR-0121154 and DMI-0600384, and ONR Award N000140610607.
}}

\date{First version: December 2013\\Last substantial revision: April 2018}

\maketitle

\thispagestyle{empty}

\begin{abstract}
In a multi-armed bandit problem, an online algorithm
chooses from a set of strategies in a sequence of
trials so as to maximize the total payoff of the
chosen strategies.
While the performance of bandit algorithms with
a small finite strategy set is quite well understood,
bandit problems with large strategy sets are still a
topic of very active
investigation, motivated by practical applications such as
online auctions and web advertisement.
The goal of such research is to identify broad and natural
classes of strategy sets and payoff functions which enable
the design of efficient solutions.

In this work we study a very general setting for the
multi-armed bandit problem in which the strategies form
a metric space, and the payoff function satisfies a Lipschitz
condition with respect to the metric. We refer to this
problem as the {\it Lipschitz MAB problem}. We present a
solution for the multi-armed bandit problem in this setting.
That is, for every metric space
we define an isometry invariant which
bounds from below the performance of Lipschitz
MAB algorithms for this metric space, and we present an algorithm which
comes arbitrarily close to meeting this bound. Furthermore, our
technique gives even better results for benign payoff functions.
We also address the full-feedback (``best expert'') version of the problem, where after every round the payoffs from all arms are revealed.
\end{abstract}

{\bf ACM Categories:}
\category{F.2.2}{Analysis of Algorithms and Problem Complexity}{Nonnumerical Algorithms and Problems}
\category{F.1.2}{Computation by Abstract Devices}{Modes of Computation}[Online computation]

\vspace{2mm}
{\bf Keywords:} multi-armed bandits, regret, online learning, metric spaces, covering dimension, Lipschitz-continuity


%
%
%


\newpage \thispagestyle{empty}
\begin{small}
\setcounter{tocdepth}{2}
\tableofcontents
\end{small}
\newpage

\section{Introduction}
\label{sec:intro}

\newcommand{\willcite}[1][Cite]{{[\sc #1]}}

In a multi-armed bandit problem, an online algorithm
must iteratively choose from a set of possible strategies (also called ``arms'') in a sequence of
$n$ trials so as to maximize the total payoff of the
chosen strategies.
These problems are the principal theoretical
tool for modeling the exploration/exploitation tradeoffs
inherent in  sequential
decision-making under uncertainty.
Studied
intensively for 
decades \citep{Thompson-1933,Robbins1952,Berry-book,CesaBL-book,Gittins-book11,Bubeck-survey12},
bandit problems are having an increasingly visible
impact on computer science because of their diverse
applications including online auctions, adaptive
routing, and the theory of learning in games.
The performance of a multi-armed bandit algorithm is often evaluated in terms of its {\it regret}, defined as the gap between the expected payoff of the algorithm and that of an optimal strategy.
While the performance of bandit algorithms with
a small finite strategy set is quite well understood,
bandit problems with exponentially or infinitely
large strategy sets are still a topic of very active
investigation (see \citet{Bubeck-survey12} for a survey).


Absent any assumptions about the strategies and their payoffs, bandit problems with large strategy sets allow for no non-trivial solutions --- any multi-armed bandit algorithm performs as badly, on some inputs, as random guessing. But in most applications it is natural to assume a structured class of payoff functions, which often enables the design of efficient learning algorithms \cite{Bobby-thesis}.  In this paper, we consider a broad and natural class of problems in which the structure is induced by a metric on the space of strategies.  While bandit problems have been studied in a few specific metric spaces (such as a one-dimensional interval) \citep{Agrawal-bandits-95,AuerOS/07,Cope/06,Bobby-nips04,yahoo-bandits07}, the case of general metric spaces has not been treated before, despite being an extremely natural setting for bandit problems.

As a motivating example, consider the problem faced by a website choosing from a database of thousands of banner ads to display to users, with the aim of maximizing the click-through rate of the ads displayed by matching ads to users' characterizations and the web content that they are currently watching. Independently experimenting with each advertisement is infeasible, or at least highly inefficient, since the number of ads is too large. Instead, the advertisements are usually organized into a taxonomy based on metadata (such as the category of product being advertised) which allows a similarity measure to be defined.  The website can then attempt to optimize its learning algorithm by generalizing from experiments with one ad to make inferences about the performance of similar ads \citep{yahoo-bandits07,yahoo-bandits-icml07}.

Another motivating example is revenue-management problems (e.g., see \citep{KleinbergL03,BZ09}). Consider a monopolistic seller with unlimited inventory of many digital products, such as songs, movies, or software. Customers arrive over time, and the seller can give customized offers to each arriving customer so as to maximize the revenue. The space of possible offers is very large, both in terms of possible product bundles and in terms of the possible prices, so experimenting with each and every offer is inefficient. Instead, the seller may be able to use experiments with one offer to make inferences about similar offers.

Abstractly, we have a bandit
problem of the following form: there is a strategy
set $X$, with an unknown payoff function
$\mu \,:\, X \rightarrow [0,1]$ satisfying a set
of predefined constraints of the form
$|\mu(x) - \mu(y)| \leq \delta(x,y)$ for some
$x,y \in X$ and $\delta(x,y) > 0$.  In each
period the algorithm chooses a point
$x \in X$ and receives payoff --  a number in the $[0,1]$ interval --
sampled independently from some distribution $\prob_x$ whose expectation is  $\mu(x)$.

A moment's thought reveals that this abstract problem
can be regarded as a bandit problem in a metric space.
Specifically, define $\mD(x,y)$ to be the
infimum of the quantity $\sum_i \delta(x_i,x_{i+1})$ over all finite paths
$(x = x_0,\, x_1, \,\ldots,\, x_k = y)$ in $X$.
Then $\mD$ is a metric
and the constraints $|\mu(x)-\mu(y)|
< \delta(x,y)$ may be summarized (equivalently reformulated) as follows:
\begin{align}\label{eq:Lipschitz-condition}
	|\mu(x)-\mu(y)| \leq \mD(x,y)	\quad
		\text{for all $x,y\in X$}.
\end{align}
In words, $\mu$ is a Lipschitz function (of Lipschitz constant
$1$) on the metric space $(X,\mD)$.

We assume that an algorithm is given the metric space $(X,\mD)$ as an input, with a promise that the payoff function $\mu$ satisfies~\eqref{eq:Lipschitz-condition}. We refer to this problem as the \emph{\problem} on $(X,\mD)$, and
we refer to the ordered triple $(X,\mD,\mu)$ as
an \emph{instance} of the \problem.\footnote{
Formally, the problem instance also includes the parameterized family of reward distributions $\prob_x$. To simplify exposition, we assume this family is fixed throughout, and therefore can be suppressed from the notation. When the metric space $(X,\mD)$ is understood from
context, we may also refer to $\mu$ as an instance.}

\subsection{Prior work}
\label{sec:intro.priorwork}

While our work is the first to treat the Lipschitz MAB
problem in general metric spaces, special cases of the
problem are implicit in prior work on the continuum-armed
bandit problem \citep{Agrawal-bandits-95,AuerOS/07,Cope/06,Bobby-nips04}
--- which corresponds to the space $[0,1]$ under the metric
$\ell_1^{1/d}$, $d \geq 1$
--- and the experimental work
on ``bandits for taxonomies'' \citep{yahoo-bandits07},
which corresponds to the
case in which $(X,\mD)$ is a tree metric.%
\footnote{Throughout the paper, $\ell_p$, $p\geq 1$ denotes a metric on a finite-dimensional real space given by
$\ell_p(x,y) = \|x-y\|_p$. }
Also, \citet{Hazan-colt07} considered a contextual bandit setting with a metric space on contexts rather than arms.

Before describing our results in greater detail, it is
helpful to put them
in context by recounting the nearly optimal bounds for the
one-dimensional continuum-armed bandit problem, a problem
first formulated in \citet{Agrawal-bandits-95}
and  solved (up to logarithmic factors)
by various authors \citep{AuerOS/07,Cope/06,Bobby-nips04}.
In the following theorem
and throughout this paper, the \emph{regret} of a multi-armed
bandit algorithm $\A$ running on an instance $(X,\mD,\mu)$ is
defined to be the function $R_{\A}(t)$ which measures
the difference between its expected payoff at time $t$
and the quantity $t \cdot \sup_{x\in X} \mu(x)$.
The latter quantity is the expected payoff of always playing
an arm $x \in \argmax \mu(x)$
if such arm exists. In regret-minimization, the main issue is typically how regret scales with $t$.

\begin{theorem}[\cite{AuerOS/07,Cope/06,Bobby-nips04}%
\footnote{\citet{AuerOS/07} and \citet{Cope/06} also achieve $\tilde{O}(\sqrt{T})$ regret under additional assumptions on the shape of the function near its optimum. }
] \label{thm:intro.1d}
For any $d \geq \N$, consider the Lipschitz MAB
problem on $([0,1],\ell_1^{1/d})$, $d\geq 1$. There is an algorithm whose regret on any instance $\mu$ satisfies
	$ R(t) = \tilde{O}(t^\gamma)$
for every $t$, where $\gamma = \frac{d+1}{d+2}$.
No such algorithm exists for any $\gamma < \frac{d+1}{d+2}$.
\end{theorem}

In fact, if the time horizon $t$ is known in advance,
the upper bound in the theorem can be achieved
by an extremely na\"{i}ve algorithm which
uses an optimal $k$-armed bandit algorithm, such as the
\UCB algorithm \citep{bandits-ucb1}, to choose
arms from the set
	$S = \{0, \tfrac1k,\, \tfrac2k,\,\ldots,1\}$,
for a suitable choice of the parameter $k$.
Here the arms in $S$ partition the strategy set in a uniform (and non-adaptive) way; hence, we call this algorithm \NaiveAlg.

\subsection{Initial result}

We make an initial observation that the analysis of algorithm \NaiveAlg in Theorem~\ref{thm:intro.1d} only relies on the covering properties of the metric space, rather than on its real-valued structure, and (with minor modifications) can be extended to any metric space of constant covering dimension $d$.

\begin{theorem}\label{thm:intro-covDim}
Consider the \problem on a metric space of covering dimension $d\geq 0$.
There is an algorithm whose regret on any instance $\mu$ satisfies
	$ R(t) = \tilde{O}(t^\gamma)$
for every $t$, where $\gamma = \frac{d+1}{d+2}$.
\end{theorem}

The covering dimension is a standard notion which summarizes covering properties of a metric space. It is defined as the smallest (infimum) number $d\geq 0$ such that $X$ can be covered by $O(\delta^{-d})$ sets of diameter $\delta$, for each $\delta>0$. We denote it $\COV(X,\mD)$, or $\COV(X)$ when the metric $\mD$ is clear from the context. The covering dimension generalizes the Euclidean dimension, in the sense that the covering dimension of $([0,1]^d,\ell_p)$, $p\geq 1$ is $d$. Unlike the Euclidean dimension, the covering dimension can take fractional values.
Theorem~\ref{thm:intro-covDim} generalizes the upper bound in Theorem~\ref{thm:intro.1d} because the covering dimension of $([0,1],\ell_1^{1/d})$ is $d$, for any $d\geq 1$.

\subsection{Present scope}
\label{subsec:intro-scope}

This paper is a comprehensive study of the \problem in arbitrary metric spaces.

While the regret bound in Theorem~\ref{thm:intro.1d} is
essentially optimal when the metric space is $([0,1],\,\ell_1^{1/d})$, it is strikingly odd that it is achieved by such a simple algorithm as \NaiveAlg.  In particular,
the algorithm approximates the strategy set by a
fixed mesh $S$ and does not refine
this mesh as it gains information about the location
of the optimal arm.  Moreover, the metric
contains seemingly useful proximity information,
but the algorithm ignores this information after
choosing its initial mesh.  Is this really the best
algorithm?

A closer examination of the lower bound proof
raises further reasons for suspicion: it is
based on a contrived, highly singular payoff function
$\mu$ that alternates between being constant on some
distance scales and being very steep on other (much smaller)
distance scales, to create a multi-scale ``needle in haystack''
phenomenon which nearly obliterates the usefulness
of the proximity information contained in the
metric $\ell_1^{1/d}$.  Can we expect algorithms to do better
when the payoff function is more benign?%
\footnote{Here and elsewhere we use ``benign'' as a non-technical term.}
For the
\problem on $([0,1],\ell_1)$, the question
was answered affirmatively in \citet{Cope/06,AuerOS/07}
for some classes of instances,
with algorithms that are tuned to the specific classes.

We are concerned with the following two directions motivated by the discussion above:
\begin{description}
\item[(Q1)] \emph{Per-metric optimality.}
What is the best possible bound on regret for a given metric space?
 (Implicitly, such regret bound is worst-case over all payoff functions consistent with this metric space.) Is \NaiveAlg, as naive as it is, really an optimal algorithm? Is covering dimension an appropriate structure to characterize such worst-case regret bounds?

\item[(Q2)] \emph{Benign problem instances.}
 Is it possible to take advantage of benign payoff functions? What structures would be useful to characterize benign payoff functions and the corresponding better-than-worst-case regret bounds? What algorithmic techniques would help?

\end{description}

Theorem~\ref{thm:intro-covDim} calibrates our intuition: for relatively ``rich'' metric spaces such as $([0,1],\ell_1^{1/d})$ we expect regret bounds of the form $\tilde{O}(t^\gamma)$, for some constant $\gamma\in (0,1)$ which depends on the metric space, and perhaps also on the problem instance. Henceforth, we will call this  \emph{polynomial regret}. Apart from metric spaces that admit polynomial regret, we are interested in the extremes: metric spaces for which the Lipschitz MAB problem becomes very easy or very difficult.

It is known that one can achieve logarithmic regret as long as the number of arms is finite \citep{Lai-Robbins-85,bandits-ucb1}.%
\footnote{The constant in front of the $\log(t)$ increases with the number of arms and also depends on instance-specific parameters.
$O(\log t)$ regret is optimal even for two arms \citep{Lai-Robbins-85}.}
On the other hand, all prior results for infinite metric spaces had regret $O(t^\gamma)$, $\gamma\geq \tfrac12$.
We view problem instances with $O(\log t)$ regret as ``very tractable'', and those with regret $t^{\gamma}$, $\gamma\geq \tfrac12$, as ``somewhat tractable''. It is natural to ask what is the transition between the two.

\begin{description}
\item[(Q3)] Is $\tilde{O}(\sqrt{t})$ regret the best possible for an infinite metric space? Alternatively, are there infinite metric spaces for which one can achieve regret $O(\log t)$? Is there any metric space for which the best possible regret is \emph{between} $O(\log t)$ and $\tilde{O}(\sqrt{t})$?
\end{description}

On the opposite end of the ``tractability spectrum'' of the \problem, there are metric spaces of infinite covering dimension, for which no algorithm can have regret of the form $O(t^\gamma)$, $\gamma<1$. Intuitively, such metric spaces are intractable. Formally, will define ``intractable'' metric spaces is those that do not admit sub-linear regret.

\begin{description}
\item[(Q4)] Which metric spaces are tractable, i.e. admit $o(t)$ regret?
\end{description}

We are also interested in the full-feedback version of the \problem, where after each round the payoff for each arm can be queried by the algorithm. Such settings have been extensively studied in the online learning literature under the name \emph{best experts problems} \citep{experts-jacm97,CesaBL-book,vovk98}. Accordingly, we call our setting the \emph{\FFproblem}. To the best of our knowledge, prior work for this setting includes the following two results: constant regret for a finite set of arms \citep{sleeping-colt08}, and $\tilde{O}(\sqrt{t})$ regret for metric spaces of bounded covering dimension \citep{Anupam-experts07}; the latter result uses a version of the \NaiveAlg. We are interested in per-metric optimality (Q1), including the extreme versions (Q3) and (Q4). For polynomial regret the goal is to handle metric spaces of infinite covering dimension.


\subsection{Our contributions: \problem}
\label{subsec:intro-poly}

We give a complete solution to (Q1), by describing for every metric space $(X,\mD)$ a family of algorithms which come arbitrarily close to achieving the best possible regret bound for this metric space. In particular, we resolve (Q3) and (Q4). We also give a satisfactory answer to (Q2); our solution is arbitrarily close to optimal in terms of the zooming dimension defined below.

\OMIT{ 
Our main technical contribution is a new algorithm,
the \emph{zooming algorithm},
that combines the upper confidence bound
technique used in earlier bandit algorithms
such as \UCB \citep{bandits-ucb1} with a novel \emph{adaptive
refinement} step that uses past history to zoom
in on regions near the apparent maxima of $\mu$
and to explore a denser mesh of strategies in
these regions.  This algorithm is a key ingredient
in our design of an optimal bandit algorithm for
every metric space $(X,\mD)$.  Moreover,
we show that the zooming algorithm can perform significantly
better on benign problem instances.  That is, for every instance
$(X,\mD,\mu)$ we define a parameter called the \emph{zooming
dimension} which is often significantly smaller than
$\MaxMinCOV(X)$, and we bound the algorithm's performance
in terms of the zooming dimension of the problem instance.
Since the zooming algorithm is self-tuning, it achieves
this bound without requiring prior knowledge of the
zooming dimension.
} 

Underpinning these contributions is a new algorithm, called the \emph{zooming algorithm}. It maintains a mesh of ``active arms'', but (unlike $\NaiveAlg$) adapts this mesh to the observed payoffs. It combines the upper confidence bound technique used in earlier bandit algorithms such as \UCB \citep{bandits-ucb1} with a novel \emph{adaptive refinement} step that uses past history to refine the mesh (``zoom in'') in regions with high observed payoffs. We show that the zooming algorithm can perform significantly better on benign problem instances. Moreover, it is a key ingredient in our design of a per-metric optimal bandit algorithm.


\xhdr{Benign problem instances.}
For every problem instance $(X,\mD,\mu)$ we define a parameter called the \emph{zooming dimension}, and use it to bound the performance of the zooming algorithm in a way that is often significantly stronger than the corresponding per-metric bound. Note that the zooming algorithm is \emph{self-tuning}, i.e. it achieves this bound without requiring prior knowledge of the zooming dimension. Somewhat surprisingly, our regret bound for the zooming algorithm result has exactly the same ``shape'' as Theorem~\ref{thm:intro-covDim}.

\begin{theorem} \label{thm:intro.zooming}
If $d$ is the zooming dimension of a Lipschitz MAB
instance then at any time $t$ the zooming
algorithm suffers regret
	$\tilde{O}(t^{\gamma})$, where $\gamma = \tfrac{d+1}{d+2}$.
\end{theorem}

The exponent $\gamma$ in the theorem is the best possible, as a function of $d$, in light of
Theorem~\ref{thm:intro.1d}.

While covering dimension is about covering the entire metric space, zooming dimension focuses on covering near-optimal arms. The lower bounds in Theorem~\ref{thm:intro.1d} and Theorem~\ref{thm:intro.pmo} are based on contrived examples with a high-dimensional set of near-optimal arms which leads to the ``needle-in-the-haystack'' phenomenon. We sidestep these examples if the set of near-optimal arms is low-dimensional, in the sense that we make formal below. We define the \emph{zooming dimension} of an instance $(X,\mD,\mu)$ as the smallest $d$ such that the following covering property holds:
for every $\delta>0$ we require only $O(\delta^{-d})$
sets of diameter $\delta/8$ to cover
the set of arms whose expected payoff falls short of
the optimum by an amount between $\delta$ and $2 \delta$.

Zooming dimension is our way to quantify the benignness of a problem instance. It is trivially no larger than the covering dimension, and can be significantly smaller. Below let us give some examples:

\begin{itemize}

\item Suppose a low-dimensional region $S\subset X$ contains all arms with optimal or near-optimal payoffs. The zooming dimension of such problem instance is bounded from above by the covering dimension of $S$. For example, $S$ can be a ``thin'' subtree of an infinitely deep tree.%
    \footnote{Consider an infinitely deep rooted tree and let arms correspond to ends of the tree (i.e., infinite paths away from the root). The distance between two ends decreases exponentially in the height of their least common ancestor (i.e., the deepest vertex belonging to both paths). Suppose there is a subtree in which the branching factor is smaller than elsewhere in the tree. Then we can take $S$ to be the set of ends of this subtree.}

\OMIT{ 
\item Consider a metric space consisting of a high-dimensional region and a low-dimensional region. For concreteness, consider a rooted tree $T$ with two top-level branches $T'$ and $T''$ which are complete infinite $k$-ary trees for $k=2$ and $k=10$, respectively. Assign edge weights in $T$ that are exponentially decreasing in the distance (number of edges) from the root, and let $\mD$ be the resulting shortest-path metric on the leaf set $X$. (Here a \emph{leaf} is as an  infinite path away from the root.) Then $T'$ is the low-dimensional region and $T''$ is a high-dimensional region. If there is a unique optimal arm that lies in $T'$ then the zooming dimension is bounded from above by the covering dimension of $T'$, whereas the covering dimension of the entire tree $T$ is that of $T''$.
} 

\item Suppose the metric space is $([0,1], \ell_1^{1/d})$, $d\in \N$, and the expected payoff of each arm $x$ is determined by its distance from the best arm:
	$\mu(x) = \max(0,\,\mu^*-\mD(x,x^*))$
    for some number $\mu^*\in (0,1]$ and some arm $x^*$. Then the zooming dimension is $0$, whereas the covering dimension is $d$.

\item Suppose the metric space is $([0,1]^d, \ell_2)$, $d\in \N$, and payoff function $\mu$ is $C^2$-smooth. Assume $\mu$ has a unique maximum $x^*$ and is strongly concave in a neighborhood of $x^*$. Then the zooming dimension is $d/2$, whereas the covering dimension is $d$.

\end{itemize}

It turns out that the analysis of the zooming algorithm does not require the similarity function $\mD$ to satisfy the triangle inequality, and needs only a relaxed version of the Lipschitz condition~\refeq{eq:Lipschitz-condition}.

\begin{theorem}[Informal] \label{thm:intro.zooming-extended}
The upper bound in Theorem~\ref{thm:intro.zooming} holds in a more general  setting where the similarity function $\mD$ does not satisfy the triangle inequality, and Lipschitz condition~\refeq{eq:Lipschitz-condition} is relaxed to hold only if one of the two arms is optimal.
\end{theorem}

In addition to the two theorems above, we apply the zooming algorithm to the following special cases, deriving improved or otherwise non-trivial regret bounds:
(i) the maximal payoff is near $1$,
(ii)    $\mu(x) = 1-f(\mD(x,S))$,
    where $S$ is a ``target set'' that is not revealed to the algorithm,
(iii) the reward from playing each arm $x$ is $\mu(x)$ plus an independent, benignly distributed noise.

In particular, we obtain an improved regret rate if the rewards are deterministic. This corollary is related to the literature on global Lipschitz optimization (e.g., see \citet{Floudas-book99}), and extends this literature by relaxing the Lipschitz assumption as in Theorem~\ref{thm:intro.zooming-extended}.

While our definitions and results so far have been tailored to infinite strategy sets, they can be extended to the finite case as well. We use a more precise, \emph{non-asymptotic} version of the zooming dimension, so that all results on the zooming algorithm are meaningful for both finite and infinite strategy sets.


\xhdr{Per-metric optimality: full characterization.}
We are interested in \emph{per-metric optimal} regret bounds: best possible regret bounds for a given metric space. We prove several theorems, which jointly provide a full characterization of per-metric optimal regret bounds for any given metric space $(X,\mD)$. To state polynomial regret bounds in this characterization, we define a parameter of the metric space called \emph{max-min-covering dimension} ($\MaxMinCOV$). Our characterization is summarized in the table below.

\begin{table}[h]
\begin{center}
\begin{tabular}{l|l|l}
If the metric completion of $(X,\mD)$ is ...  & then regret can be ...    & but not ... \\ \hline \\
finite
                          & $O(\log t)$           & $o(\log t)$ \\
compact and countable    & $\omega(\log t)$      & $O(\log t)$ \\
compact and uncountable  & & \\
~~~~~~~~~~~~
$\MaxMinCOV=0$
    & $\tilde{O}\left( t^\gamma \right)$, $\gamma >\tfrac12$
    & $o(\sqrt{t})$  \\
~~~~~~~~~~~~
$\MaxMinCOV=d\in (0,\infty)$
    & $\tilde{O}\left( t^\gamma \right)$, $\gamma >\tfrac{d+1}{d+2}$
    & $o\left( t^\gamma \right)$, $\gamma <\tfrac{d+1}{d+2}$ \\
~~~~~~~~~~~~
$\MaxMinCOV=\infty$   & $o(t)$  & $O\left( t^\gamma \right)$, $\gamma<1$ \\
non-compact           & $O(t)$  & $o(t)$
\end{tabular}
\end{center}
\caption{Per-metric optimal regret bounds for Lipschitz MAB}
\label{tab:PMO-results}
\end{table}

Table~\ref{tab:PMO-results} should be interpreted as follows. We consider regret bounds with an instance-dependent constant, i.e. those of the form
    $R(t) \leq C_\mI\, f(t)$,
for some function $f:\N\to\R$ and a constant $C_\mI$ that can depend on the problem instance $\mI$; we denote this as $R(t) = O_\mI(f(t))$. Let us say that the \problem on a given metric space is $f(t)$-tractable if there exists an algorithm whose regret satisfies $R(t) = O_\mI(f(t))$. Then \cite{Lai-Robbins-85,bandits-ucb1} show that the problem is $\log(t)$-tractable if the metric space has finitely many points (here the instance-dependent constant $C_\mI$ is essential), and not $f(t)$-tractable for any $f(t)=o(\log t)$. Thus, the first row of Table~\ref{tab:PMO-results} reads $O(\log t)$ and $o(\log t)$, respectively; other rows should be interpreted similarly.

In what follows, we discuss the individual results which comprise the characterization in Table~\ref{tab:PMO-results}.

\xhdr{Per-metric optimality: polynomial regret bounds.}
The definition of the max-min-covering dimension arises naturally as one
tries to extend the lower bound from \cite{Bobby-nips04} to general
metric spaces. The min-covering dimension of a subset $Y\subset X$ is the smallest covering dimension of any non-empty subset $U\subset Y$ which is open in the metric topology of $(Y,\mD)$. Further, the \emph{max-min-covering dimension} of $X$, denoted $\MaxMinCOV(X)$, is the largest min-covering dimension of any subset $Y\subset X$. In a formula:
\begin{align*}
\MaxMinCOV(X) = \sup_{Y\subset X} \left(
\inf_{\text{non-empty $U\subset Y$:\, $U$ is open in $(Y,\mD)$}} \;\COV(U) \right).
\end{align*}

\noindent We find that $\MaxMinCOV$ is precisely the right notion to characterize per-metric optimal regret.

\begin{theorem} \label{thm:intro.pmo}
Consider the \problem on a compact metric space with
    $d = \MaxMinCOV(X)$.
If $\gamma > \tfrac{d+1}{d+2}$ then there exists a bandit algorithm $\A$ such that for every problem instance $\mI$ its regret satisfies
	$R(t) = O_{\mI}(t^\gamma)$ for all $t$.
No such algorithm exists if $d>0$ and $\gamma < \tfrac{d+1}{d+2}$.
\end{theorem}

The fact that the above result allows an instance-dependent constant makes the corresponding lower bound more challenging: one needs to show that for any algorithm there exists a problem instance whose regret is at least $t^\gamma$ \emph{infinitely often}, whereas without an instance-dependent constant it suffices to show this for any one time $t$. The former requires a problem instance with infinitely many arms, whereas the latter can be accomplished via a simple problem instance with finitely many arms, and in fact is already done in Theorem~\ref{thm:intro.1d}.

\OMIT{ 
For metric spaces which are highly homogeneous (in the sense
that any two \eps-balls are isometric to one another) the theorem follows easily from a refinement of the techniques introduced in \cite{Bobby-nips04};
in particular, the upper bound can be achieved using a
generalization of the na\"{i}ve algorithm described earlier.
} 

In general $\MaxMinCOV(X)$ is bounded from above by the covering dimension of $X$. For metric spaces which are highly homogeneous, in the sense
that any two \eps-balls are isometric to one another, the two dimensions are equal, and the upper bound in the theorem can be achieved using a  generalization of the $\NaiveAlg$ 
algorithm described earlier.
The difficulty in Theorem~\ref{thm:intro.pmo} lies in dealing with inhomogeneities in the metric space.
\OMIT{\footnote{To appreciate this issue, it is very instructive to consider a concrete example of a metric space $(X,\mD)$ where
	$\MaxMinCOV(X)$
is strictly less than the covering dimension, and for this specific example design a bandit algorithm whose regret bounds are better than those suggested by the covering dimension. This is further discussed in Section~\ref{sec:pmo}.}}
It is important to treat the problem
at this level of generality, because some of the
most natural applications of the \problem, e.g.
the web advertising problem described earlier, are
based on highly inhomogeneous metric spaces.%
\footnote{For example, in web taxonomies, it is unreasonable to expect different categories at the same level of a topic hierarchy to have roughly the same number of descendants. Thus, in a natural interpretation of taxonomy as a metric space in Lipschitz MAB -- where each subtree is a ball whose radius equals or upper-bounds the maximal difference between expected rewards in the said subtree -- balls may be very different from one another.}

The simplest scenario in which we improve over Theorem~\ref{thm:intro-covDim} involves a point $x\in X$ and a number $\eps>0$ such that cutting out any open neighborhood of $x$ reduces the covering dimension by at least $\eps$. We think of such $x$ as a ``fat point'' in the metric space. This example can be extended to a ``fat region'' $S\subset X$ such that $\COV(S)<\COV(X)$
and cutting out any open superset of $S$ reduces the covering dimension by at least $\eps$. One can show that
$\MaxMinCOV(X)\leq \max\{\COV(S),\COV(X)-\eps\}$.

A ``fat region'' $S$ becomes an obstacle for the zooming algorithm if it contains an optimal arm, in which case the algorithm needs to instantiate too many active arms in the vicinity of $S$. To deal with this, we impose a \emph{quota} on the number of active arms outside $S$. The downside is that the set $X\setminus S$ is insufficiently covered by active arms. However, this downside does not have much impact on performance if an optimal arm lies in $S$. And if $S$ does not contain an optimal arm then the zooming algorithm learns this fact eventually, in the sense that it stops refining the mesh of active arms on some open neighborhood $U$ of $S$. From then on, the algorithm essentially limits itself to $X \setminus U$, which is a comparatively low-dimensional set.

The general algorithm in Theorem~\ref{thm:intro.pmo} combines the above
``quota-limited zooming'' idea with a delicate decomposition of the metric space which gradually ``peels off'' regions with abnormally high covering dimension. In the above example with a ``fat region'' $S$, the decomposition consists of two sets, $X$ and $S$. In general, the decomposition is a decreasing sequence of subsets $X=S_0 \supset S_1 \supset \ldots$ where each $S_i$ is a ``fat region'' with respect to $S_{i-1}$. If the sequence is finite then the algorithm has a separate quota for each $S_i$.

Further, to handle arbitrary metric spaces we allow this sequence to be infinite, and moreover \emph{transfinitely} infinite, i.e. parameterized by ordinal numbers. The algorithm proceeds in phases.
Each phase $i$ begins by ``guessing'' an ordinal $\lambda = \lambda_i$ that represents the algorithm's estimate of the largest index of a set in the transfinite sequence that intersects the set of optimal arms. During a phase, the
algorithm focuses on the set $S_\lambda$ in the sequence, and has a quota on active arms not in $S_\lambda$. In the end of the phase it uses the observed payoffs to compute the next ordinal $\lambda_{i+1}$.
The analysis of the algorithm shows that almost surely, the sequence
of guesses $\lambda_1,\lambda_2,\ldots$ is eventually constant, and that the
eventual value of this sequence is almost surely equal to the largest index
of a set in the transfinite sequence that intersects the set of optimal arms.
The regret bound then follows easily from our analysis of the zooming
algorithm.

\OMIT{We find it intellectually intriguing that transfinite ordinal numbers are essential for this algorithm. We are not aware of any other usage of ordinal numbers in the design and analysis of algorithms.}

For the lower bound, we craft a new dimensionality notion ($\MaxMinCOV$),
which captures the inhomogeneity of a metric space, and connect this notion
with the maximal possible ``strength'' of the transfinite decomposition. Further, we connect $\MaxMinCOV$ with the existence of a certain structure in the metric space
(a \emph{ball-tree}) which supports our lower-bounding example. This relation between the two structures --- $d$-dimensional transfinite decompositions and $d$-dimensional ball-trees --- is a new result on metric topology, and as such it may be of independent interest.

While the lower bound is proved using the notion of Kullback-Leibler divergence (\emph{KL-divergence}), our usage of the KL-divergence technique is encapsulated as a generic theorem statement (Theorem~\ref{thm:LB-technique-MAB}). A similar encapsulation (Theorem~\ref{thm:LB-technique}) is stated and proved for the full-feedback version. These theorems and the corresponding setup may be of independent interest. In particular, Theorem~\ref{thm:LB-technique-MAB} has been used in \cite{contextualMAB-colt11} to encapsulate a version of the KL-divergence argument that underlies a lower bound on regret in a contextual bandit setting.


\OMIT{
\begin{definition}\label{def:tractability}
Consider the \problem on a fixed metric space. A bandit algorithm is \emph{$f(t)$-tractable} if for any problem instance $\mI$ the algorithm's regret is
	$R(t) = O_{\mI}(f(t))$.
The problem is \emph{$f(t)$-tractable} if such an algorithm exists.
\end{definition}

Further, Theorem~\ref{thm:intro.pmo} can be restated:

\begin{theorem*} 
Consider the \problem on a compact metric space with
    $d = \MaxMinCOV$.
Then the problem is $t^\gamma$-tractable if $\gamma > \tfrac{d+1}{d+2}$, and not $t^\gamma$-tractable if $d>0$ and $\gamma < \tfrac{d+1}{d+2}$.
\end{theorem*}
}

\OMIT{ 
Consider the \problem on a fixed metric space. Let us say that the problem is \emph{$f(t)$-tractable}, for some function $f$, if there exists an algorithm whose regret satisfies
	$R(t) \leq C_\mu \, f(t)$ for all $t$,
for some constant $C_\mu$ that can depend on the payoff function $\mu$.
Then \cite{Lai-Robbins-85,bandits-ucb1} show that the problem is not $f(t)$-tractable for any $f(t)=o(\log t)$, and it is $\log(t)$-tractable if the metric space is finite (here the instance-dependent constant $C_\mu$ is essential).
} 

\xhdr{Per-metric optimality: beyond polynomial regret.}
To resolve question (Q3), we show that the apparent gap between logarithmic and polynomial regret is inherent to the \problem.

\begin{theorem}\label{thm:intro-dichotomy-MAB} 
For Lipschitz MAB on any fixed metric space $(X,\mD)$, the following dichotomy holds: either it is $f(t)$-tractable for every $f\in \omega(\log t)$, or it is not $g(t)$-tractable for any $g\in o(\sqrt{t})$.  In fact, the former occurs if and only if the metric completion of $(X,\mD)$ is a compact metric space with countably many points.
\end{theorem}

Thus, we establish the $\log(t)$ vs. $\sqrt{t}$ regret dichotomy, and moreover show that it is determined by some of the most basic set-theoretic and topological properties of the metric space.
For compact metric spaces,
the dichotomy corresponds to the transition from countable to uncountable strategy sets.
This is also surprising; in particular, it was natural to conjecture that if the dichotomy exists and admits a simple characterization, it would correspond to the finite vs. infinite transition.

Given the $\Omega(\log t)$ lower bound in \cite{Lai-Robbins-85}, our upper bound for the \problem in compact, countable metric spaces is nearly the best possible bound for such spaces, modulo the gap between ``$f(t) = \log t$" and ``$\forall f\in\omega(\log t)$". Furthermore, we show that this gap is inevitable for infinite metric spaces:

\begin{theorem}\label{thm:logT}
The \problem on any infinite metric space is not $(\log t)$-tractable.
\end{theorem}

To answer question (Q4), we show that the tractability of the \problem on a complete metric space hinges on the compactness of the metric space.

\begin{theorem}\label{thm:boundary-of-tractability-bandits}
The \problem on a fixed metric space $(X,\mD)$ is $f(t)$-tractable for some $f\in o(t)$ if and only if the metric completion of $(X,\mD)$ is a compact metric space.
\end{theorem}

The main technical contribution in the above theorems is an interplay of online learning and point-set topology, which requires novel algorithmic and lower-bounding techniques. For the $\log(t)$ vs. $\sqrt{t}$ dichotomy result, we identify a simple topological property (existence of a topological well-ordering) which entails the algorithmic result, and another topological property (\emph{perfectness}) which entails the lower bound. The equivalence of the first property
to countability and the second to uncountability (for compact metric spaces) follows from classical theorems of Cantor-Bendixson \citep{Cantor83} and Mazurkiewicz-Sierpinski \citep{MazSier}.

\subsection{Our contributions: the \FFproblem}
\label{sec:intro-experts}

We turn our attention to the \emph{\FFproblem}: the full-feedback version of the \problem. Formally, a problem instance is specified by a triple $(X,\mD,\prob)$, where $(X,\mD)$ is a metric space and $\prob$ is a  probability measure with universe $[0,1]^X$, the set of all functions from $X$ to $[0,1]$, such that the expected payoff function
    $\mu: x \mapsto \E_{f \in \prob}[f(x)]$
is a Lipschitz function on $(X,\mD)$. The metric structure of $(X,\mD)$ is known to the algorithm, the measure $\prob$ is not. We will refer to $\prob$ as the problem instance when the metric space $(X,\mD)$ is clear from the context.

In each round $t$ the algorithm picks a strategy $x_t\in X$, then
the environment chooses an independent sample $f_t: X\to [0,1]$ distributed according to the measure $\prob$. The algorithm receives payoff $f_t(x_t)$, and also observes the entire payoff function $f_t$. More formally, the algorithm can query the value of $f_t$ at an arbitrary finite number of points. Some of our upper bounds are for a (very) restricted version, called \emph{double feedback}, where in each round the algorithm picks two arms $(x,y)$, receives the payoff for $x$ and also observes the payoff for $y$. By abuse of notation, we will treat the bandit setting as a special case of the experts setting.

Note that the payoffs for different arms in a given round are not necessarily independent. This is essential because for any limit point $x$ in the metric space one could use many independent samples from the vicinity of $x$ to learn the expected payoff at $x$ in a single round.

\xhdr{Regret dichotomies.}
We show that the \FFproblem exhibits a regret dichotomy similar to the one in Theorem~\ref{thm:intro-dichotomy-MAB}. Since the optimal regret for a finite strategy set is \emph{constant} \citep{sleeping-colt08}, the dichotomy is between $O(1)$ and $\sqrt{t}$ regret.

\begin{theorem}\label{thm:main-experts}
The \FFproblem on metric space $(X,\mD)$ is either $1$-tractable, even with double feedback, or it is not $g(t)$-tractable for any $g\in o(\sqrt{t})$, even with full feedback. The former case occurs if and only if the completion of $X$ is a compact metric space with countably many points.
\end{theorem}

Theorem~\ref{thm:main-experts} and its bandit counterpart (Theorem~\ref{thm:intro-dichotomy-MAB}) are proved jointly, using essentially the same ideas. In both theorems, the regret dichotomy corresponds to the transition from countable to uncountable strategy set (assuming the metric space is compact and complete). Note that the upper bound in Theorem~\ref{thm:main-experts} only assumes double feedback, whereas the lower bound is for the unrestricted full feedback.

Next, we investigate for which metric spaces the \FFproblem is $o(t)$-tractable. We extend Theorem~\ref{thm:boundary-of-tractability-bandits} for the \problem to another regret dichotomy where the upper bound is for the bandit setting, whereas the lower bound is for full feedback. 

\begin{theorem}\label{thm:boundary-of-tractability}
The \FFproblem on metric space $(X,\mD)$ is either $f(t)$-tractable for some $f\in o(t)$, even in the bandit setting, or it is not $g(t)$-tractable for any $g\in o(t)$, even with full feedback. The former occurs if and only if the completion of $X$ is a compact metric space.
\end{theorem}

\xhdr{Polynomial regret in (very) high dimension.}
In view of the $\sqrt{t}$ lower bound from Theorems~\ref{thm:main-experts}, we are interested in matching upper bounds. \citet{Anupam-experts07} observed that such bounds hold for every metric space $(X,\mD)$ of finite covering dimension: namely, the \FFproblem on $(X,\mD)$ is $\sqrt{t}$-tractable. Therefore it is natural to ask whether there exist metric spaces of \emph{infinite} covering dimension with polynomial regret.

We settle this question by proving a characterization with nearly matching upper and lower bounds in terms of a novel dimensionality notion tailored to the experts problem. We define the \emph{log-covering dimension} of $(X,\mD)$ as the smallest number $d\geq 0$ such that $X$ can be covered by
    $O\left(2^{r^{-d}} \right)$
sets of diameter $r$ for all $r>0$.
More formally:
\begin{align}\label{eq:LCD}
    \LCD(X) = \limsup_{r \to 0}\, \frac{\log\log N_r(X)}{\log (1/r)}.
\end{align}
where $N_r(X)$ is the minimal size (cardinality) of a $r$-covering of $X$, i.e. the smallest number of sets of diameter at most $r$ sufficient to cover $X$. Note that the number of sets allowed by this definition is exponentially larger than the one allowed by the covering dimension.

To give an example of a metric space with a non-trivial log-covering dimension, let us consider a \emph{uniform tree} -- a rooted tree in which all nodes at the same level have the same number of children. An \emph{\eps-uniform tree metric} is a metric on the ends of an infinitely deep uniform tree, in which the distance between two ends is $\eps^{-i}$, where $i$ is the level of their least common ancestor. It is easy to see that an \eps-uniform tree metric such that the branching factor at each level $i$ is
    $ \exp(\eps^{-id} (2^d - 1))$
has log-covering dimension $d$.

For another example, consider the set of all probability measures over $X=[0,1]^d$ under the Wasserstein $W_1$ metric, a.k.a. the Earthmover distance.
\footnote{\label{fn:earthmover}
The Wasserstein $W_1$ metric is one of the standard ways to define a distance on probability measures. In particular, it is widely used in Computer Science literature to compare discrete distributions, e.g. in the context of image retrieval \citep{earthmover-00}.}
We show that the log-covering dimension of this metric space is equal to the covering dimension of $(X,\mD)$. In fact, this example extends to any metric space $X$ of finite diameter and covering dimension $d$; see Appendix~\ref{app:earthmover} for the details.

\begin{theorem}\label{thm:intro-LCD}
Let $(X,\mD)$ be a metric space of log-covering dimension $d$. Then the \FFproblem is
 $(t^\gamma)$-tractable for any $\gamma> \tfrac{d+1}{d+2}$.
\end{theorem}

The algorithm in Theorem~\ref{thm:intro-LCD} is a version of \NaiveAlg. The same algorithm enjoys a better regret bound if each function $f\in \mathtt{support}(\mathbb{P})$ is itself a Lipschitz function on $(X,\mD)$. We term this special case the \emph{\ULproblem}.

\begin{theorem}\label{thm:intro-LCD-uniform}
Let $(X,\mD)$ be a metric space of log-covering dimension $d$. Then the \ULproblem is
 $(t^\gamma)$-tractable for any $\gamma> \tfrac{d-1}{d}$.
\end{theorem}

The analysis is much more sophisticated compared to Theorem~\ref{thm:intro-LCD}, using a chaining technique from empirical process theory (see \citet{Talagrand-book05} for background).


\xhdr{Per-metric optimal regret bounds.}
We find that the log-covering dimension is not the right notion to characterize optimal regret for arbitrary metric spaces. Instead, we define the \emph{max-min-log-covering dimension} ($\MaxMinLCD$): essentially, we take the definition of $\MaxMinCOV$ and replace covering dimension with log-covering dimension. \begin{align}\label{eq:MaxMinLCD}
\MaxMinLCD(X) = \textstyle{\sup_{Y\subset X}} \;
    \inf\{ \,\LCD(Z):\; \text{open non-empty $Z\subset Y$} \}.
\end{align}
\noindent Note that in general $\MaxMinLCD(X) \leq \LCD(X)$. Equality holds for ``homogeneous" metric spaces such as \eps-uniform tree metrics.
We derive the regret characterization in terms $\MaxMinLCD$; the characterization is tight for the \ULproblem.

\begin{theorem}\label{thm:intro-MaxMinLCD}
Let $(X,\mD)$ be an uncountable metric space and $d = \MaxMinLCD\geq 0$. Then:
\begin{OneLiners}
\item[(a)] the \FFproblem is $(t^\gamma)$-tractable for any $\gamma> \tfrac{d+1}{d+2}$,

\item[(b)] the \ULproblem is $(t^\gamma)$-tractable for any
    $\gamma> \max(\tfrac{d-1}{d}, \tfrac12)$,

\item[(c)] the \ULproblem is not $(t^\gamma)$-tractable for any $\gamma < \max(\tfrac{d-1}{d}, \tfrac12)$.
\end{OneLiners}
\end{theorem}

The algorithms in parts (a) and (b) use a generalization of the transfinite decomposition from the bandit per-metric optimal algorithm (Theorem~\ref{thm:intro.pmo}). The lower bound in part (c) builds on the lower-bounding technique for the $\sqrt{t}$ lower bound on uncountable metric spaces. 

Our results for Lipschitz experts amount to a nearly complete characterization of per-metric optimal regret bounds, analogous to that in Table~\ref{tab:PMO-results} on page~\pageref{tab:PMO-results}.
This characterization is summarized in the table below.
(The characterization falls short of being complete because the
upper and lower bounds for finite $\MaxMinLCD = d \in [0,\infty)$
do not match.)

\begin{table}[h]
\label{tab:experts-results}
\caption{Per-metric optimal bounds for Lipschitz experts}
\begin{tabular}{l|l|l}
If the completion of $(X,\mD)$ is ...  & then regret can be ...    & but not ... \\ \hline\\
compact and countable    & $O(1)$      & ~~--- \\
compact and uncountable  & & \\
~~~~~~~~~~~~ finite covering dimension
    & $\tilde{O}\left( \sqrt{t} \,\right)$
    & $o(\sqrt{t})$  \\
~~~~~~~~~~~~
$\MaxMinLCD=d\in [0,\infty)$
    & $\tilde{O}\left( t^\gamma \right)$, $\gamma >\tfrac{d+1}{d+2}$
    & $o\left( t^\gamma \right)$, $\gamma=\tfrac12$ or $\gamma <\tfrac{d-1}{d}$  \\
~~~~~~~~~~~~
$\MaxMinLCD=\infty$   & $o(t)$  & $O\left( t^\gamma \right)$, $\gamma<1$ \\
non-compact           & $O(t)$  & $o(t)$
\end{tabular}
\end{table}

\newcommand{\DpthOracle}{{\mathtt{Length}}}
\newcommand{\DCovOracle}{{\mbox{$\mD$-$\mathtt{Cov}$}}}

\subsection{Discussion}
\label{subsec:intro-access}

\xhdr{Accessing the metric space.}
In stating the theorems above, we have been imprecise about specifying the model of computation.  In particular, we have ignored the thorny issue of how to provide an algorithm with an input describing 
a metric space which may have an infinite number of points. The simplest way to interpret our theorems is to ignore implementation details and interpret an ``algorithm'' to mean an abstract decision rule, i.e. a (possibly randomized) Borel-measurable function mapping the history of past observations to an arm $x \in X$ which
is played in the current period. All of our theorems are valid under this interpretation, but they can also be made into precise algorithmic results provided that the algorithm is given appropriate oracle access to the metric space.

The zooming algorithm requires only a \emph{covering oracle} which takes a finite collection of open balls and either declares that they cover $X$ or outputs an uncovered point. The algorithm poses only one oracle query in each round $t$, for a collection of at most $t$ balls. (For infinite metric spaces of interest that admit a finite description, e.g.~rational convex polytopes in Euclidean space, it is generally easy to implement a covering oracle given a description of the metric space.)
The per-metric optimal algorithm in Theorem~\ref{thm:intro.pmo} uses more complicated oracles, and we defer the definition of these oracles to Section~\ref{sec:pmo}.

The $\omega(\log t)$-regret algorithms for countably infinite metric spaces (Theorems~\ref{thm:intro-dichotomy-MAB} and~\ref{thm:main-experts}) require an oracle which represents the well-ordering of the metric space. We also provide an extension for compact metric spaces with a finite number of limit points for which a more intuitive oracle access suffices. In fact, this extension holds for a much wider family of metric spaces: those with a finite \emph{Cantor-Bendixson rank}, a classic notion from point-set topology.

\xhdr{Further directions.}
While general, our model is idealized in several ways. Numerical similarity information, such as the distances and the Lipschitz constant, may be difficult to obtain in practice. The notion of similarity is ``worst-case", so that the distances may need to be large in order to accommodate a few outliers. The reward distribution does not change over time. These issues gave rise to a line of follow-up work, detailed in Section~\ref{sec:related-followup}.


\xhdr{Map of the paper.}
We discuss related work in Section~\ref{sec:related}. In particular, a considerable amount of \emph{follow-up work} is surveyed in Section~\ref{sec:related-followup}. Preliminaries are presented in Section~\ref{sec:prelims}, including sufficient background on metric topology and dimensionality notions, and the proof of the initial observation 
(Theorem~\ref{thm:intro-covDim}).

In the rest of the paper we present our technical results. Section~\ref{sec:adaptive-exploration} is on Lipschitz bandits with benign payoff functions; it presents the zooming algorithm and extensions thereof. Section~\ref{sec:pmo} is on the per-metric optimal algorithms for Lipschitz bandits, focusing on polynomial regret.
The next two sections concern both Lipschitz bandits and Lipschitz experts:
Section~\ref{sec:dichotomies} is on the dichotomy between (sub)logarithmic and $\sqrt{t}$ regret, and Section~\ref{sec:boundary-body} studies for which metric spaces the Lipschitz bandits/experts problem is $o(t)$-tractable. Section~\ref{sec:FFproblem} is on the polynomial-regret algorithms for Lipschitz experts. We conclude with directions for further work in Section~\ref{sec:conclusions}.

To preserve the flow of the paper, some material is deferred to appendices. In Appendix~\ref{sec:KL-divergence} we present sufficient background on Kullback-Leibler divergence (KL-divergence) and the technical proofs which use the KL-divergence technique. In Appendix~\ref{sec:reduction} we reduce the Lipschitz bandits/experts problem to that on complete metric spaces. In Appendix~\ref{sec:topological} we present a self-contained proof of a theorem from general topology, implicit in \cite{Cantor83,MazSier}, which ties together the upper and lower bounds of the regret dichotomy result. Finally, in Appendix~\ref{app:earthmover} we flesh out the Earthmover distance example from Section~\ref{sec:intro-experts}.

\section{Related and follow-up work}
\label{sec:related}

\xhdr{Multi-armed bandits.}
MAB problems have a long history; a thorough survey is beyond the scope of this paper. For background, a reader can refer to a book \cite{CesaBL-book} and a recent survey \cite{Bubeck-survey12} on regret-minimizing bandits.
The Bayesian perspective (less relevant to the present paper) can be found in books and surveys
\citep{Bergemann-survey06,Gittins-book11}.
On a very high level, there is a crucial distinction between regret-minimizing formulations and Bayesian/MDP formulations. Among regret-minimizing formulations, an important distinction is between stochastic payoffs \citep{Lai-Robbins-85,bandits-ucb1} and adversarial payoffs \citep{bandits-exp3}.

This paper is on regret minimization with stochastic payoffs. The basic setting here is $k<\infty$ arms with no additional structure. Then the optimal regret is $R(t) = O(\sqrt{kt})$, and $R(t) = O_\mI(\log t)$ with an instance-dependent constant \citep{Lai-Robbins-85,bandits-exp3,bandits-ucb1}.
Note that the distinction between regret rates with and without instance-dependent constants is inherent even in this basic bandit setting.
The $\UCB$ algorithm \citep{bandits-ucb1} achieves the $O_\mI(\log t)$ bound and simultaneously matches the $O(\sqrt{kt})$ bound up to a logarithmic factor.

\OMIT{Several recent papers \cite{Bubeck-colt09,Audibert-TCS09,Honda-colt10,Auer-UCB-10,Maillard-colt11,Garivier-colt11} improve over $\UCB$, obtaining algorithms with regret bounds that are even closer to the lower bound.}

Our zooming algorithm relies on the ``UCB index'' technique from \citet{bandits-ucb1}. This is a simple but very powerful idea: arms are chosen according to a numerical score, called \emph{index}, which is defined as an upper confidence bound (UCB) on the expected payoff of a given arm. Thus, the UCB index can be represented as a sample average plus a confidence term, which represent, respectively, exploitation and exploration, so that the sum represents a balance between the two. Several papers \citep{Bubeck-colt09,Audibert-TCS09,Honda-colt10,Auer-UCB-10,Maillard-colt11,Garivier-colt11} designed improved versions of the UCB index for the $k$-armed MAB problem with stochastic payoffs, achieving regret bounds which are even closer to the lower bound. Moreover, the UCB index idea and various extensions thereof have been tremendously useful in many other settings with exploration-exploitation tradeoff, e.g. \citep{Auer-focs00,DynamicMAB-colt08,sleeping-colt08,Munos-nips08,Bubeck-colt10,contextualMAB-colt11,Csaba-nips11,DynPricing-ec12}.
It is worth noting that the zooming algorithm, as originally published in \citep{LipschitzMAB-stoc08}, was one of the first results in this line of work.

Many papers enriched the basic MAB setting by assuming some structure on arms, typically in order to handle settings where the number of arms is very large or infinite. Most relevant to this paper is the work on \emph{continuum-armed bandits} \citep{Agrawal-bandits-95,Bobby-nips04,AuerOS/07}, a special case of Lipschitz MAB where the metric space is $([0,1], \ell_1)$. A closely related model posits that arms correspond to leaves on a tree, but no metric space is revealed to the algorithm \citep{Kocsis-ecml06,yahoo-bandits07,Munos-uai07,ImplicitMAB-nips11}.
Another commonly assumed structure is linear or convex payoffs,  e.g.~\citep{Bobby-stoc04,FlaxmanKM-soda05,DaniHK-nips07,AbernethyHR-colt08,Hazan-soda09}.
Linear/convex payoffs is a much stronger assumption than similarity, essentially because it allows to make strong inferences about far-away arms. Accordingly, it admits much stronger regret bounds, such as $\tilde{O}(d\sqrt{t})$ for arms in $\R^d$. Other structures in the literature include infinitely many i.i.d. arms \citep{Berry-bandits-97,Munos-nips08}, Gaussian Process Bandits \citep{GPbandits-icml10,GPbandits-nips11,GPbandits-icml12} and Functional bandits \citep{functionalMAB-colt11}; Gaussian Process MAB and Functional MAB are discussed in more detail in Section~\ref{sec:related-followup}.

Closely related to continuum-armed bandits is the model of (regret-minimizing) \emph{dynamic pricing} with unlimited supply \citep{Blum03,KleinbergL03}. In this model, an algorithm is a seller with unlimited supply of identical items, such as a digital good (a movie, a song, or a program) that can be replicated at no cost. Customers arrive sequentially, and to each customer the algorithm offers one item at a non-negotiable price. Here prices correspond to arms, and accordingly the ``arms'' have a real-valued structure. Due to the discontinuous nature of demand (a customer who values the item at $v$ will pay a price of $v-\eps$ but will pay nothing if offered a price of $v+\eps$) dynamic pricing is not a special case of Lipschitz MAB, but there is a close relationship between the techniques that have been used to solve both problems. Moreover, when the distribution of customer values has bounded support and bounded probability density, the expected revenue is a Lipschitz function of the offered price, so regret-minimizing dynamic pricing in this case reduces to the \problem.%
\footnote{Some of the work on dynamic pricing, e.g. \citep{BZ09,Wang-OR14}, makes the Lipschitz assumption directly.}
One can also consider selling $d>1$ products, offering a different price for each. When the expected revenue is a Lipschitz function of the offered price vector, this is a special case of Lipschitz MAB with arms in $\R^d$.


Interestingly, the dichotomy between (poly)logarithmic and $\sqrt{t}$ regret has appeared in four different MAB settings: Theorem~\ref{thm:intro-dichotomy-MAB} in this paper, $k$-armed bandits with stochastic payoffs (as mentioned above),
bandits with linear payoffs \citep{DaniHK-colt08},
and an extension of MAB to pay-per-click auctions \citep{MechMAB-ec09,DevanurK09,Transform-ec10-jacm}.
These four dichotomy results have no obvious technical connection.

\OMIT{ 
Most similar to Theorem~\ref{thm:intro-dichotomy-MAB} is the result on Linear MAB in \cite{DaniHK-colt08}, where the polylogarithmic tractability is determined by a certain topological property.%
\footnote{However, the topological property in \cite{DaniHK-colt08} -- a uniformly positive gap in expected payoff between the optimal point and any other extremal point -- is very different from the one in Theorem~\ref{thm:intro-dichotomy-MAB}. Note that this property depends the payoff function, rather than the space of arms alone. \BKnote{1}{I'm not sure about this discussion. First of all, it is a bit of a misnomer to call this a topological property; topologically, extreme points are no different from other points on the boundary of a convex set. Secondly, if one allows instance-dependent constants as in our paper, then Dani et al don't require a uniformly positive gap between the optimal point and any other extreme point, they only require a positive gap. So I don't think these two dichotomies are as different as you are making them out to be.}}
} 

\xhdr{Metric spaces and dimensionality notions.}
Algorithmic problems on metric spaces have a long history in many different domains. These domains include: constructing space-efficient and/or algorithmically tractable representations such as metric embeddings, distance labels, or distance oracles; problems with costs where costs have a metric structure, e.g. facility location and traveling salesman; offline and online optimization on a metric space; finding hidden structure in a metric space (classification and clustering).

Covering dimension is closely related to several other notions of dimensionality of a metric space, such as Haussdorff dimension, capacity dimension, box-counting dimension, and Minkowski-Bouligand Dimension. All these notions are used to characterize the covering properties of a metric space in fractal geometry; discussing fine distinctions between them is beyond our scope. A reader can refer to \cite{Schroeder-1991} for background.

Covering numbers and covering dimension have been widely used in Machine Learning to characterize the complexity of the \emph{hypothesis space}: a space of functions over $X$, the domain for which the learner needs to predict or classify, under functional $\ell_2$ norm and some distribution over $X$. This is different from the way covering numbers and similar notions are used in the context of the \problem, and we are not aware of a clear technical connection.%
\footnote{In the \problem, one is interested in the family $\F$ of all Lipschitz-continuous functions on $(X,\mD)$, and therefore one could consider the covering numbers for $\F$, or use any other standard notions such as VC-dimension or fat-shattering dimension. However, we have not been able to reach useful results with this approach.}
Non-metric notions to characterize the complexity of function classes include VC-dimension, fat-shattering dimension, and Rademacher averages; see \cite{Shai-MLbook-2014} for background.

\OMIT{ 
\BKnote{1}{I think the following work is not related to our paper in any
meaningful sense, and I propose eliminating this paragraph.}
Covering dimension has been used in \cite{Abrahao-imc08} to empirically characterize the ``intrinsic dimensionality'' of the Internet delay space (i.e., the matrix of pairwise round-trip times). While prior work \cite{Meridian-sigcomm05,Viennot-infocom08} considered doubling dimension for the same purpose, \cite{Abrahao-imc08} found that covering dimension is much more suitable.
} 

Various notions of dimensionality of metric spaces have been studied in the theoretical computer science literature, with a goal to arrive at (more) algorithmically tractable problem instances.  The most popular notions have
been the ball-growth dimension, e.g. \citep{Karger-stoc02,Hildrum-spaa04,Abr05-SPAA,Slivkins-podc07}, and the doubling dimension, e.g. \citep{Gup03,Tal04,Slivkins-focs04,Slivkins-podc05,Men05,Cha05}. These notions have been useful in many different problems, including metric embeddings, other space-efficient representations such as distance labels and sparse spanners, network primitives such as routing schemes and distributed hash tables, and approximation algorithms for various optimization problems such as traveling salesman, $k$-median, and facility location.

\subsection*{Concurrent and independent work}
\label{sec:related-concurrent}

\cite{xbandits-nips08-conf,xbandits-nips08} obtain results similar to Theorem~\ref{thm:intro.zooming} and Theorem~\ref{thm:intro.zooming-extended}. They use similar, but technically different notions of instance-dependent metric dimension and relaxed Lipschitzness. They also obtain stronger regret bounds for some special cases; these extensions are similar in spirit to the extended analysis of the zooming algorithm in this paper (but technically different).  Their results use a different algorithm and the proof techniques appear different.

While the publication of our conference version \citep{LipschitzMAB-stoc08} predated the submission of theirs \citep{xbandits-nips08-conf},%
\footnote{\citet{xbandits-nips08-conf} acknowledge Theorem~\ref{thm:intro.zooming} as prior work. It appears that the authors of \citep{xbandits-nips08-conf} have not been aware of other results in \citep{LipschitzMAB-stoc08} at the time (which were only briefly mentioned in the conference version, and fleshed out in the full version \citep{LipschitzMAB-arxiv}).}
 we believe the latter is concurrent and independent work.

\OMIT{
Theorems~\ref{thm:intro.zooming} and~\ref{thm:intro.zooming-extended} have appeared in \cite{LipschitzMAB-stoc08,LipschitzMAB-arxiv}, which predates \cite{xbandits-nips08}. However, \cite{xbandits-nips08} is a follow-up work with respect to Theorem~\ref{thm:intro.zooming} and independent work with respect to Theorem~\ref{thm:intro.zooming-extended} and other extensions of the zooming algorithm.

Theorem~\ref{thm:intro.zooming} is one of the main results in the conference version of \cite{LipschitzMAB-stoc08}, published in ACM STOC in May 2008. Theorem~\ref{thm:intro.zooming-extended} and other extensions are briefly mentioned in the conference version of \cite{LipschitzMAB-stoc08} and fleshed out in the full version, which has been available from the authors' webpages since May 2008 and published as a technical report on {\tt arxiv.org} in September 2008.

\cite{xbandits-nips08} was submitted to NIPS in June 2008 and appeared in NIPS in December 2008. To the best of our understanding, at the time of the NIPS submission the authors of \cite{xbandits-nips08} were aware of Theorem~\ref{thm:intro.zooming} but not of Theorem~\ref{thm:intro.zooming-extended} and other extensions.
} 

\subsection*{Follow-up work}
\label{sec:related-followup}

Since the conference publication of \cite{LipschitzMAB-stoc08} there has been a considerable amount of follow-up work on Lipschitz MAB and various extensions thereof.

\xhdr{Lower bounds.}
While our lower bound for ``benign" problem instances (in Theorem~\ref{thm:intro.zooming}) comes from the worst-case scenario when the zooming dimension equals the covering dimension, \citet{contextualMAB-colt11} and \citet{Combes-colt14} provide more refined, \emph{instance-dependent} lower bounds. \citet{contextualMAB-colt11} proves that the upper bound in Theorem~\ref{thm:zooming-covering-num} is tight, up to $O(\log^2 t)$ factors, \emph{for every value of the said upper bound}.%
\footnote{In fact, this lower bound extends to the contextual bandit setting.}
\citet{Combes-colt14} focus on regret bounds of the form $C \cdot \log(t)+O(1)$, where $C$ depends on the problem instance, but not on time. They derive a lower bound on the $C$, and provide an algorithm which comes arbitrarily close to this lower bound.


\xhdr{Contextual Lipschitz MAB and applications.}
\cite{Pal-Bandits-aistats10} and \cite{contextualMAB-colt11}, simultaneous and independent w.r.t. one another,%
\footnote{The initial version of \cite{contextualMAB-colt11} has appeared on {\tt arxiv.org} in 2009. It contained the main algorithm, the same as in the final version, but only derived results for the covering dimension. The conference version from \emph{COLT 2011} contained essentially the same results as in the final journal version from 2014.}
extend Lipschitz MAB to
the contextual bandit setting, where in each round the algorithm receives a context (``hint'') $h$ and picks an arm $x$, and the expected payoff is a function of both $h$ and $x$. The motivational examples include placing ads on webpages (webpages and/or users are contexts, ads are arms), serving documents to users (users are contexts, documents are arms), and offering prices to customers (customers are contexts, prices are arms). The similarity information is expressed as a metric on contexts and a metric on arms, with the corresponding two Lipschitz conditions.
\cite{Pal-Bandits-aistats10} consider this setting and extend \NaiveAlg to obtain regret bounds in terms of the covering dimensions of the two metric spaces. \cite{contextualMAB-colt11} extends the  zooming algorithm to the contextual setting and obtains improved regret bounds in terms of a suitable ``contextual" version of the zooming dimension.

The ``contextual zooming algorithm" from \cite{contextualMAB-colt11} works in a more general setting where similarity information is represented as a metric space on the set of ``allowed'' context-arm pairs, and the expected payoff function is Lipschitz with respect to this metric space. This is a very versatile setting:
it can also encode \emph{sleeping bandits} \citep{Blum-sleeping97,Freund-colt97,sleeping-colt08} (in each round, some arms are ``asleep'', i.e. not available) and slowly changing payoffs \citep{DynamicMAB-colt08} (here in each round $t$ the context is $t$ itself, and the metric on contexts expresses the constraint how fast the expected payoffs can change). This setting showcases the full power of the adaptive refinement technique which underlies the zooming algorithm.

Further, \cite{ZoomingRBA-icml10} use the zooming algorithm from this paper and its contextual version from \cite{contextualMAB-colt11} in the context of \emph{ranked bandits} \citep{RBA-icml08}. Here in each round a bandit algorithm chooses an ordered list of $k$ documents (from a much larger pool of available documents) and presents it to a user who scrolls the list top-down and clicks on the first document that she finds  relevant. The user may leave after the first click; the goal is to minimize the number of users with no clicks. The contribution of \citep{ZoomingRBA-icml10} is to combine ranked bandits with Lipschitz MAB; among other things, this requires a significantly extended model: if two documents are close in the metric space,
their click probabilities are similar even conditional on the event that some other documents are not clicked by the current user.

\xhdr{Partial similarity information.}
A number of papers tackle the issue that the numerical similarity information required for the \problem may be difficult to obtain in practice. These papers make various assumptions on what is and is not revealed to the algorithm, with a general goal to do (almost) as well as if the full metric space were known. \cite{Bubeck-alt11} study a version with strategy set $[0,1]^d$ and Lipschitz constant that is not revealed, and match the optimal regret rate for algorithms that know the Lipschitz constant. \cite{Minsker-colt13} considers the same strategy set and distance function of the form $\|x-y\|_\infty^\beta$, where the smoothness parameter $\beta\in(0,1]$ is not known.
\citep{ImplicitMAB-nips11,Munos-nips11,Bull-bandits14,Valko-icml13} study a version in which the algorithm only inputs a ``taxonomy" on arms (i.e., a tree whose leaves are arms), whereas the numerical similarity information is not revealed at all. This version features a \emph{second} exploration-exploitation tradeoff: the tradeoff between learning more about the numerical similarity information (or some relevant portions thereof), and exploiting this knowledge to run a Lipschitz MAB algorithm.

The latter line of work proceeds as follows. \citet{ImplicitMAB-nips11} considers the metric space implicitly defined by the taxonomy, where the distance between any two arms is the maximal difference in expected rewards in the least common subtree. He puts forward an extension of the zooming algorithm which adaptively reconstructs the implicit metric space, and (under some additional assumptions) essentially matches the performance of the zooming algorithm on the same metric space. \cite{Munos-nips11} and \cite{Valko-icml13} allow a more general relation between the implicit metric space and the taxonomy, and moreover relax the Lipschitz condition to only hold w.r.t. the maximum (as in Theorem~\ref{thm:zooming-dim-relaxed}). \cite{Munos-nips11} focuses on deterministic rewards, and essentially matches the regret bound in Corollary~\ref{cor:confRad-deterministic}, whereas \cite{Valko-icml13} study the general IID case. Finally, \cite{Bull-bandits14} considers a somewhat more general setting with multiple taxonomies on arms (or with arms embedded in $[0,1]^d$, where the embedding is then used to define the taxonomies). The paper extends and refines the algorithm from \cite{ImplicitMAB-nips11}, and carefully traces out the conditions under which one can achieve $\tilde{O}(\sqrt{T})$ regret. \citet{Munos-trends14} surveys some of this work, with emphasis on the techniques from \citep{xbandits-nips08,Munos-nips11,Valko-icml13}.

As a stepping stone to the result mentioned above, \cite{Munos-nips11} considers Lipschitz MAB with deterministic rewards and essentially matches our result for this setting (Corollary~\ref{cor:confRad-deterministic}).
\footnote{While our original publication \citep{LipschitzMAB-stoc08,LipschitzMAB-arxiv}  predates \citep{Munos-nips11}, the latter is independent work to the best of our understanding.}

\xhdr{Beyond IID rewards.}
Several papers \citep{Azar-icml14,Munos-ecml10,contextualMAB-colt11} consider Lipschitz bandits/experts with non-IID rewards.\footnote{The first result in this direction appeared in \citet{Bobby-nips04}. He considers Lipschitz MAB with adversarial rewards, and proposes a version of \NaiveAlg where an adversarial bandit algorithm is used instead of \UCB. This algorithm achieves the same worst-case regret as \NaiveAlg on IID rewards.}
\citet{Azar-icml14} consider a version of Lipschitz MAB in which the IID condition is replaced by more sophisticated ergodicity and mixing assumptions, and essentially recover the performance of the zooming algorithm. \cite{Munos-ecml10} consider Lipschitz experts in a Euclidean space $(\R^d, \ell_2)$ of constant dimension $d$. Assuming the Lipschitz condition on realized payoffs (rather than expected payoffs), they achieve a surprisingly strong regret of $O(\sqrt{t})$.
\cite{contextualMAB-colt11} considers contextual bandits with Lipschitz condition on expected payoffs, and provides a ``meta-algorithm'' which uses an off-the-shelf bandit algorithm such as \EXP \citep{bandits-exp3} as a subroutine and adaptively refines the space of contexts. Also, as discussed above, the contextual zoooming algorithm from \cite{contextualMAB-colt11} can handle Lipschitz MAB with slowly changing rewards.

\xhdr{Other structural models of MAB.}
One drawback of Lipschitz MAB as a model is that $\mD(x,y)$ only gives a worst-case notion of similarity between arms $x$ and $y$: a hard upper bound on $|\mu(x)-\mu(y)|$ rather than a typical or expected upper bound. In particular, the distances may need to be very large in order to accommodate a few outliers, which would make $\mD$ less informative elsewhere.%
\footnote{This concern is partially addressed by Theorem~\ref{thm:intro.zooming-extended}.}
With this criticism in mind, \cite{GPbandits-icml10} define a probabilistic model, called \emph{Gaussian Processes Bandits},  where the expected payoff function is distributed according to a suitable Gaussian Process on $X$, thus ensuring a notion of ``probabilistic smoothness'' with respect to $X$. Further work in this model includes \cite{GPbandits-nips11} and \cite{GPbandits-icml12}.

Given the work on Lipschitz MAB (and other ``structured'' bandit models such as linear payoffs) it is tempting to consider MAB with \emph{arbitrary} known structure on payoff functions. \cite{functionalMAB-colt11} initiate this direction: in their model, the structure is explicitly represented as the collection of all possible payoff functions. However, their results do not subsume any prior work on Lipschitz MAB or MAB with linear or convex payoffs.

\xhdr{Further applications of our techniques.}
\cite{RepeatedPA-ec14} design a version of the zooming algorithm in the context of crowdsourcing markets. Here the algorithm is an employer who offers a quality-contingent contract to each arriving worker, and adjusts the contract over time. This is an MAB problem in which arms are contracts (essentially, vectors of prices), and a single round is modeled as a standard ``principal-agent model" from contract theory. \cite{RepeatedPA-ec14} do not assume a Lipschitz condition, or any other explicit guarantee on similarity between arms. Instead, their algorithm estimates the similarity information on the fly, taking advantage of the structure provided by the principal-agent model.%
\footnote{\citep{ImplicitMAB-nips11,Bull-bandits14,Valko-icml13} estimate the ``hidden" similarity information for a general Lipschitz MAB setting (using some additional assumptions), but \cite{RepeatedPA-ec14} uses a different, problem-specific approach which side-steps some of the assumptions.}

On a final note, one of our minor results -- the improved confidence radius from Section~\ref{subsec:zooming-max} -- may be of independent interest. In particular, this result is essential for some of the main results in \citep{DynPricing-ec12,BwK-focs13,AgrawalDevanur-ec14,Shipra-ec16}, in the context of dynamic pricing and other MAB problems with global supply/budget constraints.


\section{Preliminaries}
\label{sec:prelims}

This section contains various definitions which make the paper essentially self-contained (the only exception being \emph{ordinal numbers}). In particular, the paper uses notions from General Topology which are typically covered in any introductory text or course on the subject.

\subsection*{Problem formulation and notation}

In the \problem, the problem instance is a triple $(X,\mD,\mu)$, where $(X,\mD)$ is a metric space and $\mu \,:\, X \rightarrow [0,1]$
is a  a Lipschitz function on $(X,\mD)$ with Lipschitz constant $1$. (In other words, $\mu$ satisfies \emph{Lipschitz condition}~\eqref{eq:Lipschitz-condition}). $(X,\mD)$ is revealed to an algorithm, whereas $\mu$ is not. In each round $t$ the algorithm chooses a strategy $x=x_t \in X$ and receives payoff $f_t(x)\in [0,1]$ chosen independently from some distribution $\prob_x$ with expectation $\mu(x)$. Without loss of generality, the diameter of $(X,\mD)$ is at most $1$. To simplify exposition, the parameterized family of reward distributions $\prob_x$ is assumed to be fixed over time, and suppressed from the notation.

Throughout the paper, $(X,\mD)$ and $\mu$ will denote, respectively, a metric space of diameter $\leq 1$ and a Lipschitz function as above. We will say that $X$ is the set of strategies (``arms''), $\mD$ is the \emph{similarity function}, and $\mu$ is the \emph{payoff function}.

\OMIT{ 
In the full feedback version (the \FFproblem), the setup is slightly more complicated. We have $(X,\mD,\mu)$ as above. Further, letting $[0,1]^X$ be the set of all functions from $X$ to $[0,1]$, we have a probability measure $\prob$ with universe $[0,1]^X$
such that
    $\mu(x) =  \E_{f \in \prob}[f(x)]$.
A problem instance is a triple $(X,\mD,\prob)$, where neither $\mu$ nor $\prob$ are revealed to the algorithm.  In each round the algorithm picks a strategy $x\in X$, then
the environment chooses an i.i.d. sample $f_t: X\to [0,1]$ according to the measure $\prob$. The algorithm receives payoff $f(x)$, and also observes the entire payoff function $f$. More formally, the algorithm can query $f$ at an arbitrary finite number of points. In the \emph{double feedback} version, the algorithm can query $f$ at at most one point other than $x$. The only assumption on $\prob$ is that it is a Borel measure
with respect to the Borel $\sigma$-algebra induced by the product topology on $[0,1]^X$.
} 

Performance of an algorithm is measured via \emph{regret} with respect to the best fixed strategy:
\begin{align}\label{eq:def-regret}
	R(t) = t\,\sup_{x\in X} \mu(x) - \E\left[ \sum_{s=1}^t \mu(x_s)\right],
\end{align}
where $x_t\in X$ is the strategy chosen by the algorithm in round $t$. Note that when the supremum is attained, the first summand in \eqref{eq:def-regret} is the expected reward of an algorithm that always plays the best strategy.

Throughout the paper, the constants in the $O(\cdot)$ notation are absolute unless specified otherwise. The notation $O_{\text{subscript}}$ means that the constant in $O()$ can depend on the things listed in the subscript.
Denote
	$\sup(\mu,X) = \sup_{x\in X} \mu(x)$
and similarly
	$\argmax(\mu,X) = \argmax_{x\in X} \mu(x)$.

\subsection*{Metric topology and set theory}

Let $X$ be a set and let $(X,\mD)$ be a metric space. An open ball of radius $r$ around point $x\in X$ is
	$B(x,r) = \{y\in X:\, \mD(x,y) <r\}$.
The diameter of a set is the maximal distance between any two points in this set.

A \emph{Cauchy sequence} in $(X,\mD)$ is a sequence such that for every $\delta>0$, there is an open ball of radius $\delta$ containing all but finitely many points of the sequence.  We say $X$ is \emph{complete} if every Cauchy sequence has a limit point in $X$.  For two Cauchy sequences $\mathbf{x}=(x_1,x_2,\ldots)$ and $\mathbf{y}=(y_1,y_2,\ldots)$ the \emph{distance} $d(\mathbf{x},\mathbf{y}) = \lim_{i \rightarrow \infty} d(x_i,y_i)$ is well-defined.  Two Cauchy sequences are declared to be equivalent if their distance is $0$.  The equivalence classes of Cauchy sequences form a metric space $(X^*,\mD)$ called the \emph{(metric) completion} of $(X,\mD)$. The subspace of all constant sequences is identified with $(X,\mD)$: formally, it is a dense subspace of $(X^*,\mD)$ which is isometric to $(X,\mD)$. A metric space $(X,\mD)$ is \emph{compact} if every collection of open balls  covering $(X,\mD)$ has a finite subcollection that also covers $(X,\mD)$.  Every compact metric space is complete, but not vice-versa.

A family $\F$ of subsets of $X$ is called a \emph{topology} if it contains $\emptyset$ and $X$ and is closed under arbitrary unions and finite intersections. When a specific topology is fixed and clear from the context, the elements of $\F$ are called \emph{open sets}, and their complements are called \emph{closed sets}.  Throughout this paper, these terms will refer to the \emph{metric topology} of the underlying metric space, the smallest topology that contains all open balls (namely, the intersection of all such topologies). A point $x$ is called \emph{isolated} if the singleton set $\{x\}$ is open. A function between topological spaces is \emph{continuous} if the inverse image of every open set is open.

A \emph{well-ordering} on a set $X$ is a total order on $X$ with the property that every non-empty subset of $X$ has a least element in this order.
In Section~\ref{sec:MaxMinLCD-UB} we use \emph{ordinals}, a.k.a. \emph{ordinal numbers}, a class of well-ordered sets that, in some sense, extends natural numbers beyond infinity. Understanding this paper requires only the basic notions about ordinals, namely the standard (von Neumann) definition of ordinals, successor and limit ordinals, and transfinite induction. The necessary material can be found in any introductory text on Mathematical Logic and Set Theory, and also on \emph{Wikipedia}.

\subsection*{Dimensionality notions}
\label{sec:dim-notions}

\newcommand{\Npack}{N^{\mathtt{pack}}} 

Throughout this paper we will use various notions of dimensionality of a metric space. The basic notion will be the \emph{covering dimension}, which is a version of the \emph{fractal dimension} that is based on covering numbers. We will also use several refinements of the covering dimension that are tuned to the Lipschitz MAB problem.

\begin{definition}\label{def:CovDim}
Let $Y$ be a set of points in a metric space $(X,\mD)$. For each $r>0$,
an \emph{$r$-covering} of $Y$ is a collection of subsets of $Y$, each of diameter strictly less than $r$, that cover $Y$. The minimal number of subsets in an $r$-covering is called the \emph{$r$-covering number} of $Y$ and denoted $N_r(Y)$. The \emph{covering dimension} of $Y$ with multiplier $c$, denoted $\COV_c(Y)$, is the infimum of all $d\geq 0$ such that
    $N_r(Y) \leq c\, r^{-d}$
for each $r>0$.
\end{definition}

\noindent This definition is \emph{robust}: $N_r(Y')\leq N_r(Y)$ for any $Y'\subset Y$, and consequently
    $\COV_c(Y') \leq \COV_c(Y)$.
While covering numbers are often defined via radius-$r$ balls rather than diameter-$r$ sets, the former alternative does not have this appealing ``robustness" property.

\begin{note}{Remark.}
Fractal dimensions of infinite spaces are often defined using $\limsup$ as the distance scale tends to $0$. The $\limsup$-version of the covering dimension would be
\begin{align}
\COV(Y)
    &\triangleq \limsup_{r\to 0} \frac{\log N_r(Y)}{\log 1/r}
        \label{eq:cov-dim-limsup}\\
    &=  \inf \left\{ \,d\geq 0:\; \exists c\;
            \forall r>0\quad N(r) \leq c r^{-d} \, \right\} \nonumber \\
    &= \lim_{c\to \infty} \COV_c(Y). \nonumber
\end{align}
This definition is simpler in that it does not require an extra parameter $c$. However, it hides an arbitrarily large constant, and is uninformative for finite metric spaces. On the contrary, the version in Definition~\ref{def:CovDim} makes the constant explicit (which allows for numerically sharper bounds), and is meaningful for both finite and infinite metric spaces.
\end{note}

\begin{note}{Remark.}
Instead of the covering-based notions in Definition~\ref{def:CovDim} one could define and use the corresponding packing-based notions. A subset $S\subset Y$ is an \emph{$r$-packing} of $Y$ if the distance between any two points in $S$ is at least $r$.
An \emph{$r$-net} of $Y$ is a set-wise maximal $r$-packing; equivalently, $S$ is an $r$-net if and only if it is an $r$-packing and the balls $B(x,r), \, x \in S$ cover $Y$.
The maximal number of points of an $r$-packing is called the \emph{$r$-packing number} of $Y$ and denoted $\Npack_r(Y)$. The ``packing dimension" can then be defined as in Definition~\ref{def:CovDim}. It is a well-known folklore result that the packing and covering notions are closely related:

\begin{fact}\label{fact:packing-covering}
$N_{2r}(Y) \leq \Npack_r(Y) \leq N_r(Y)$.
\end{fact}

\begin{proof}
Suppose the maximal size of an $r$-packing is finite, and let $S$ be an $r$-packing of this size. First, for any $r$-covering $\{Y_i\}$, each set $Y_i$ can contain at most one point in $S$, and each point in $S$ is contained in some $Y_i$. So the $r$-covering has size at least $|S|$. Thus, $\Npack_r(Y) \leq N_r(Y)$. Second, $\{B(x,r):\, x\in S\}$ is a $2r$-covering: else there exists a point $x_0$ that is not covered, and $S\cup \{x_0\}$ is an $r$-packing of larger size. So $N_{2r}(Y) \leq \Npack_r(Y)$. It remains to consider the case when there exists an $r$-packing $S$ of infinite size. Then using the same argument as above we show that any $r$-covering consists of infinitely many sets.
\end{proof}
\end{note}

For any set of finite diameter, the covering dimension (with multiplier $1$) is at most the doubling dimension, which in turn is at most $d$ for any point set in $(\R^d, \ell_p)$. The doubling dimension~\cite{Heinonen01} has been a standard notion in the theoretical computer science literature (e.g.~\cite{Gup03,Tal04,Slivkins-focs04,Cole-stoc06}). For the sake of completeness, and because we use it in Section~\ref{subsec:zooming-target}, let us give the definition: the \emph{doubling dimension} of a subset $Y\subset X$ is the smallest (infimum) $d>0$ such that any subset $S\subset Y$ whose diameter is $r$ can be covered by $2^d$ sets of diameter at most $r/2$.
The doubling dimension is much more restrictive that the covering dimension. For example, $Y=\{2^{-i}:\, i\in \N\}$ under the $\ell_1$ metric has doubling dimension $1$ and covering dimension $0$.

\OMIT{The following fact is well-known: if distance between any two points in $S$ is $>r$, then any ball of radius $r$ contains at most $\Cdbl^2$ points of $S$.}

\subsection*{Concentration inequalities}

We use an elementary concentration inequality known as {\em the Chernoff bounds}. Several formulations exist in the literature; the one we use is from~\cite{MitzUpfal-book05}.

\begin{theorem}[Chernoff Bounds~\cite{MitzUpfal-book05}]
\label{thm:chernoff}
Consider i.i.d. random variables $Z_1, \ldots, Z_n$ with values in $[0,1]$. Let
    $Z = \tfrac{1}{n} \sum_{i=1}^n Z_i$ be their average, and let $\zeta = \E[Z]$. Then:
\begin{OneLiners}
\item[(a)]
    $\Pr\left[ |Z-\zeta| > \delta \zeta \right]
        < 2\, \exp(-\zeta n \delta^2/3) $
	for any $\delta\in (0,1)$.
\item[(b)] $\Pr[ Z > a ] < 2^{-an} $
	for any $a>6\zeta$.
\end{OneLiners}
\end{theorem}

\subsection*{Initial observation: proof of Theorem~\ref{thm:intro-covDim}}

We extend algorithm $\NaiveAlg$ from metric space $([0,1], \ell_1^d)$ to an arbitrary metric space of covering dimension $d$. The algorithm is parameterized by $d$. It divides time into phases of exponentially increasing length. During each phase $i$, the algorithm picks some $\delta_i>0$, chooses an arbitrary $\delta_i$-net $S_i$ for the metric space,
and only plays arms in $S_i$ throughout the phase.
Specifically, it runs an
$|S_i|$-armed bandit algorithm
on the arms in $S_i$.
For concreteness, let us say that we use \UCB
(any other bandit algorithm with the same regret
guarantee will suffice), and each phase $i$ lasts $2^i$ rounds. The parameter $\delta_i$ is tuned optimally given $d$ and the phase duration $T$; the optimal value turns out to be $\delta = \tilde{O}(T^{-1/(d+2)})$. The algorithm can be analyzed using the technique from~\cite{Bobby-nips04}.

\begin{theorem}\label{thm:naiveAlg}
Consider the \problem on a metric space $(X,\mD)$. Let $d$ be the covering dimension of $(X,\mD)$ with multiplier $c$. Then regret of \NaiveAlg, parameterized by $d$, satisfies
\begin{align}\label{eq:thm-naiveAlg}
	R(t) = O\left( (c \log t)^{1/(d+2)}\; t^{1-1/(d+2)} \right)\quad
        \text{for every time $t$}.
\end{align}
\end{theorem}

\OMIT{
Indeed, for any $\delta'<\delta$ the metric space can be covered by $c\,\delta^{-d}$ sets of diameter at most $\delta'$, each of which can contain only one point from $S$.}

\begin{proof}
Let us analyze a given phase $i$ of the algorithm. Let $R_i(t)$ be the regret accumulated in rounds $1$ to $t$ in this phase. Let $\delta = \delta_i$ and let $K= |S_i|$ be the number of arms in this phase that are considered by the algorithm. The regret of \UCB in $t$ rounds is
    $O(\sqrt{K\, t \log t})$~\cite{bandits-ucb1}.
It follows that
$$ R_i(t) \leq O(\sqrt{K t \log t}) + t ( \mu^*-\sup(\mu,S_i)),
    \text{where $\mu^* = \sup(\mu,X)$}.
$$

Note that
    $\sup(\mu,S_i) \geq \mu^*-\delta$.
(Indeed, since $\mu$ is a Lipschitz-continuous function on a compact metric space, there exists an optimal arm $x^*\in X$ such that $\mu(x^*) = \mu^*$. Take an arm $x\in S_i$ such that $\mD(x,x^*)<\delta$. Then $\mu(x) \geq \mu^*-\delta$.)
Further,
    $K\leq c\delta^{-d}$
since $S_i$ is a $\delta$-net. We obtain:
$$ R_i(t) \leq O(\sqrt{c\,\delta^{-d}\, t\log t} + \delta t).
$$
Substituting $t= 2^i$,
    $\delta = (ct\,\log t)^{-1/(d+2)}$,
yields
	$R_i(t) = O\left( (c \log t)^{1/(d+2)}\; t^{1-1/(d+2)} \right)$.

We obtain~\eqref{eq:thm-naiveAlg} by summing over all phases
    $i = 1,2 \LDOTS \cel{\log t} $.
For the last phase $i = \cel{\log t}$ (which is possibly incomplete), the regret accumulated in this phase is at most $R_i(2^i)$.
\end{proof}

\newcommand{\Czoom}{C}
\newcommand{\Phase}{\ensuremath{\mathtt{ph}}}
\newcommand{\iPhase}{\ensuremath{i_\mathtt{ph}}\xspace}

\section{The zooming algorithm for Lipschitz MAB}
\label{sec:adaptive-exploration}

This section is on the {\it zooming algorithm}, which uses adaptive refinement to take advantage of ``benign'' input instances. We state and anlyze the algorithm, and derive a number of extensions.

The zooming algorithm proceeds in phases $i=1,2,3,\ldots\;$. Each phase $i$ lasts $2^i$ rounds. Let us define the algorithm for a single phase \iPhase of the algorithm. For each arm $x\in X$ and time $t$, let $n_t(x)$ be the number of times arm $x$ has been played in this phase before time $t$, and let $\mu_t(x)$ be the corresponding average reward. Define $\mu_t(x)=0$ if $n_t(x)=0$. Note that at time $t$ both quantities are known to the algorithm.

Define the \emph{confidence radius} of arm $x$ at time $t$ as
\begin{equation}\label{eq:confidence-radius}
	r_t(x) := \sqrt{\frac{8\,\iPhase }{1+n_t(x)}}.
\end{equation}
\noindent The meaning of the confidence radius is that with high probability
(i.e., with probability tending to 1 exponentially fast as
\iPhase increases)
it bounds from above the deviation of $\mu_t(x)$ from its expectation $\mu(x)$. That is:\footnote{Here and throughout this paper, we use the
abbreviation ``w.h.p.'' to denote the phrase {\em with high probability}.}
\begin{align}\label{eq:confRad-meaning}
	\mbox{w.h.p.} \quad |\mu_t(x) - \mu(x)| \leq r_t(x) \quad
    \text{for all times $t$ and arms $x$}.
\end{align}
Our intuition is that the samples from arm $x$ available at time $t$ allow us to estimate $\mu(x)$ only up to $\pm r_t(x)$. Thus, the available samples from $x$ do not provide enough confidence to distinguish $x$ from any other arm in the ball of radius $r_t(x)$ around $x$. Call $B(x,\, r_t(x))$ the \emph{confidence ball} of arm $x$ (at time $t$).

Throughout the execution of the algorithm, a finite number of arms are designated \emph{active}, so that in each round the algorithm only selects among the active arms. In each round at most one additional arm is activated. Once an arm becomes active, it stays active until the end of the phase. It remains to specify two things: the \emph{selection rule} which decides which arm to play in a given round, and the \emph{activation rule} which decides whether and which arm to activate.
\begin{itemize}
\item \emph{Selection rule.}
Choose an active arm $x$ with the maximal \emph{index}, defined as
\begin{align}\label{eq:index}
	I_t (x) = \mu_t(x) + 2\, r_t(x).
\end{align}
\noindent This definition of the index is meaningful because as long as~\eqref{eq:confRad-meaning} holds, the index of $x$ is an upper bound on the expected payoff of any arm in the confidence ball of $x$. (We will prove this later.) The factor 2 in~\eqref{eq:index} is needed because we ``spend'' one $+r_t(x)$ to take care of the sampling uncertainty, and another $+r_t(x)$ to generalize from $x$ to the confidence ball of $x$. Note that the index in algorithm~\UCB \citep{bandits-ucb1} is essentially $\mu_t(x) + r_t(x)$.

\item \emph{Activation rule.} Say that an arm is \emph{covered} at time $t$ if it is contained in the confidence ball of some active arm. We maintain the invariant that at each time all arms are covered. The activation rule simply maintains this invariant: if there is an arm which is not covered, pick any such arm and make it active. Note that the confidence radius of this newly activated arm is initially greater than $1$, so all arms are trivially covered. In particular, it suffices to activate at most one arm per round. The activation rule is implemented using the \emph{covering oracle}, as defined in Section~\ref{subsec:intro-access}.
\end{itemize}
The bare pseudocode of the algorithm is very simple; see Algorithm~\ref{alg:zooming}.

\begin{algorithm}[h]
\caption{Zooming Algorithm}
\label{alg:zooming}

\begin{algorithmic}
\FOR{phase $i=1,2,3, \ldots$}
\STATE Initially, no arms are active.
\FOR{round $t=1,2,3, \ldots, 2^i$}
\STATE \emph{Activation rule:} if some arm is not covered, pick any such arm and activate it.
\STATE \emph{Selection rule:} play any active arm with the maximal index~\refeq{eq:index}.
\ENDFOR
\ENDFOR
\end{algorithmic}
\end{algorithm}

To state the provable guarantees, we need the notion of \emph{zooming dimension} of a problem instance. As discussed in Section~\ref{subsec:intro-poly}, this notion bounds the covering number of near-optimal arms, thus sidestepping the worst-case lower-bound examples. Throughout this section, 
    $\mu^* \triangleq \sup(\mu,X)$
denotes the maximal reward, and $\Delta(x) = \mu^*-\mu(x)$ is the ``badness'' of arm $x$.

\begin{definition}
Consider a problem instance $(\mD,X,\mu)$.
The set of near-optimal arms at scale $r\in(0,1]$ is defined to be
\begin{align*}
X_{\mu,\,r} \triangleq \{x\in X:\, \tfrac{r}{2} < \Delta(x) \leq r  \}.
\end{align*}
The \emph{zooming dimension} with multiplier $c>0$ is the smallest $d\geq 0$ such that for every scale $r \in(0,1]$ the set $X_{\mu,\,r}$
can be covered by $c\,r^{-d}$ sets of diameter strictly less than $r/8$.
\label{def:zooming-dim}
\end{definition}

\begin{theorem}\label{thm:zooming-dim}
Consider an instance of the \problem. Fix any $c>0$ and let $d$ be the zooming dimension with multiplier $c$. Then the regret $R(t)$ of the zooming algorithm satisfies:
\begin{equation}\label{eq:thm-zooming-dim}
	R(t) \leq O(c \log t)^{\frac{1}{d+2}}\,\times t^{\frac{d+1}{d+2}}\;\;
	\text{for all times $t$.}
\end{equation}
\end{theorem}

The zooming algorithm is \emph{self-tuning} in that it does not input the zooming dimension~$d$. Moreover, it is not parameterized by the multiplier $c$, and yet it satisfies the corresponding regret bound for any given $c>0$.  For  sharper guarantees, $c$ can be tuned to the specific problem instance and specific time $t$.

Note that the regret bound in Theorem~\ref{thm:zooming-dim} has the same ``shape'' as the worst-case result (Theorem~\ref{thm:intro-covDim}), except that $d$ now stands for the zooming dimension rather than the covering dimension.
Thus, the zooming dimension is our way to quantify the benignness of a problem instance.
(It is immediate from Definition~\ref{def:zooming-dim} that the covering dimension with multiplier $c$ is an upper bound
on the zooming dimension with the same multiplier.)
Let us flesh out (and generalize) two examples from Section~\ref{sec:intro} where the zooming dimension is small:
\begin{itemize}
\item all arms with $\Delta(v)<r$ lie in a low-dimensional region $S\subset X$, for some $r>0$.

\item $\mu(x) = \max(0,\, \mu^*-\mD(x,S))$ for some $\mu^*\in(0,1]$ and subset $S\subset X$.
\end{itemize}
In both examples, for a sufficiently large constant multiplier $c$, the zooming dimension is bounded from above by $\COV(S)$ (as opposed to $\COV(X)$). Note that in the second example a natural special case is when $S$ is a finite point set, in which case $\COV(S)=0$. The technical fine print is very mild: $(X,\mD)$ can be any compact metric space, and the second example requires some open neighborhood of $S$ to have constant doubling dimension. The first example is immediate; the second example is analyzed in Section~\ref{subsec:zooming-target}.

Our proof of Theorem~\ref{thm:zooming-dim} does not require all the assumptions in the Lipschitz MAB problem. It never uses the triangle inequality, and it only needs a relaxed version of the Lipschitz condition~\refeq{eq:Lipschitz-condition}. If there exists a unique best arm $x^*$ then the relaxed Lipschitz condition is~\refeq{eq:Lipschitz-condition} with $y=x^*$. In a more efficient notation: $\Delta(x)\leq \mD(x,x^*)$ for each arm $x$. This needs to hold for each best arm $x^*$ if there is more than one. A more general version, not assuming that the optimal payoff $\sup(\mu,X)$ is attained by some arm, is as follows:
\begin{align}\label{eq:relaxedLipschitz}
	(\forall \eps>0)\quad
	(\exists x^*\in X) \quad
	(\forall x\in X)\quad
	\Delta(x) \leq \mD(x,x^*) + \eps.
\end{align}

\begin{theorem}\label{thm:zooming-dim-relaxed}
The guarantees in Theorem~\ref{thm:zooming-dim} hold even if the similarity function $\mD$ is not required to satisfy the triangle inequality,\footnote{Formally, we require $\mD$ to be a symmetric  function $X\times X\rightarrow [0, \infty]$ such that $\mD(x,x)=0$ for all $x\in X$. We call such a function a \emph{\quasimetric} on $X$.} and the Lipschitz condition~\refeq{eq:Lipschitz-condition} is relaxed to~\eqref{eq:relaxedLipschitz}.
\end{theorem}

Further, we obtain a regret bound in terms of the covering numbers.

\begin{theorem}\label{thm:zooming-covering-num}
Fix an instance of the \problem (relaxed as in Theorem~\ref{thm:zooming-dim-relaxed}). Then the regret $R(t)$ of the zooming algorithm satisfies
\begin{align*}
R(t) \leq \min_{\rho>0} \left(
    \rho t + O(\log^2 t)\,\textstyle \sum_{r\in \mS:\, r\geq \rho}
        \tfrac{1}{r} \; N_{r/8}(X_{\mu,r})
\right), \text{ where }
\mS = \{2^{-i}:\, i\in\N\}.
\end{align*}
\end{theorem}

This regret bound takes advantage of problem instances for which $X_{\mu,r}$ is a much smaller set than $X$. It can be useful even if the benignness of the problem instance cannot be summarized via a non-trivial upper-bound on the zooming dimension.

\vspace{2mm}

The rest of this section is organized as follows. In Section~\ref{subsec:zooming-analysis} we prove the above theorems. In addition, we provide some extensions and applications.
\begin{itemize}

\item In Section~\ref{subsec:zooming-max} we derive a regret bound that matches~\eqref{eq:thm-zooming-dim} and gets much smaller if the maximal payoff is close to $1$. This result relies on an improved confidence radius, which may be of independent interest.

\item In Section~\ref{subsec:zooming-target} we analyze the special case in which the expected payoff of a given arm is a function of the distance from this arm to the (unknown) ``target set'' $S\subset X$. This is a generalization of the $\mu(x) = \max(0,\, \mu^*-\mD(x,S))$ example above.
\item In Section~\ref{subsec:zooming-noise} we prove improved regret bounds for several examples in which the payoff of each arm $x$ is $\mu(x)$ plus i.i.d. noise of known and ``benign'' distribution. For these results, we replace $\mu_t(x),\,r_t(x)$ with better estimates: $\hat{\mu}_t(x),\,\hat{r}_t(x)$ such that
    $|\hat{\mu}_t(x) - \mu(x)| \leq \hat{r}_t(x)< r_t(x)$
with high probability.

\end{itemize}

\subsection{Analysis of the zooming algorithm}
\label{subsec:zooming-analysis}

First we use Chernoff bounds to prove~\eqref{eq:confRad-meaning}. A given phase will be called \emph{clean} if for each round $t$ in this phase and each arm $x\in X$ we have
	$ |\mu_t(x) - \mu(x)| \leq r_t(x)$.

\begin{claim}\label{cl:conf-rad}
Each phase \iPhase is clean with probability at least $1-4^{-\iPhase}$.
\end{claim}
\begin{proof}
The only difficulty is to set up a suitable application of Chernoff bounds along with the union bound. Let $T=2^{\iPhase}$ be the duration of a given phase \iPhase.

Fix some arm $x$. Recall that each time an algorithm plays arm $x$, the payoff is sampled i.i.d.\ from some distribution $\prob_x$. Define random variables $Z_{x,s}$ for $1 \leq s \leq T$ as follows: for $s \leq n(x)$, $Z_{x,s}$ is the payoff from the $s$-th time arm $x$ is played, and for $s > n(x)$ it is an independent sample from $\prob_x$.
For each $k\leq T$ we can apply Chernoff bounds to $\{Z_{x,s}:\, 1\leq s \leq k\}$
and obtain that
 \begin{equation} \label{eq:cl-conf-rad.0}
     \Pr\left[\; \left|
        \mu(x) - \textstyle \tfrac{1}{k}\sum_{s=1}^k Z_{x,s}
        \right| \leq
        \sqrt{\tfrac{8 \,\iPhase}{1+k}}
       \right] > 1-T^{-4}.
 \end{equation}
Let $N$ be the number of arms activated in phase \iPhase; note that $N\leq T$. Define $\mbox{$X$-valued}$ random variables $x_1,\ldots,x_T$ as follows: $x_j$ is the $\min(j,N)$-th arm activated in this phase.
For any $x \in X$ and $j \leq T$, the event $\{x = x_j\}$ is independent
of the random variables $\{Z_{x,s}\}$; the former event depends only on
payoffs observed before $x$ is activated, while the latter set of random
variables has no dependence on payoffs of arms other than $x$.
Therefore,~\eqref{eq:cl-conf-rad.0} remains valid if we
replace the probability on the left side with conditional probability,
conditioned on the event $\{x = x_j\}$.
Taking the union bound over all $k\leq T$, and using the notation of $\mu_t(x)$ and $r_t(x)$, it follows that
    $$ \Pr\left[\; \forall t\; \left| \mu(x) - \mu_t(x)\right| \leq r_t(x) \,\mid\, x_j=x \;\right]
        > 1-T^{-3},$$
where $t$ ranges over all rounds in phase \iPhase. Integrating over all arms $x$ we obtain
$$ \Pr\left[\; \forall t\; \left| \mu(x_j) - \mu_t(x_j)\right| \leq r_t(x_j) \;\right]
        > 1-T^{-3}.$$
Finally, we obtain the claim by taking the union bound over all $j\leq T$.
\end{proof}

Next, we present a crucial argument which connects the best arm and the arm played at a given round, which in turn allows us to bound the number of plays of a suboptimal arm in terms of its badness.

\begin{lemma} \label{lm:bound-active}
If phase \iPhase is clean then we have
	$\Delta(x) \leq 3\, r_t(x)$
for any time $t$ and any arm $x$.
\end{lemma}

\begin{proof}
Suppose arm $x$ is played at time $t$ in clean phase \iPhase.
First we claim that
	$I_t(x)\geq \mu^*$.
Indeed, fix $\eps>0$. By definition of $\mu^*$ there exists a arm $x^*$ such that $\Delta(x^*) < \eps$. Recall that all arms are covered at all times, so there exists an active arm $x_t$ that covers $x^*$ at time $t$, meaning that $x^*$ is contained in the confidence ball of $x_t$. Since arm $x$ was chosen over $x_t$, we have
	$I_t(x) \geq I_t(x_t)$.
Since this is a clean phase, it follows that
	$I_t(x_t) \geq \mu(x_t) + r_t(x_t)$.
By the Lipschitz property we have
	$\mu(x_t) \geq \mu(x^*) - \mD(x_t, x^*)$.
Since $x_t$ covers $x^*$, we have
	$\mD(x_t, x^*) \leq r_t(x_t)$
Putting all these inequalities together, we have
	$I_t(x) \geq \mu(x^*) \geq \mu^* - \eps$.
Since this inequality holds for an arbitrary $\eps>0$, we in fact have $I_t(x)\geq \mu^*$. Claim proved.

Furthermore, note that by the definitions of ``clean phase'' and ``index''
we have
	$$\mu^* \leq I_t(x) \leq \mu(x) + 3\, r_t(x)$$
and therefore
	$\Delta(x) \leq 3\, r_t(x)$.

Now suppose arm $x$ is not played at time $t$. If it has never been played before time $t$ in this phase, then $r_t(x) > 1$ and thus  the lemma is trivial. Else, let $s$ be the last time arm $x$ has been played before time $t$. Then by definition of the confidence radius
	$r_t(x) = r_s(x) \geq \tfrac13\, \Delta(x) $.
\end{proof}

\begin{corollary}\label{cor:bound-active}
If phase \iPhase is clean then each arm $x$ is played at most $O(\iPhase)\, \left(\Delta(x) \right)^{-2}$
times.
\end{corollary}
\begin{proof}
This follows by plugging the definition of the confidence radius into Lemma~\ref{lm:bound-active}.
\end{proof}


\begin{corollary} \label{cor:sparsity}
In a clean phase, for any active arms $x,y$ we have
	$\mD(x,y) > \tfrac13 \min(\Delta(x), \Delta(y))$.
\end{corollary}

\begin{proof}
Assume $x$ has been activated before $y$. Let $s$ be the time when $y$ has been activated. Then by the algorithm specification we have
	$\mD(x,y)> r_s(x)$.
By Lemma~\ref{lm:bound-active}
	$r_s(x)\geq \tfrac13 \Delta(x)$.
\end{proof}

Consider round $t$ which belongs to a clean phase \iPhase. Let $S_t$ be the set of all arms that are active at time $t$, and let
	$$A_{(i,t)} = \left\{ x\in S_t:\; 2^i \leq \tfrac{1}{\Delta(x)} < 2^{i+1} \right\}.$$

\noindent Recall that by Corollary~\ref{cor:bound-active} for each $x\in A_{(i,t)}$ we have
	$ n_t(x)  \leq O(\log t)\, (\Delta(x))^{-2}$.
Therefore:
\begin{align*}
\sum_{x\in A_{(i,t)}} \Delta(x)\, n_t(x)
\leq O(\log t) \, \sum_{x\in A_{(i,t)}} \tfrac{1}{\Delta(x)}
\leq O(2^i\,\log t) \, |A_{(i,t)}|.
\end{align*}

\noindent
Letting $r=2^{-i}$, note that by Corollary~\ref{cor:sparsity} any set of diameter less than $r/8$ contains at most one arm from $A_{(i,t)}$. It follows that
	$|A_{(i,t)}| \leq N_{r/8}(X_{\mu,r})$,
the smallest number of sets of diameter less than $r/8$ sufficient to cover all arms $x$ such that $\tfrac{r}{2} < \Delta(x)\leq r$.
It follows that
\begin{align*}
\sum_{x\in A_{(i,t)}} \Delta(x)\, n_t(x)
\leq O(\log t) \, \tfrac{1}{r}\, N_{r/8}(X_{\mu,r}).
\end{align*}

Let $\mS = \{2^{-i}:\, i\in\N\}$.
For each $\rho\in(0,1)$,  we have:
\begin{align}
\sum_{x\in S_t} \Delta(x)\, n_t(x)
&\leq
    \sum_{x\in S_t:\; \Delta(x) \leq  \rho} \Delta(x)\, n_t(x)
  \;+\; \sum_{i< \log(1/\rho)}\; \sum_{x\in A_{(i,t)}} \Delta(x)\, n_t(x) \nonumber \\
  &\leq \rho (t-2^{\iPhase-1})
    + O(\log t)\,\textstyle \sum_{r\in \mS:\, r\geq \rho}
        \tfrac{1}{r} \; N_{r/8}(X_{\mu,r}).
        \label{eq:clean-phase}
\end{align}
Here $t-2^{\iPhase-1}$ is the number of rounds in phase \iPhase before and including round $t$.

Let $R_\Phase(t)$ be the left-hand side of~\eqref{eq:clean-phase}. By Claim~\ref{cl:conf-rad}, the probability that phase \iPhase is non-clean is negligible. Therefore, we obtain the following:

\begin{claim}
Fix round $t$ and let \iPhase be the round to which $t$ belongs. Then:
\begin{align}
\E[R_\Phase(t)]
&\leq \inf_{\rho>0} \left( \rho (t-2^{\iPhase-1})
    + O(\log t)\,\textstyle \sum_{r\in \mS:\, r\geq \rho}
        \tfrac{1}{r} \; N_{r/8}(X_{\mu,r})
    \right) \label{eq:clean-phase-2}.
\end{align}
\end{claim}

We complete the proof as follows. Let $t$ be the current round and let \iPhase be the current phase. Let $t_i=2^i \; (\mbox{for } i<\iPhase)$ be the last round of each phase $i<\iPhase$, and let $t_{\iPhase}=t$. Note that regret up to time $t$ can be expressed as
    $$ R(t) = \sum_{i\leq \iPhase} \E\left[\, R_\Phase(t_i) \,\right]. $$

Theorem~\ref{thm:zooming-covering-num} follows by summing up \eqref{eq:clean-phase-2} over all phases $i\leq \iPhase$.

We derive Theorem~\ref{thm:zooming-dim-relaxed} from \eqref{eq:clean-phase-2} as follows. Note that
    $N_{r/8}(X_{\mu,r}) \leq c\, r^{-d} $
by definition of the zooming dimension $d$ with multiplier $c>0$. For a given phase $i$, letting $t_0 = t_i -2^{i-1}$ and choosing $\rho$ such that
        $ \rho\, t_0 = (\tfrac{1}{\rho})^{d+1} (c\,\log t)$,
we obtain
\begin{align*}
\E[R_\Phase(t_i)]
 \leq O(c \log t)^{1/(d+2)}\,\times t_0^{(d+1)/(d+2)}.
\end{align*}
We obtain Theorem~\ref{thm:zooming-dim-relaxed} by summing this over all phases $i\leq \iPhase$.

\subsection{Extension: maximal expected payoff close to 1}
\label{subsec:zooming-max}

\newcommand{\Rrel}{\hat{r}}

We obtain a sharper regret bound which matches~\eqref{eq:thm-zooming-dim} and gets much smaller if the optimal reward $\mu^* = \sup(\mu,X)$ is close to~$1$. The key ingredient here is a more elaborate confidence radius:
\begin{equation}\label{eq:confRad-sharp}
 \Rrel_t(x) \triangleq \frac{\alpha}{1+n_t(x)} +
	\sqrt{ \alpha\; \frac{1-\mu_t(x)}{1+n_t(x)}}\;\;
	\text{for some $\alpha= \Theta(\iPhase)$}.
\end{equation}
The confidence radius in~\eqref{eq:confRad-sharp} performs as well as $r_t(\cdot)$ (up to constant factors) in the worst case:
    $\Rrel_t(x) \leq \sqrt{\tfrac{O(\iPhase)}{n_t(x)}}$,
and gets much better when $\mu_t(x)$ is close to~$1$:
    $\Rrel_t(x) \leq \tfrac{O(\iPhase)}{n_t(x)}$.
Note that the right side of~\eqref{eq:confRad-sharp} can be computed from the observable data; in particular, it does not require the knowledge of $\mu^*$.

\begin{theorem}\label{thm:mu-close-to-1}
Consider an instance of the \problem, in the relaxed setting of Theorem~\ref{thm:zooming-dim-relaxed}. Fix any $c>0$ and let $d$ be the zooming dimension with multiplier $c$. Let $\mu^* = \sup(\mu,X)$ be the optimal reward. Then zooming algorithm with confidence radius~\eqref{eq:confRad-sharp} satisfies, for all times $t$:
\begin{align*}
R(t) \leq O(c\log^2 t)+ O(c \log t)^{\frac{1}{d+1}}
    \,\times \max\left(
        t^{1-\tfrac{1}{d+1}},\; (1-\mu^*)\, t^{1-\frac{1}{d+2}}
    \right).
\end{align*}
\end{theorem}

Compared to the regret bound in Theorem~\ref{thm:zooming-dim}, this result effectively reduces the zooming dimension by $1$ if $\mu^*$ is close to $1$ (and $d>1$). Moreover, regret becomes \emph{polylogarithmic} if $\mu^*=1$ and $d=0$.

We analyze the new confidence radius~\refeq{eq:confRad-sharp} using the following corollary of Chernoff bounds which, to the best of our knowledge, has not appeared in the literature, and may be of independent interest.

\begin{theorem}
\label{lm:my-chernoff}
Consider $n$ i.i.d. random variables $Z_1 \ldots Z_n$ on $[0,1]$. Let $Z$ be their average, and let $\zeta = \E[Z]$. Then for any $\alpha>0$, letting
$r(\alpha,x)
		= \tfrac{\alpha}{n} + \sqrt{\tfrac{\alpha x}{n}}$,
we have:
\begin{align*}
 \Pr\left[\, |Z-\zeta| < r(\alpha,Z) < 3\,r(\alpha,\zeta) \, \right]
    > 1- \left( 2^{-\alpha} + 2\,e^{-\alpha/72} \right).
\end{align*}
\end{theorem}

\begin{proof}
Suppose $\zeta< \tfrac{\alpha}{6n}$. Then using Chernoff bounds (Theorem~\ref{thm:chernoff}(b)) with $a = \tfrac{\alpha}{n}>6 \zeta$, we obtain that with probability at least $1- 2^{-\alpha}$ we have
	$Z < \tfrac{\alpha}{n}$,
and therefore
    $|Z-\zeta| < \tfrac{\alpha}{n} < r(\alpha, Z)$
and
$$ |Z-\zeta| < \tfrac{\alpha}{n} < r(\alpha, Z) <  (1+\sqrt{2})\, \tfrac{\alpha}{n}
	< 3\, r(\alpha, \zeta).
$$

Now, suppose
	$\zeta\geq \tfrac{\alpha}{6n} $. Apply Chernoff bounds (Theorem~\ref{thm:chernoff}(a)) with
	$\delta = \tfrac12 \sqrt{\tfrac{\alpha}{6\zeta n}}$.
Thus with probability at least $1- 2\,e^{-\alpha/72}$ we have
	$|Z-\zeta| < \delta\zeta \leq \zeta/2$.
Plugging in $\delta$,
\begin{align*}
 |Z-\zeta|
	< \tfrac12  \sqrt{\tfrac{\alpha \zeta}{n}}
	\leq \sqrt{\tfrac{\alpha Z}{n}}
	\leq r(\alpha, Z) < 1.5\, r(\alpha, \zeta). \qquad \qquad \qedhere
\end{align*}
\end{proof}

\begin{proofof}{Theorem~\ref{thm:mu-close-to-1}}
Let us fix an arm $x$ and time $t$. Let us use Theorem~\ref{lm:my-chernoff} with $n=n_t(x)$ and $\alpha = \Theta(\iPhase)$ as in~\eqref{eq:confRad-sharp}, setting each random variable $X_i$ equal to 1 minus the reward from the $i$-th time arm $x$ is played in the current phase. Then $\zeta = 1-\mu(x)$ and $Z = 1-\mu_t(x)$, so the theorem says that
\begin{align}\label{eq:max1-pf}
 \Pr\left[\,
	|\mu_t(x)-\mu(x)| < r_t(x)
	< 3\left(
			\frac{\alpha}{n_t(x)} + \sqrt{\frac{\alpha\, (1-\mu(x))}{n_t(x)}}\,
 	   \right)
 	\right]
  > 1-2^{\Omega(\alpha)}.
\end{align}

We modify the analysis in Section~\ref{subsec:zooming-analysis} as follows. We redefine a ``clean phase'' to mean that the event in the left-hand side of~\eqref{eq:max1-pf} holds for all rounds $t$ and all arms $x$. We use~\eqref{eq:max1-pf} instead of the standard Chernoff bound in the proof of Claim~\ref{cl:conf-rad} to show that each phase \iPhase is clean with probability at least $1-4^{\iPhase}$. Then we obtain Lemma~\ref{lm:bound-active} as is, for the new definition of $r_t(x)$.
Then we replace Corollary~\ref{cor:bound-active} with a more efficient corollary based on the new $r_t(x)$. More precisely, we derive two regret bounds: one assuming
    $r_t(x) = \frac{O(\iPhase)}{n_t(x)}$,
and another assuming
    $r_t(x) = \sqrt{\frac{O(\iPhase)(1-\mu^*)}{n_t(x)}}$,
and take the maximum of the two. We omit the easy details.
\end{proofof}


\subsection{Application: Lipschitz MAB with a ``target set''}
\label{subsec:zooming-target}

We consider a version of the \problem in which the expected reward of each arm $x$ is determined by the distance between this arm and a fixed \emph{target set} $S\subset X$ which is not revealed to the algorithm. Here the distance is defined as
    $\mD(x,S) \triangleq \inf_{y\in S} \mD(x,y)$.
The motivating example is
    $\mu(x) = \max(0,\,\mu^*-\mD(x,S))$.
More generally, we assume that
    $ \mu(x) = f(\mD(x,S))) $
for each arm $x$, for some known non-increasing function $f:[0,1]\to [0,1]$.
We call this version the \emph{Target MAB problem} with target set $S$ and shape function $f$.%
\footnote{Note that the payoff function $\mu$ does not necessarily satisfy the Lipschitz condition with respect to $\mD$.
However, if $f(z) = \mu^*-z$ then $\mu(x) = \mu^* - \mD(x,S)$, and the Lipschitz condition is satisfied because
$\mD(x,S)-\mD(y,S) \leq \mD(x,y)$.
}

The key idea is to use the quasi-distance
function
	$\mD_f(x,y) = f(0) - f(\mD(x,y))$.
It is easy to see that $\mD_f$ satisfies~\eqref{eq:relaxedLipschitz}. Indeed, fix any arm $x^*\in S$. Then for each $x\in X$ we have:
$$ \Delta(x) = \mu(x^*) - \mu(x) = f(0) - f(\mD(x,S)) = \mD_f(x,S)
	\leq \mD_f(x,x^*).
$$
Therefore Theorem~\ref{thm:zooming-dim-relaxed} applies: we can use the zooming algorithm in conjunction with $\mD_f$ rather than $\mD$. The performance of this algorithm depends on the zooming dimension of the problem instance $(X,\mD_f,\mu)$.

\begin{theorem}\label{thm:targetMAB-general}
Consider the Target MAB problem with target set $S\subset X$ and shape function $f$. For some fixed multiplier $c>0$, let $d$ be the zooming dimension of  $(X,\mD_f,\mu)$. Then the zooming algorithm on $(\mD_f,X)$ has regret
$ R(t) \leq (c\,\log t)^{\frac{1}{d+2}}\; \; t^{\frac{d+1}{d+2}}$
for all times $t$.
\end{theorem}

Note that the zooming algorithm is self-tuning: it does not need to know the properties of $S$ or $f$, and in fact it does not even need to know that it is presented with an instance of the Target MAB problem. We obtain a further improvement via Theorem~\ref{thm:mu-close-to-1} if $f(0)$ is close to $1$.

Let us consider the main example
    $\mu(x) = \max(0,\,\mu^* - \mD(x,S))$
and more generally
\begin{align}\label{eq:targetMAB-example}
    \mu(x) = \max(\mu_0,\,\mu^*- \mD(x,S)^{1/\alpha})
\end{align}
for some constant $\alpha>0$ and $0\leq \mu_0 <\mu^* \leq 1$. Here $\mu_0$ and $\mu^*$ are, respectively, the minimal and maximal expected payoffs. \eqref{eq:targetMAB-example} corresponds to
    $f(z) = \max(\mu_0,\,\mu^*-z^{1/\alpha})$.
Then
    $$\mD_f(x,y) = \min(\mu^*-\mu_0,\;(\mD(x,y))^{1/\alpha}).$$

\noindent
We find that the zooming dimension of the problem instance $(X,\mD_f,\mu)$ is, essentially, at most $\alpha$ times the covering dimension of $S$. (This result holds as long as $(X,\mD)$ has constant doubling dimension.) Intuitively, $S$ is a low-dimensional subset of the metric space, in the sense that it has a (much) smaller covering dimension.

\newcommand{\dDBL}{d_{\mathtt{DBL}}}
\newcommand{\cZOOM}{c_{\mathtt{zoom}}}

\begin{lemma}\label{lm:targetMAB-shape}
Consider the Target MAB problem with payoff function given by~\eqref{eq:targetMAB-example}. Let $d$ be the covering dimension of the target set $S$, for any fixed multiplier $c>0$. Let $\dDBL$ be the doubling dimension of $(X,\mD)$; assume it is finite. Then the zooming dimension of $(X,\mD_f,\mu)$ is $\alpha d$, with constant multiplier
\begin{align*}
    \cZOOM = \left(\max\left( c\, 2^{4\alpha+2},\;
        \tfrac{2}{\mu^*-\mu_0} \right)\right)^{\dDBL}.
\end{align*}
\end{lemma}

\begin{proof}
For each $r>0$, it suffices to cover the set
	$S_r = \{x\in X:\, \Delta(x)\leq r\}$
with $\cZOOM\, r^{-\alpha d}$ sets of $\mD_f$-diameter at most $r/16$.
Note that
    $ \Delta(x) = \min(\mu^*-\mu_0,\, (\mD(x,S))^{1/\alpha})$.

Assume $r<\mu^*-\mu_0$. Then for each
	$x\in S_r$ we have $ \mD(x,S) \leq r^\alpha$.
By definition of the covering dimension, $S$ can be covered with
    $c\,r^{-\alpha d}$
sets
	$\{\, C_i \,\}_i$
of $\mD$-diameter at most $r^\alpha$. It follows that $S_r$ can be covered with $r^{-\alpha d}$ sets
	$\{\, B(C_i, r) \,\}_i$,
where $B(C_i, r) \triangleq \cup_{u\in C_i}\,B(x,r) $. The $\mD$-diameter of each such set is at most $3\,r^{\alpha}$. Since $\dDBL$ is the doubling dimension of $(X,\mD)$, each $B(C_i, r)$ can be covered by with $2^{(4\alpha+2)\,\dDBL}$ of sets of $\mD$-diameter at most $(r/16)^\alpha$. Therefore, $S_r$ can be covered by
	$c\, 2^{(4\alpha+2)\,\dDBL}\,r^{-\alpha d} $
sets whose $\mD$-diameter is at most $(r/16)^\alpha$, so that their $\mD_f$-diameter is at most $r/16$.

For $r\geq \mu^*-\mu_0$ we have $S_r = X$, and by definition of the doubling dimension $X$ can be covered by
    $\left(\frac{2}{\mu^*-\mu_0}\right)^{\dDBL}$
sets of diameter at most
    $\mu^*-\mu_0$.
\end{proof}

The most striking (and very reasonable) special case is when $S$ consists of finitely many points.

\begin{corollary}\label{cor:targetMAB-shape}
Consider the Target MAB problem with payoff function given by~\eqref{eq:targetMAB-example}. Suppose the target set $S$ consists of finitely many points. Let $\cZOOM$ be from Lemma~\ref{lm:targetMAB-shape} with $c=|S|$.
Then the zooming algorithm on $(\mD_f,X)$ has regret
    $ R(t) = O(\sqrt{\cZOOM\,t\,\log t})$
for all times $t$. Moreover, the regret is
    $R(t) = O(\cZOOM \log t)^2 $
if $\mu^*=1$.
\end{corollary}

\begin{proof}
The covering dimension of $S$ is $0$ with multiplier $c=|S|$. Then by Lemma~\ref{lm:targetMAB-shape} the zooming dimension is $0$, with multiplier $\cZOOM$. We obtain the $\tilde{O}(\sqrt{\cZOOM\, t\log t})$ regret using Theorem~\ref{thm:targetMAB-general}, and the $O(\cZOOM\, \log t)^2$ regret result  using Theorem~\ref{thm:mu-close-to-1}.
\end{proof}

The proof of Lemma~\ref{lm:targetMAB-shape} easily extends to shape functions $f$ such that
\begin{align*}
    x^{1/\alpha} \leq f(0)-f(x) \leq x^{1/\alpha'}\quad \forall x\in (0,1],
\end{align*}
for some constants $\alpha\geq\alpha'>0$. Then, using the notation in Lemma~\ref{lm:targetMAB-shape}, the zooming dimension of $(X,\mD_f,\mu)$ is $\alpha d$, with multiplier
$\cZOOM = \max\left( c\, 2^{(4\alpha'+2)\,\dDBL},\;
        \left(\frac{2}{\mu^*-\mu_0}\right)^{\dDBL} \right)$.

\subsection{Application: mean-zero noise with known shape}
\label{subsec:zooming-noise}

\newcommand{\PP}{\ensuremath{\mathcal{P}}}

Improved regret bounds are possible if the reward from playing each arm $x$ is $\mu(x)$ plus noise of known shape. More precisely, we assume that the reward from playing arm $x$ is $\mu(x)$ plus an independent random sample from some fixed, mean-zero distribution $\PP$, called the \emph{noise distribution}, which is revealed to the algorithm. We call this version the \emph{noisy Lipschitz MAB problem}. We present several examples in which we take advantage of a ``benign" shape of \PP. In these examples, the payoff distributions are not restricted to have bounded support.\footnote{Recall that throughout the paper the payoff distribution of each arm $x$ has support
	$\mathcal{S}(x)\subset [0, 1]$.
In this subsection, by a slight abuse of notation, we do not make this assumption.}

\xhdr{Normal distributions.}
We start with perhaps the most natural example when the noise distribution $\PP$ is the zero-mean normal distribution. Then instead of the confidence radius $r_t$ defined by~\eqref{eq:confidence-radius} we can use the confidence radius
    $\hat{r}_t(\cdot) = \sigma\, r_t(\cdot)$,
where $\sigma$ is the standard deviation of $\PP$. Consequently we obtain a regret bound~\refeq{eq:thm-zooming-dim} with the right-hand side multiplied by $\sigma$.

In fact, this result can be generalized to all noise distributions $\PP$ such that
\begin{align}\label{eq:stoc-bdd}
 \E_{Z\sim \PP}\left[ e^{rZ} \right] \leq e^{r^2 \sigma^2/2}
    \text{~~for all $r\in [-\rho, \rho]$}.
\end{align}
The normal distribution with standard deviation $\sigma$ satisfies~\eqref{eq:stoc-bdd} for $\rho=\infty$. Any distribution with support $[-\sigma,\sigma]$ satisfies~\eqref{eq:stoc-bdd} for $\rho=1$.  The meaning of~\eqref{eq:stoc-bdd} is that it is precisely the condition needed to establish an Azuma-type inequality: if $Z_1,\ldots,Z_n$ are independent samples from $\PP$ then $\sum_{i=1}^n Z_i \leq \Tilde{O}(\sigma \sqrt{n})$ with high probability. More precisely:
\begin{equation}\label{eq:stochBdd-Azuma}
 \Pr \left[ \textstyle \sum_{i=1}^n Z_i > \lambda \sigma \sqrt{n}  \right]
    \leq \exp(-\lambda^2/2)\quad
    \text{for any $\lambda \leq  \tfrac12\, \rho\,\sigma\sqrt{n} $.}
\end{equation}
We can derive an analog of Claim~\ref{cl:conf-rad} for the new confidence radius $\hat{r}_t(\cdot) = \sigma\, r_t(\cdot)$  by using~\eqref{eq:stochBdd-Azuma} instead of the standard Chernoff bound; we omit the easy details.

\xhdr{Tool: generalized confidence radius.}
More generally, we may be able to use a different, smaller confidence radius $\hat{r}_t(\cdot)$ instead of $r_t(\cdot)$ from~\eqref{eq:confidence-radius}, perhaps in conjunction with a different estimate $\hat{\mu}_t(\cdot)$ of $\mu(\cdot)$ instead of the sample average $\mu_t(\cdot)$. We will need the pair $(\hat{\mu}_t, \hat{r}_t)$ to satisfy an analog of Claim~\ref{cl:conf-rad}:
\begin{align}\label{eq:confRad-gen}
\Pr\left[\; |\hat{\mu}_t(x) - \mu(x)| \leq \hat{r}_t(x) \;
        \text{for all times $t$ and arms $x$} \;\right]
    \geq 1-4^{-\iPhase}.
\end{align}
Further, we will need the confidence radius $\hat{r}_t$ to be small in the following sense:
\begin{align}\label{eq:confRad-small}
\text{for each arm $x$ and any $r>0$, inequality
$\hat{r}_t(x) \leq r$ implies $n_t(x)\leq c_0\, r^{-\beta}\log t$},
\end{align}
for some constants $c_0$ and $\beta\geq 0$. Recall that $\hat{r}_t=r_t$ satisfies ~\eqref{eq:confRad-small} with $\beta=2$ and $c_0=O(1)$.

\begin{lemma}\label{lm:confRad-generalized}
Consider the \problem (relaxed as in Theorem~\ref{thm:zooming-dim-relaxed}).
Consider the zooming algorithm with estimator $\hat{\mu}_t$ and confidence radius
$\hat{r}_t$, and consider a problem instance such that the pair
    $(\hat{\mu}_t,\hat{r}_t)$
satisfies~\eqref{eq:confRad-gen}. Suppose $\hat{r}_t$ satisfies~\eqref{eq:confRad-small}. Let $d$ be the zooming dimension of the problem instance, for any fixed multiplier $c>0$. Then regret of the algorithm is
\begin{align}\label{eq:confRad-generalized}
R(t) \leq O(c\,c_0\,\log^2 t) +
      O(c\,c_0\,\log^2 t)^{1/(d+\beta)}\;\times t^{1-1/(d+\beta)} \;
      \text{for all times $t$}.
\end{align}
\end{lemma}

Lemma~\ref{lm:confRad-generalized} is proved by plugging in the improved confidence radius into the analysis in Section~\ref{subsec:zooming-analysis}; we omit the easy details. We obtain an improvement over Theorem~\ref{thm:zooming-dim} and Theorem~\ref{thm:zooming-dim-relaxed} whenever $\beta<2$. Below we give some examples for which we can construct improved
    $(\hat{\mu}_t,\hat{r}_t)$.

\xhdr{Example: deterministic rewards.}
For the important special case of deterministic rewards, we obtain regret bound~\refeq{eq:confRad-generalized} with $\beta=0$. (The proof is a special case of the next example.)

\begin{corollary}\label{cor:confRad-deterministic}
Consider the \problem with deterministic rewards (relaxed as in Theorem~\ref{thm:zooming-dim-relaxed}). Then the zooming algorithm with suitably defined estimator $\hat{\mu}_t$ and confidence radius $\hat{r}_t$ achieve regret bound~\refeq{eq:confRad-generalized} with $\beta=0$.
\end{corollary}

\xhdr{Example: noise distribution with a point mass.}
Consider noise distributions \PP\ having at least one \emph{point mass}: a point $z\in\R$ of positive probability mass: $\PP(z)>0$. (Deterministic rewards correspond to the special case $\PP(0)=1$).

\begin{corollary}\label{cor:confRad-pointMass}
Consider the \problem (relaxed as in Theorem~\ref{thm:zooming-dim-relaxed}). Assume mean-zero noise distribution with at least one point mass. Then the zooming algorithm with suitably defined estimator $\hat{\mu}_t$ and confidence radius $\hat{r}_t$ achieve regret bound~\refeq{eq:confRad-generalized} with $\beta=0$.
\end{corollary}

\begin{proof}
We will show that we can use a confidence radius
    $\hat{r}_t(u)= r_t(u)\,\indicator{n_t(u)\leq c_\PP \log t}$,
for some constant $c_\PP$ that depends only on $\PP$. This implies regret bound~\refeq{eq:confRad-generalized} with $\beta=0$.

Indeed, let $p = \max_{z\in \R} \PP(z)$ be the largest point mass in distribution $\PP$, and $q = \max_{z\in\R:\, \PP(z)<p} \PP(z)$ be the second largest point mass. Let
    $S = \{ z\in \R: \PP(z)=p\}$,
and let $k = |S| + \frac1q$ if $q>0$, or $k=|S|$ if $q=0$.
Then for some
    $c_\PP = \Theta\left(\tfrac{\log(|S|+k)}{p-q} \right)$,
it suffices to have $n\geq c_\PP \log t$ independent samples from $\PP$  to ensure that with probability at least $1-t^{-4}$ each number $z\in S$ is sampled at least $n(p+q)/2$ times, whereas any number $z\not\in S$ is sampled less often.%
\footnote{To prove that each number $z\not\in S$ is sampled less than $n(p+q)/2$ times when $q>0$, we need to be somewhat careful in how we apply the Union Bound. It is possible to partition the set $\R\setminus S$ into at most $O(|S|+\tfrac{1}{q})$ measurable subsets, namely intervals or points, whose measure is at most $q$ (and at least $q/2$). Apply Chernoff bound to each subset separately, then take the Union Bound.}

For a given arm $x$ and time $t$, we define a new estimator $\hat{\mu}_t(x)$ as follows. Let $n = n_t(x)$ be the number of rewards from $x$ so far. If $n<c_\PP \log t$, use the sample average: let $\hat{\mu}_t(x) = \mu_t(x)$. Else, let $R$ be the set of rewards that have appeared at least $n(p+q)/2$ times. Then $R = \mu(x)+S$ with probability at least $1-t^{-4}$. In particular,
    $\max(R) = \mu(x) + \max(S)$.
So we can define
    $\hat{\mu}_t(x) = \max(R)-\max(S)$.
\end{proof}

\xhdr{Example: noise distributions with a sharp peak.}
If the noise distribution $\PP$ has a sharp peak around $0$, then small regions around this peak can be identified more efficiently than using the standard confidence radius $r_t$.

More precisely, suppose $\PP$ has a probability density function $f(z)$ which is symmetric around $0$ and non-increasing for $z>0$, and suppose $f(z)$ has a sharp peak:
	$ f(z) = \Theta(|z|^{-\alpha})$
on some open neighborhood of $0$, for some constant $\alpha\in (0,1)$. We will show that we can use a new confidence radius
    $\hat{r}_t(x) = C\,(\iPhase/n_t(x))^{1/(1-\alpha)}$,
for a sufficiently high constant $C$, which leads to regret bound~\refeq{eq:confRad-generalized} with $\beta = 1-\alpha$.

Fix arm $x$ and time $t$. We define the estimator $\hat{\mu}_t(x)$ as follows.
Let $S$ be the multiset of rewards received from arm $x$ so far. Let
    $r = \tfrac12\, \hat{r}_t(x)$.
Cover the $[0,1]$ interval with $\cel{1/r}$ subintervals $I_j = [jr,\, (j+1)r)$.
Pick the subinterval that has most points from $S$ (break ties arbitrarily), and define $\hat{\mu}_t(x)$ as some point in this subinterval.

Let us show that
    $|\mu(x) - \hat{\mu}_t(x)| \leq \hat{r}_t(x)$
with high probability. Let $I_j$ be the subinterval that contains $\mu(x)$. Let
    $n = n_t(x)$
be the number of times arm $x$ has been played so far;
note that
    $n > \Omega(C\, r^{\alpha-1}\,\log t)$.
By Chernoff bounds, for a sufficiently high constant $C$, it holds that with probability at least $1-t^{-4}$ subinterval $I_j$ contains more points from $S$ than any other subinterval $I_\ell$ such that $|j-\ell| \geq 2$. Conditional on this high-probability event, the estimate
    $\hat{\mu}_t(x)$
lies in subinterval $I_\ell$ such that $|j-\ell|\leq 1$, which implies that
    $|\mu(x) - \hat{\mu}_t(x)| \leq 2r$.

\section{Optimal per-metric performance}
\label{sec:pmo}

This section is concerned with Question (Q1) raised in Section~\ref{subsec:intro-scope}: What is the best possible algorithm for the
\problem on a given metric space $(X,\mD)$.  We consider the worst-case regret of a given algorithm over all possible problem instances on $(X,\mD)$.%
\footnote{Formally, we can define the \emph{per-metric performance} of an algorithm on a given metric space as the worst-case regret of this algorithm over all problem instances on this metric space.}
We focus on minimizing the exponent $\gamma$ such that for each payoff function $\mu$ the algorithm's regret is
    $R(t) \leq t^\gamma$
for all $t\geq t_0(\mu)$. With Theorem~\ref{thm:intro-covDim} in mind, we will use a more focused notation: we define the \emph{regret dimension} of an algorithm on $(X,\mD)$ as, essentially, the smallest $d\geq 0$ such that one can achieve the exponent
    $\gamma = \tfrac{d+1}{d+2}$.

\OMIT{
\begin{definition}\label{def:tractability}
Consider the \problem on a fixed metric space. A bandit algorithm is \emph{$f(t)$-tractable} if its regret satisfies
	$R(t) \leq C_\mu\, f(t)$ for all $t$,
for some constant $C_\mu$ that can depend on the payoff function $\mu$. The problem is \emph{$f(t)$-tractable} if such an algorithm exists.
\end{definition}
} 

\begin{definition}
Consider the \problem on a given metric space $(X,\mD)$. For algorithm \A\ and payoff function $\mu$, define the \emph{instance-specific regret dimension} of $\A$ as
\begin{align*}
\RegretDim_{\mu}(\A)
    &= \inf \{ {d\geq 0} \;\; \mid  \;\; \exists t_0=t_0(\mu)\quad
	       R_\A(t) \leq t^{1-1/(d+2)} \quad \text{for all $t\geq t_0$}  \} \\
    &= \inf \{ {d\geq 0} \;\; \mid   \;\; \exists C=C(\mu) \quad
	       R_\A(t) \leq C\,t^{1-1/(d+2)} \quad \text{for all $t$} \} .
\end{align*}
The \emph{regret dimension} of \A\ is
	$\RegretDim(\A) = \sup_{\mu}\, \RegretDim_{\mu}(\A)$,
where the supremum is over all payoff functions $\mu$.
\end{definition}

Thus, according to  Theorem~\ref{thm:intro-covDim}, the regret dimension of \NaiveAlg is at most the covering dimension of the metric space. We ask: {\bf\em is it possible to achieve a better regret dimension}, perhaps using a more sophisticated algorithm? We show that this is indeed the case. Moreover, we provide an algorithm such that for any given metric space its regret dimension is arbitrarily close to optimal. Our main result as follows:

\begin{theorem} \label{thm:pmo}
Consider the \problem on a compact metric space $(X,\mD)$. Then for any $d > \MaxMinCOV(X)$ then there exists a bandit algorithm $\A$ whose regret dimension is at most $d$; moreover, the instance-specific regret dimension of $\A$ is at most the zooming dimension. No algorithm can have regret dimension strictly less than $\MaxMinCOV(X)$.
\end{theorem}

Here $\MaxMinCOV(X)$ is the \emph{max-min-covering dimension} which we defined in Section~\ref{sec:intro}. We show that $\MaxMinCOV(X)$ can be arbitrarily small compared to $\COV(X)$.

The rest of this section is organized as follows. The first two subsections are concerned with the lower bound: in Section~\ref{sec:PMO-LB} we develop a lower bound on regret dimension which relies on a certain ``tree of balls'' structure, and in Section~\ref{sec:pmo-MaxMinCOV} we derive the existence of this structure from the max-min-covering dimension. A lengthy KL-divergence argument (which is similar to prior work) is deferred to Section~\ref{sec:KL-divergence}. The next two subsections deal with an instructive special case: in Section~\ref{sec:fat-subset} we define a family of metric spaces for which $\MaxMinCOV(X)$ can be arbitrarily small compared to $\COV(X)$, and in Section~\ref{sec:warm-up} we design a version of the zooming algorithm tailored to such metric spaces. Finally, in Section~\ref{sec:fatness-PMO} we design and analyze an algorithm whose regret dimension is arbitrarily close to $\MaxMinCOV(X)$. We use the max-min-covering dimension to derive the existence of a certain decomposition of the metric space which we then take advantage of algorithmically. Our per-metric optimal algorithm builds on the machinery developed for the special case. Collectively, these results amount to Theorem~\ref{thm:pmo}.

\OMIT{
One can prove a matching lower bound: $\RegretDim(\A)\geq d$. In fact, this bound applies to any algorithm which uses a generic $k$-armed bandit algorithm on an evenly spaced sample of the strategy space, in the sense made precise in Section~\ref{sec:LB-naive}. In Section~\ref{sec:LB-naive} we prove a matching lower bound for the na\"ive algorithm.

\subsection{Lower bound for the na\"ive-type algorithms}
\label{sec:LB-naive}

We consider a more general version of the na\"ive algorithm described above, abstracting away the crucial idea of using a generic $k$-armed bandit algorithm on an evenly spaced sample of the strategy space.

\begin{definition}
A bandit algorithm is called \emph{uniformly exploring} if it proceeds in phases such that in the beginning of each phase $i$ it chooses $\delta_i>0$, a $\delta_i$-net $\mathcal{N}_i$ of the strategy space, and a $k$-armed bandit algorithm \A, $k=|\mathcal{N}_i|$
\end{definition}

We prove that the covering dimension is the best regret dimension achievable by a na\"ive algorithm.

\begin{theorem}\label{thm:naive-alg}
Consider the na\"ive algorithm for the Lipschitz MAB problem on a metric space of covering dimension $d$. The regret dimension of this algorithm is exactly $d$.
\end{theorem}

\begin{proofof}{\ref{thm:naive-alg}}
\bf{add proof!}
\end{proofof}
} 

\subsection{Lower bound on regret dimension}
\label{sec:PMO-LB}

It is known \citep{bandits-exp3} that a worst-case instance of the $K$-armed bandit problem consists of $K-1$ arms with identical payoff distributions, and one which is slightly better.  We refer to this as a ``needle-in-haystack'' instance.  Our lower bound relies on a \emph{multi-scale} needle-in-haystack instance in which there are $K$ disjoint open sets, and $K-1$ of them consist of arms with identical payoff distributions, but in the remaining open set there are arms whose payoff is slightly better.  Moreover, this special open set contains $K' \gg K$ disjoint subsets, only one of which contains arms superior to the others, and so on down through infinitely many levels of recursion.

In more precise terms, we require the existence of a certain structure: an infinitely deep rooted tree whose nodes correspond to balls in the metric space, so that for any parent ball $B$ the children balls are disjoint subsets of $B$.

\begin{definition}[ball-tree]\label{def:ball-tree}
Fix a metric space $(X,\mD)$. Let an \emph{\myball} be a pair $w=(x,r)$, where $x\in X$ is a ``center'' and $r\in (0,1]$ is a ``radius''.%
\footnote{Note that an open ball $B(x,r)$ denotes a subset of the metric space, so there can be distinct \myballs $(x,r)$ and $(x',r')$ such that $B(x,r) = B(x',r')$. We use \myballs to avoid this ambiguity.}
A \emph{ball-tree} is an infinite rooted tree where each node corresponds to an \myball. The following properties are required:
\begin{OneLiners}
\item all children of the same parent have the same radius, which is at most a quarter of the parent's.
\item if $(x,r)$ is a parent of $(x',r')$ then
    $\mD(x,x') + r' < r/2$.
\item if $(x, r_x)$ and $(y,r_y)$ are siblings, then
	$r_x + r_y < \mD(x,y)$.
\end{OneLiners}
The ball-tree has \emph{strength} $d\geq 0$ if each tree node with children of radius $r$ has at least $\max(2,r^{-d})$ children.
\end{definition}

Once there exists a ball-tree of strength $d$, we can show that, essentially, regret
    $O(t^{1-1/(d+2)})$
is the best possible. More precisely, we construct a probability distribution over problem instances which is hard for every given algorithm. Intuitively, this is the best possible ``shape'' of a regret bound since, obviously, a single problem instance cannot be hard for every algorithm.

\begin{lemma} \label{lm:pmo-LB-balltree}
Consider the \problem on a metric space $(X,\mD)$ such that there exists a ball-tree of strength $d\geq 0$. Assume 0-1 payoffs (i.e., the payoff of each arm is either $1$ or $0$). Then there exist a distribution $\mathcal{P}$ over problem instances $\mu$ and an absolute constant $C>0$ such that for any bandit algorithm $\A$ the following holds:
\begin{align}\label{eq:LB-stronger}
 \Pr_{\mu\in\mathcal{P}}
	\left[\; R_{(\A,\,\mu)}(t) \geq C\,t^{1-1/(d+2)}\; \text{for infinitely many $t$} \;\right] = 1.
\end{align}
It follows that the regret dimension of any algorithm is at least $d$.
\end{lemma}

\OMIT{ 
\begin{align}\label{eq:LB-stronger}
 (\forall C>0,\;\gamma<\tfrac{d+1}{d+2})\quad
\Pr_{\mu\in\mathcal{P}}
	\left[\; (\exists t_0\; \forall t\geq t_0)\;
		R_{(\A,\,\mu)}(t) \leq C\,t^\gamma \;\right] = 0.
\end{align}
} 

For our purposes, a weaker version of~\eqref{eq:LB-stronger} suffices: for any algorithm $\A$ there exists a payoff function $\mu$ such that the event in~\refeq{eq:LB-stronger} holds (which implies $\RegretDim(\A)\geq d$). In Section~\ref{sec:dichotomies} we will also use this lower bound for $d=0$. In the rest of this subsection we prove Lemma~\ref{lm:pmo-LB-balltree}.

\xhdr{Randomized problem instance.}
Given a metric space $(X,\mD)$ with a ball-tree, we construct a distribution $\mP$ over payoff functions as follows. For each tree node $w = (x_0, r_0)$ define the \emph{bump function}
	$F_w: X \rightarrow [0,1]$ by
\begin{align}\label{eq:bump-fn}
F_w(x) = \begin{cases}
	\min\{ r_0 - \mD(x,x_0),\, r_0/2  \} & \text{if $x\in B(x_0, r_0)$}, \\
	0	& \text{otherwise.}
\end{cases}
\end{align}
This function constitutes a ``bump'' supported on $B(x_0,r_0)$.

An \emph{end} in a ball-tree is an infinite path from the root:
	$\mathbf{w} = (w_0, w_1, w_2,\,\ldots)$.
Let us define the payoff function induced by each node $w = w_j$ as
\begin{align*}
	\mu_w = \frac13+ \frac13\,\sum_{i=1}^j F_{w_i},
\end{align*}
and the payoff function induced by the end $\mathbf{w}$ as
\begin{align*} 
\mu_\mathbf{w}
    := \lim_{j\to \infty} \mu_{w_j}
    = \frac13+ \frac13\,\sum_{i=1}^\infty F_{w_i}.
\end{align*}
Let $\mP$ be the distribution over payoff functions $\mu_\mathbf{w}$ in which the end $\mathbf{w}$ is sampled uniformly at random from the ball-tree (that is, $w_0$ is the root, and each subsequent node $w_{i+1}$ is sampled independently and uniformly at random among the children of $w_i$).

Let us show that $\mu_\mathbf{w}$ is a valid payoff function for the \problem. First, $\mu_\mathbf{w}(x)\in[0,1]$ for each arm $x\in X$, and the sum in the definition of $\mu_\mathbf{w}$ converges, because for each $i\geq 1$, letting $r_i$ be the radius of $w_i$, we have
    $r_i\leq r_1/4^i$
and
    $F_{w_i}(x)\in [0,\; r_i]$.
In fact, it is easy to see that the payoff function induced by any node or end in the ball-tree is bounded on $[\tfrac13,\tfrac23]$.
Second, $\mu_\mathbf{w}$ is Lipschitz on $(X,\mD)$ due to Lemma~\ref{lm:LB-Lipschitz}, which we state and prove in Section~\ref{sec:pmo-Lipschitz} so as not to break the flow.

The salient property of our construction is as follows.

\begin{lemma}\label{lm:PMO-LB-salient}
Consider a tree node $u$ in a ball-tree. Let
    $u_1 \LDOTS u_k$
be the children of $u$, and let $r$ be their radius. Let
    $B_1 \LDOTS B_k$
be the corresponding balls. Fix an arbitrary child $u_i$, and let $\mathbf{w}$ be an arbitrary end in the ball-tree such that $u_i\in \mathbf{w}$. Then:
\begin{OneLiners}
\item[(i)] $\mu_\mathbf{w}$ coincides with $\mu_u$ on all $B_\ell$, $\ell\neq i$.

\item[(ii)] $\sup(\mu_\mathbf{w}, B_i) - \sup(\mu_u,X) \geq r/6$.

\item[(iii)] $0\leq \mu_i-\mu_u \leq r/3$.
\end{OneLiners}
\end{lemma}

\begin{proof}
Let $\mathbf{w} = (w_0, w_1,\, \ldots)$, and let $j_0$ be the depth of $u_i$ in the ball-tree. Then $u=w_{j_0-1}$ and $u_i = w_{j_0}$, and
    $ \mu_{\mathbf{w}} = \mu_u + \tfrac13\, F_{u_i} +
        \tfrac13\, \sum_{j>j_0} F_{w_j}$.

Let $B$ be the ball corresponding to $u$. Observe that the balls corresponding to tree nodes $w_j$, $j>j_0$ are contained in $B_i$. It follows that on $B\setminus B_i$ all functions $F_{w_j}$, $j\geq j_0$ are identically $0$, and consequently $\mu_{\mathbf{w}} = \mu_u$. Since the balls $B_1 \LDOTS B_k$ are pairwise disjoint, this implies part (i).

For part (ii), let $(x^*,r^*)$ be the \myball corresponding to tree node $u$. Note that $\mu_u$ attains its supremum on $B(x^*,r^*/2)$, and in fact is a constant on that set. Also, recall that $B_i \subset B(x^*,r^*/2)$ by definition of the ball-tree. Let $x_i$ be the center of $u_i$. Observe that on $B(x_i,r/2)$ we have
    $F_{u_i} = r/2$,
and therefore
    $\mu_\mathbf{w} \geq \mu_u + r/6$.

For part (iii), note that $\mu_i-\mu_0 = \sum_{j\geq j_0} F_{w_j}$, and the latter is at most
    $\sum_{j\geq j_0} \tfrac12\, r/4^{j-j_0} \leq r/3$.
\end{proof}

Below we use Lemma~\ref{lm:PMO-LB-salient} to derive a lower bound on regret.

\xhdr{Regret lower bounds via $(\eps,k)$-ensembles.}
We make use of the lower-bounding technique from \citet{bandits-exp3} for the basic $k$-armed bandit problem. For a cleaner exposition, we encapsulate the usage of this technique in a theorem.  This theorem is considerably  more general than the original lower bound in \cite{bandits-exp3}, but the underlying idea and the proof are very similar. The theorem formulation (mainly, Definition~\ref{def:eps-k-ensemble} below) is new and may be of independent interest.

We use a very general MAB setting where the algorithm is given a strategy set $X$ and a collection $\F$ of feasible payoff functions; we call it the \emph{feasible MAB problem} on $(X,\F)$. In our construction, $\F$ consists of all functions
$\mu:X\to [0,1]$ that are Lipschitz with respect to the metric space.
The lower bound relies on the existence of a collection of subsets of $\F$ with certain properties, as defined below. These subsets correspond to children of a given tree node in the ball-tree (we give a precise connection after we state the definition).

\begin{definition}\label{def:eps-k-ensemble}
Let $X$ be the strategy set and $\F$ be the set of all feasible payoff functions. An \emph{$(\eps,k)$-ensemble} is
 a collection of subsets $\F_1 \LDOTS \F_k \subset \F$ such that there exist mutually disjoint subsets
        $S_1 \LDOTS S_k \subset X$ and a
function $\mu_0 : X \to [\tfrac13, \tfrac23]$
such that
for each $i=1 \ldots k$ and each function $\mu_i\in \F_i$ the following holds:
(i) $\mu_i \equiv \mu_0$ on each $S_\ell$, $\ell\neq i$, and
(ii) $\sup(\mu_i, S_i) - \sup(\mu_0,X) \geq \eps$, and
(iii) $0\leq \mu_i-\mu_0\leq 2\eps$ on $S_i$.
\end{definition}

For each tree node $u$, let
    $\F(u) = \{\mu_{\mathbf{w}}: u\in \mathbf{w} \}$
be the set of all payoff functions induced by ends $\mathbf{w}$ that contain $u$. By Lemma~\ref{lm:PMO-LB-salient}, if $u_1 \LDOTS u_k$ are siblings whose radius is $r$, then
    $(\F(u_1) \LDOTS \F(u_k))$
form an $(\tfrac{r}{6},k)$-ensemble.

\OMIT{In our application, subsets $S_1 \LDOTS S_k$ correspond to children $x_1 \LDOTS x_k$  of a given tree node in the ball-tree, and each $\F_i$ consists of payoff functions induced by the ends in the subtree rooted at $x_i$.}

\begin{theorem}\label{thm:LB-technique-MAB}
Consider the feasible MAB problem with 0-1 payoffs. Let  $\F_1, \ldots, \F_k$ be an $(\eps,k)$-ensemble,
where $k\geq 2$ and $\eps\in(0,\,\tfrac{1}{12})$. Then for any
    $t \leq \tfrac{1}{128}\, k\,\eps^{-2}$
and any bandit algorithm there exist at least $k/2$ distinct $i$'s such that the regret of this algorithm on any payoff function from $\F_i$ is at least $\tfrac{1}{60}\,\eps t$.
\end{theorem}

The idea is that if the payoff function $\mu$ lies in $\cup_i\, \F_i$, an algorithm needs to play arms in $S_i$ for at least $\Omega(\eps^{-2})$ rounds in order to determine whether $\mu$ lies in a given $\F_i$, and each such round incurs regret at $\eps$ (or more) if $\mu\not\in\F_i$.

\citet{bandits-exp3} analyzed a special case in which there are $k$ arms
	$x_1 \LDOTS x_k$,
and each $\F_i$ consists of a single payoff function that assigns expected payoff $\tfrac12+\eps$ to arm $x_i$, and $\tfrac12$ to all other arms. To preserve the flow of the paper, the proof of Theorem~\ref{thm:LB-technique-MAB} is presented in Appendix~\ref{sec:KL-divergence}.

\xhdr{Regret analysis.}
Let us fix a bandit algorithm \A, and let $T$ the ball-tree of strength $d\geq 0$. Without loss of generality, let us assume that each tree node in $T$ has finitely many children. Recall that each end of $T$ induces a payoff function. Let
    $\F_T$
be the set of all payoff functions induced by the ends of $T$. Throughout, the constants in $\Omega(\cdot)$ are absolute.

Consider a level-$j$ tree node $u$ in $T$. Let
    $u_1 \LDOTS u_k$
be the children of $u$, and let $r$ be their radius. Recall that $k\geq \max(2,\, r^{-d})$. Then
$\{\F(u_i): 1\leq i \leq k \}$ is a $(\tfrac{r}{6},k)$-ensemble.
By Theorem~\ref{thm:LB-technique-MAB} there exist a subset $I_u \subset \{1 \LDOTS k\}$ and a time
$t>\Omega(k\, r^{-2})$ such that for any payoff function $\mu\in \F(u_i)$, $i\in I_u$ we have
	$R_{(\A,\,\mu)}(t) \geq \Omega(rt)$.
Plugging in $k\geq r^{-d}$ and $r\leq 4^{-j}$, we see that there exists a time $t\geq 2^{\Omega(j)}$ such that for each payoff function $\mu\in \F(u_i)$, $i\in I_u$ we have
    $R_{(\A,\,\mu)}(t) \geq \Omega(t^{1-1/(d+2)})$.

Consider the distribution $\PP$ over payoff functions $\F_T$ from our construction. Let $\mE_j$ be the event that $\mu\in \F(u_i)$, $i\in I_u$ for some level-$j$ tree node $u$. If $\mu\in \mE_j$ for infinitely many $j$'s then
    $R_{(\A,\,\mu)}(t) \geq \Omega\left( t^{1-1/(d+2)} \right)$
for infinitely many times $t$. We complete the proof of Lemma~\ref{lm:pmo-LB-balltree} by showing that
\begin{align}\label{eq:LB-pmo-analysis}
    \Pr_{\mu\sim \PP} \left[\; \mu\in\mE_j\; \text{for infinitely many $j$} \;\right] = 1.
\end{align}

The proof of~\eqref{eq:LB-pmo-analysis} is similar to the proof of the Borel-Cantelli Lemma. If $\mu\in\mE_j$ only for finitely many $j$'s, then
    $\mu\in \cap_{j\geq j_0}\; \neg \mE_j$ for some $j_0$.
Fix some $j_0\in \N$, and let us show that
    $\Pr[\cap_{j\geq j_0}\; \neg \mE_j] =0$.
For each tree node $u$ at level $j$ we have
    $\Pr[ \mE_j \,|\, \mu \in \F(u)] \geq \tfrac12$.
It follows that for any $j>j_0$
    $$ \Pr\left[ \neg \mE_j \,|\, \neg\mE_{j_0} \LDOTS \neg\mE_{j-1} \right] \leq \tfrac12.$$
Therefore
    $ \Pr[\cap_{j\geq j_0}\; \neg \mE_j]
        = \Pr[ \neg\mE_{j_0} ]\times \prod_{j>j_0} \Pr\left[ \neg \mE_j \,|\, \neg\mE_{j_0} \LDOTS \neg\mE_{j-1} \right] =0$,
claim proved.

\subsubsection{Lipschitz-continuity of the lower-bounding construction}
\label{sec:pmo-Lipschitz}

In this subsection we prove that $\mu_\mathbf{w}$ is Lipschitz function on $(X,\mD)$, for any end $\mathbf{w}$ of the ball-tree. In fact, we state and prove a more general lemma in which the bump functions are summed over all tree nodes with arbitrary weights in $[-1,1]$. This lemma will also be used for the lower-bounding constructions in Section~\ref{sec:lower-bound} and Section~\ref{sec:MaxMinLCD-LB}.

\begin{lemma}\label{lm:LB-Lipschitz}
Consider a ball-tree on a metric space $(X,\mD)$. Let $V$ be the set of all tree nodes. For any given weight vector $\sigma:V\to [-1,1]$ and an absolute constant $c_0\in[0,\tfrac12]$ define the payoff function
\begin{align*}
	\mu_\sigma =c_0+\frac13\,\sum_{w\in V} \sigma(w)\cdot F_w,
\end{align*}
where $F_w$ is the bump function from \refeq{eq:bump-fn}. Then $\mu_\sigma$
is Lipschitz on $(X,\mD)$.
\end{lemma}

In the rest of this subsection we prove Lemma~\ref{lm:LB-Lipschitz}. (The proof for the special case $\mu_\sigma = \mu_\mathbf{w}$ uses essentially the same ideas and does not get much simpler.)

\newcommand{\parent}{\vdash} 
\newcommand{\child}{\dashv}
\newcommand{\ancestor}{\succ} 
\newcommand{\descend}{\prec}
\newcommand{\ancestoreq}{\succeq} 
\newcommand{\descendeq}{\preceq}

Let us specify some notation. Throughout, $u,v,w$ denote tree nodes.
Write $u\parent w$ if $u$ is a parent of $w$, and $u\ancestor w$ if $u$ is an ancestor of $w$. (Accordingly, define $\child$ and $\descend$ relations.) Generally our convention will be that
    $u\ancestor w\ancestor v$.
 Let $x_w$ and $r_w$ be, resp., the center and the radius of $w$, and let $B_w = B(x_w,r_w)$ denote the corresponding ball. Fix the weight vector $\sigma:V\to [-1,1]$ and arms $x,y\in X$. Write $\mu = \mu_\sigma$ for brevity.
We need to prove that $|\mu(x)-\mu(y)|\leq \mD(x,y)$.

We start with some observations about the bump functions. First, $F_w(y)\leq \mD(x,y)$ under appropriate conditions.

\begin{claim}\label{cl:pmo-lipschitz-5}
If $y\in B_w$ and $x\not\in B_w$, then $F_w(y) \leq \mD(x,y)$.
\end{claim}
\begin{proof}
Observe that
\begin{align*}
F_w(y)
    &\leq r_w-\mD(x_w,y)            & \text{(because $y\in B_w$)} \\
    &\leq \mD(x_w,x) - \mD(x_w,y)   & \text{(because $x\not\in B_w$)} \\
    &\leq \mD(x,y)                  & \text{(by triangle inequality)} \qquad \qedhere
\end{align*}
\end{proof}

Second, each bump function $F_w$ is Lipschitz.

\begin{claim}\label{cl:pmo-lipschitz-6}
$|F_w(x)-F_w(y)| \leq \mD(x,y)$.
\end{claim}
\begin{proof}
If $x,y\not\in B_w$, then $F_w(x)=F_w(y)=0$, and we are done. If $x\not\in B_w$, but $y\in B_w$, then $F_w(x)=0$ and $F_w(y)\leq \mD(x,y)$ by Claim~\ref{cl:pmo-lipschitz-5}, and we are done. In what follows, assume $x,y\in B_w$.

We consider four cases, depending on whether $\mD(x,x_w)$ and $\mD(y,x_w)$ are larger than $r_w/2$. If both $\mD(x,x_w)$ and $\mD(y,x_w)$ are at least $r_w/2$, then
    $F_w(x) = F_w(y)=r_w/2$, and we are done.
If both $\mD(x,x_w)$ and $\mD(y,x_w)$ are at most $r_w/2$, then
    $F_w(x)-F_w(y) = \mD(y,x_w)-\mD(x,x_w)$,
which is at most $\mD(x,y)$ by triangle inequality, and we are done.
If
    $\mD(x,x_w) \leq r_w/2 \leq \mD(y,x_w)$,
then
\begin{align*}
F_w(x) - F_w(y)
    &= r_w/2 - (r_w - \mD(y,x_w))
    = \mD(y,x_w)-r_w/2 \\
    &\leq \mD(y,x_w) - \mD(x,x_w) \leq \mD(x,y).
\end{align*}
The fourth case is treated similarly. $\qquad\qedhere$
\end{proof}

Third, we give a convenient upper bound for $F_w(y)-F_w(x)$, assuming $x\in B_w$.

\begin{claim}\label{cl:pmo-lipschitz-2}
Assume $x\in B_w$. Then
    $F_w(y)-F_w(x) \leq \max(0,\, \mD(x,x_w)-r_w/2)$.
\end{claim}
\begin{proof}
This is because $F_w(y)\leq r_w/2$ and
    $F_w(x) = \min(r_w/2, r_w-\mD(x,x_w))$.
$\qquad\qedhere$
\end{proof}

Some of the key arguments are encapsulated below. First, $F_u$ is constant on $B_w$, $u\ancestor w$.

\begin{claim}\label{cl:pmo-lipschitz-1}
$F_u(x) = r_u/2$ whenever $x\in B_w$ and $u\ancestor w$.
\end{claim}
\begin{proof}
Since $B_w\supset B_v$ whenever $w \ancestor v$, it suffices to assume that $u$ is a parent of $w$. Then
\begin{align*}
\mD(x,x_u)
    &\leq \mD(x_w,x_u) + \mD(x,x_w)  & \text{(by triangle inequality)} \\
    &< \mD(x_w,x_u) + r_w            & \text{(since $x\in B_w$)} \\
    &< r_u/2                        & \text{(by definition of ball-tree)}. \qquad\qedhere
\end{align*}
\end{proof}

Second, suppose $B_w$ separates $x$ and $y$ (in the sense that $B_w$ contains $y$ but not $x$), and $\mD(x,y)$ is small compared to $r_w$. We show that $y\not\in B_v$, $w\parent v$.

\begin{claim}\label{cl:pmo-lipschitz-3}
Assume $y\in B_w$ and $x\not\in B_w$. Then $y\not\in B_v$, whenever $w\parent v$ and $\mD(x,y)\leq r_w/2$.
\end{claim}
\begin{proof}
Observe that
\begin{align*}
\mD(x_v,y)
    &\geq \mD(x_w,y) - \mD(x_w,x_v) & \text{(by triangle inequality)} \\
    &\geq \mD(x_w,x)- \mD(x,y) - \mD(x_w,x_v) & \text{(by triangle inequality)} \\
    &\geq r_w- \mD(x,y) - \mD(x_w,x_v) & \text{(because $x\not\in B_w$ )} \\
    &\geq r_v + r_w/2 - \mD(x,y).
\end{align*}
The last inequality follows because
    $r_w/2 - \mD(x_w,x_v) \geq r_v$
by definition of ball-tree. It follows that
    $\mD(x_v,y) \geq r_v$
whenever $\mD(x,y)\leq r_w/2$.
\end{proof}

\begin{claim}\label{cl:pmo-lipschitz-4}
Assume $w\parent v$ and $x\in B_w\setminus B_v$ and $y\in B_v$. Then
    $ F_w(y) - F_w(x) + F_v(y) \leq \mD(x,y)$.
\end{claim}
\begin{proof}
If $F_w(y)\leq F_w(x)$ then it suffices to observe that
    $F_v(y) \leq \mD(x,y)$
by Claim~\ref{cl:pmo-lipschitz-5}. From here on, assume
    $F_w(y)>F_w(x)\geq 0$.
Observe that
\begin{align*}
F_w(y) - F_w(x)
    &\leq \mD(x,x_w) - r_w/2 & \text{(by Claim~\ref{cl:pmo-lipschitz-2})} \\
F_v(y)
    &\leq r_v - \mD(x_v,y) & \text{(by definition of $F_v$)} \\
0 &\leq r_w/2 - r_v - \mD(x_w,x_v) & \text{(by definition of ball-tree)}.
\end{align*}
Summing this up,
\begin{align*}
F_w(y) - F_w(x) + F_v(y)
    &\leq \mD(x,x_w) - \mD(x_w,x_v) - \mD(x_v,y) &\\
    &\leq \mD(x,x_w) - \mD(x_w,y)    & \text{(by triangle inequality)} \\
    &\leq \mD(x,y)                   & \text{(by triangle inequality)} \qquad\qedhere
\end{align*}
\end{proof}

Now we are ready to put the pieces together. Let $w$ be the least common ancestor of $x$ and $y$ in the ball-tree, i.e., the smallest tree node $w$ such that $x,y\in B_w$.
(Such $w$ exists because the radii of the tree nodes go to zero along any end of the ball-tree.) Observe that:
\begin{OneLiners}
\item $F_u(x)=F_u(y)=r_u/2$ for all $u\ancestor w$ (by Claim~\ref{cl:pmo-lipschitz-1}).

\item $F_{w'}(x) = F_{w'}(y)=0$ for all tree nodes $w'$ incomparable with $w$, because $x,y\not\in B_{w'}$.
\end{OneLiners}
Therefore,
\begin{align}\label{eq:pmo-lipschitz-pf-0}
\mu(y)-\mu(x)
    =   \frac13 \left(
        \sum_{v\descendeq w} \sigma(v)\,F_v(x)
    \right).
\end{align}

Let $w_x$ (resp., $w_y$) be the unique child containing $x$ (resp., $y$) if such child exists, and an arbitrary child of $w$ otherwise. By minimality of $w$ and the fact that $w$ has at least two children, we can pick $w_x$ and $w_y$ so that they are distinct. Then:
\begin{OneLiners}
\item $F_v(x)=0$ for all tree nodes $v \descendeq w_y$, because $x\not\in B_v$.

\item $F_v(y)=0$ for all tree nodes $v \descendeq w_x$, because $y\not\in B_v$.
\end{OneLiners}
Plugging these observations into \refeq{eq:pmo-lipschitz-pf-0}, we obtain:
\begin{align}
\mu(y)-\mu(x)
    &=   \frac13 \left(
        \sigma(w) \left(F_w(y)-F_w(x)\right) +
        \sum_{v\descendeq w_x} \sigma(v)\,F_v(x) +
        \sum_{v\descendeq w_y} \sigma(v)\,F_v(y)
    \right). \nonumber \\
|\mu(y)-\mu(x)|
    &\leq  \frac13 \left(
        \left|F_w(y)-F_w(x)\right| +
        \sum_{v\descendeq w_x} F_v(x) +
        \sum_{v\descendeq w_y} F_v(y).
    \right). \label{eq:pmo-lipschitz-pf-1}
\end{align}
Note that \refeq{eq:pmo-lipschitz-pf-1} no longer depends on the weight vector $\sigma$.

To complete the proof, it suffices to show the following:
\begin{align}
\left|F_w(y)-F_w(x)\right| +
        \sum_{v\descendeq w_x} F_v(x)
        &\leq \tfrac43\, \mD(x,y). \label{eq:pmo-lipschitz-pf-2} \\
\left|F_w(y)-F_w(x)\right| +
        \sum_{v\descendeq w_y} F_v(y)
        &\leq \tfrac43\, \mD(x,y). \label{eq:pmo-lipschitz-pf-3}
\end{align}

In the remainder of the proof we show \eqref{eq:pmo-lipschitz-pf-3} (and \eqref{eq:pmo-lipschitz-pf-2} follows similarly). Let $\Gamma$ denote the left-hand side of \eqref{eq:pmo-lipschitz-pf-3}. If $y\not\in B_{w_y}$, then $F_v(y) = 0$ for all $v\descendeq w_y$, and $\Gamma\leq \mD(x,y)$
by Claim~\ref{cl:pmo-lipschitz-6}. From here on, assume $y\in B_{w_y}$. Then by Claim~\ref{cl:pmo-lipschitz-1} we have
    $F_w(y) = r_w/2$.
It follows that
\begin{align} \label{eq:pmo-lipschitz-pf-4}
\Gamma
    &= F_w(y)-F_w(x) + \sum_{v\descendeq w_y} F_v(y)
    \leq \mD(x,y) + \sum_{v\descend w_y} F_v(y),
 \end{align}
where the last inequality follows from Claim~\ref{cl:pmo-lipschitz-4}.

Now consider two cases, depending on whether $\mD(x,y)\leq r_{w_y}/2$. If so, then $y\not\in B_v$ for any $v\descend w_y$ by Claim~\ref{cl:pmo-lipschitz-3}, and therefore $F_v(y)=0$ for all such $v$ and we are done.

The remaining case is that $\mD(x,y)> r_{w_y}/2$. Define the sequence of tree nodes $(v_j: j\in\N)$ inductively by $v_0 = w_y$ and for each $j\in \N$ letting $v_{j+1}$ be the child of $v_j$ that contains $y$, if such child exists, and any child of $v_j$ otherwise. Then
\begin{align*}
F_{v_j}(y)
    &\leq r_{v_j}/2
     \leq 4^{-j}\, r_{v_0}/2
     < 4^{-j}\, \mD(x,y) \qquad \forall j\geq 1 \\
 \sum_{v\descend w_y} F_v(y)
    &= \sum_{j=1}^\infty F_{v_j}(y)
    \leq \sum_{j=1}^\infty 4^{-j}\, \mD(x,y)
    = \mD(x,y)/3.
\end{align*}
Plugging this into \refeq{eq:pmo-lipschitz-pf-4} completes the proof of \refeq{eq:pmo-lipschitz-pf-3}, which in turn completes the proof of Lemma~\ref{lm:LB-Lipschitz}.

\subsection{The max-min-covering dimension}
\label{sec:pmo-MaxMinCOV}

We would like to derive the existence of a strength-$d$ ball-tree using a covering property similar to the covering dimension. We need a more nuanced notion, which we call the max-min-covering dimension, to ensure that \emph{each} of the open sets arising in the construction of the ball-tree has sufficiently many disjoint subsets to continue to the next level of recursion.%
\footnote{%
We've defined this notion while stating our results in Section~\ref{subsec:intro-poly}. Here we restate if for the sake of convenience.}
Further, this new notion is an intermediary that connects our lower bound with the upper bound that we develop in the forthcoming subsections.

\begin{definition} \label{def:mincov}
For a metric space $(X,\mD)$ and subsets $Y\subseteq X$ we define
\begin{align*}
\MinCOV(Y)    &= \inf \{\COV(U):
    \text{$U\subset Y$ is non-empty and open in $(Y,\mD)$} \},\\
\MaxMinCOV(X) &= \sup \{ \MinCOV(Y) \,:\, Y \subseteq X \}.
\end{align*}
We call them the \emph{min-covering dimension} and the \emph{max-min-covering dimension} of $X$, respectively.
\end{definition}

The infimum over open $U \subseteq Y$ in the
definition of min-covering dimension ensures that every
open set which may arise in the needle-in-haystack
construction described above will contain $\Omega(\delta^{\eps-d})$
disjoint $\delta$-balls
for some sufficiently small positive $\delta, \eps$.
Constructing lower bounds for Lipschitz MAB algorithms
in a metric space $X$ only requires  that $X$ should have \emph{subsets}
with large min-covering dimension, which explains the
supremum over subsets in the definition of
max-min-covering dimension. Note that for every subset $Y\subseteq X$ we have
    $\MinCOV(Y) \leq \COV(Y)\leq \COV(X)$,
which implies $\MaxMinCOV(X)\leq \COV(X)$.

\begin{lemma} \label{lm:pmo-LM-MaxMinCOV}
Consider the \problem on a metric space $(X,\mD)$ and $\MaxMinCOV(X)>0$. Then for any $d\in (0,\MaxMinCOV(X))$ there exists a ball-tree of strength $d$. It follows (using Lemma~\ref{lm:pmo-LB-balltree}) that the regret dimension of any algorithm is at least $\MaxMinCOV(X)$.
\end{lemma}

Recall from Section~\ref{sec:dim-notions} that an \emph{$r$-packing} of a metric space $(X,\mD)$ be a subset $S\subset X$ such that any two points in $S$ are at distance at least $r$ from one another. The proof will use the following simple packing lemma.

\begin{lemma}[Folklore] \label{lem:packing-folklore}
Suppose $(X,\mD)$ is a metric space of covering dimension $d$. Then for any $b<d$, $r_0>0$ and $C>0$ there exists
$r \in (0,r_0)$ such that $X$ contains an $r$-packing of size at least $C\,r^{-b}$.
\end{lemma}

\begin{proof}
Let $r < r_0$ be a positive number such that every covering of $(X,\mD)$ with radius-$r$ balls requires more than $C\,r^{-b}$ balls.  Such an $r$ exists, because the covering dimension of $(X,\mD)$ is strictly greater than $b$.

Now let $S$ be any maximal $r$-packing in $(X,\mD)$. For every $x \in X$ there must exist some point $y\in S$ such that $\mD(x,y)<r$, as otherwise $S\cup \{x\}$ would be an $r$-packing, contradicting the maximality of $S$. Therefore, balls $B(x,r)$, $x\in S$ cover the metric space. It follows that $|S| \geq C\,r^{-b}$ as desired.
\end{proof}

\begin{proofof}{Lemma~\ref{lm:pmo-LM-MaxMinCOV}}
Pick $c\in (d,\MaxMinCOV(X))$. Choose $Y\subset X$ such that $\MinCOV(Y) \geq c$.

Let us recursively construct a ball-tree of strength $d$. Each tree node will correspond to an \myball centered in $Y$. Define the root to be some radius-$1$ \myball with center in $Y$.%
\footnote{Recall from Definition~\ref{def:ball-tree} that \myball is a pair $(x,r)$ where $x\in X$ is the ``center'' and $r\in (0,1]$ is the ``radius''.}
Now suppose we have defined a tree node $w$ which corresponds to an \myball with center $y\in Y$ and radius $r$. Let us consider the set
    $B = Y\cap B(y,\tfrac{r}{4})$.
Then $B$ is non-empty and open in the metric space $(Y,\mD)$. By definition of the min-covering dimension we have
    $\COV(B)\geq c$.
Now Lemma~\ref{lem:packing-folklore} guarantees the existence of a $(2r')$-packing $S\subset B$ such that $r'<r/4$ and $|S|\geq (r')^{-d}$. Let the children of $w$ correspond to points in $S$, so that for each $x\in S$ there is a child with center $x$ and radius $r'$.
\end{proofof}

\subsection{Special case: metric space with a ``fat subset''}
\label{sec:fat-subset}

To gain intuition on the max-min-covering dimension, let us present a family of metric spaces where for a given covering dimension the max-min-covering dimension can be arbitrarily small.

Let us start with two concrete examples. Both examples involve an infinite rooted tree where the out-degree is low for most nodes and very high for a few. On every level of the tree the high-degree nodes produce exponentially more children than the low-degree nodes. For concreteness, let us say that all low-degree nodes have degree $2$, all high-degree nodes on a given level of the tree have the same degree, and this degree is such that the tree contains $4^i$ nodes on every level $i$. The two examples are as follows:
\begin{itemize}
\item one high-degree node on every level; the high-degree nodes form a path, called the \emph{fat end}.

\item $2^i$ high-degree nodes on every level $i$; the high-degree nodes form a binary tree, called the \emph{fat subtree}.
\end{itemize}
We assign a \emph{width} of $2^{-i/d}$, for
some constant $d>0$, to each level-$i$ node; this is
the diameter of the set of points contained in the corresponding
subtree.  The tree induces a metric space $(X,\mD)$ where $X$ is the set of all
ends\footnote{Recall from Section~\ref{sec:PMO-LB} that
an end of an infinite rooted tree is an infinite path starting
at the root.},
and for $x,y \in X$ we define
$\mD(x,y)$ to be the width of the least common ancestor of ends
$x$ and $y$.

In both examples the covering dimension of the entire metric space is $2d$, whereas there exists a low-dimensional ``fat subset'' --- the fat end or the fat subtree --- which is, in some sense, responsible for the high covering dimension of $X$. Specifically, for any subtree $U$ containing the fat end (which is just a point in the metric space) it holds that $\COV(U)=2d$ but
    $\COV(X\setminus U)=d$.
Similarly, if $S$ is the fat subtree and $U$ is a union of subtrees that cover $S$ then $\COV(U)=2d$ but $\COV(S\cup (X\setminus U))=d$.

It is easy to generalize the notion of ``fat subtree'' to an arbitrary metric space:

\begin{definition}
Given a metric space $(X,\mD)$, a closed subset $S\subset X$ is called \emph{$d$-fat}, $d<\COV(X)$, if $\COV(S)\leq d$ and $\COV(X\setminus U) \leq d$ for any open neighborhood $U$ of $S$.
\end{definition}

\OMIT{ 
Note that any open neighborhood $U$ of a $d$-fat subset $S$ has covering dimension $\COV(X)$, whereas cutting out $U\setminus S$ reduces the covering dimension by a positive amount that is bounded away from zero,
irrespective of the choice of $U$.
}
In both examples above, the max-min-covering dimension is $d$. This is because  every point outside the fat subset has an open neighborhood whose covering dimension is at most $d$ (and the covering dimension of the fat subset itself is at most $d$, too). We formalize this argument as follows:

\begin{claim}\label{cl:fat-subset-MaxMinCOV}
Suppose metric space $(X,\mD)$ contains a subset $S\subset X$ of covering dimension at most $d$ such that every $x\in X\setminus S$ has an open neighborhood $N_x$ of covering dimension at most $d$. Then $\MaxMinCOV(X)\leq d$.
\end{claim}

\begin{proof}
Equivalently, we need to show that $\MinCOV(Y)\leq d$ for any subset $Y\subset X$.

Fix $Y\subset X$. For each $\eps>0$ we need to produce a non-empty subset $U\subset Y$ such that $U$ is open in $(Y,\mD)$ and its covering dimension is at most $d+\eps$. If $Y\subset S$ we can simply take $U=Y$, because $\COV(Y)\leq \COV(S)\leq d $. Now suppose there exists a point $x\in Y\setminus S$. Then $U = N_x \cap Y$ is non-empty and open in the metric space restricted to $Y$, and
    $\COV(U)\leq \COV(N_x)\leq d$.
\end{proof}

In fact, this property applies to any $d$-fat subset in a compact metric space.

\begin{lemma}\label{lm:fat-subset-MaxMinCOV}
Suppose a compact metric space $(X,\mD)$ contains a $d$-fat subset $S\subset X$. Then 
\begin{OneLiners}
\item[(a)] every $x\in X\setminus S$ has an open neighborhood $N_x$ of covering dimension at most $d$.
\item[(b)] $\MaxMinCOV(X)\leq d$.
\end{OneLiners}
\end{lemma}

\begin{proof}
To prove part (a), consider some point $x\in X\setminus S$. Since $S$ is closed, $\mD(x,S)>0$. Denoting $r=\tfrac14\,\mD(x,S)$, let $U$ be the union of all radius-$r$ open balls centered in $S$. Then $U$ is an open set containing $S$, so $\COV(X\setminus U)\leq d$. Since $B(x,r) \subset X\setminus U$, its covering dimension is at most $d$, too. 

Part (b) follows from part(a) by Claim~\ref{cl:fat-subset-MaxMinCOV}, using the fact that $\COV(S)\leq d$.
\end{proof}

\subsection{Warm-up: taking advantage of fat subsets}
\label{sec:warm-up}

\OMIT{ 
This far, we know that for every metric space $(X,\mD)$ there exists a bandit
algorithm whose regret dimension is $\COV(X)$, and no bandit algorithm can have regret dimension less than $\MaxMinCOV(X)$. When $\MaxMinCOV(X) <\COV(X)$, which of these two bounds is correct, or can they both be wrong? }

As a warm-up for our general algorithmic result, let us consider metric spaces with $d^*$-fat subsets, and design a modification of the zooming algorithm whose regret dimension can be arbitrarily close to $d^*$ for such metric spaces. In particular, we establish that $\COV(X)$ is, in general, not an optimal regret dimension. Further, our algorithm essentially retains the instance-specific guarantee with respect to the zooming dimension. As a by-product, we develop much of the technology needed for the general result in the next subsection.

The zooming algorithm from Section~\ref{sec:adaptive-exploration} may perform poorly on metric spaces with a fat subset $S$ if the optimal arm $x^*$ is located inside $S$. This is because as the confidence ball containing $x^*$ shrinks, it may be too burdensome to keep covering%
\footnote{Recall that an arm $x$ is called \emph{covered} at time $t$ if for some active arm $y$ we have $\mD(x,y) \leq r_t(y)$.}
the profusion of arms located near $x^*$, in the sense that it may require activating too many arms.
We fix this problem by imposing \emph{quotas} on the number of active arms. Thus some arms may not be covered. However, we show that (for a sufficiently long phase, with very high probability) there exists an optimal arm that is covered, which suffices for the technique in Section~\ref{subsec:zooming-analysis} to produce the desired regret bound.

We define the quotas as follows. For a phase of duration $T$ and a fixed $d>d^*$, the quotas are
\begin{align*}
\forall\, Y\in \{ X\setminus S,\, S\}: \quad
|\{ \emph{active arms $x \in Y$}:\, r_t(x) \geq \rho \}| \leq \rho^{-d},
\quad \rho= T^{- 1/(d+2)}.
\end{align*}
\noindent We use a generic modification of the zooming algorithm where the activation rule only considers arms that, if activated, do not violate any of the given quotas.

\newcommand{\Eligible}{\mathtt{EligibleArms}}

\begin{algorithm}[h]
\caption{(zooming algorithm with quotas)}
\label{alg:zooming-fat}

\begin{algorithmic}
\FOR{phase $i=1,2,3, \ldots$}
\STATE Initially, no arms are active.
\FOR{round $t=1,2,3, \ldots, 2^i$}
\STATE
    $\Eligible = \{ \text{arms $x\in X$: activating $x$ does not violate any quotas}\}$.
\STATE \emph{Activation rule:} if some arm $x\in \Eligible$ is not covered,
\STATE \hspace{5mm}   pick any such arm and activate it.
\STATE \emph{Selection rule:} play any active arm with the maximal index~\refeq{eq:index}.
\ENDFOR
\ENDFOR
\end{algorithmic}
\end{algorithm}

To pave the way for a generalization, consider a sequence of sets $(S_0,S_1,S_2) = (X,S,\emptyset)$, and let $k=1$ be the number of non-trivial sets in this sequence. Our algorithm and analysis easily generalize to a sequence of closed subsets
    $S_0 \LDOTS S_{k+1} \subset X$, $k\geq 1$
which satisfies the following properties:
\begin{OneLiners}
\item $X = S_0 \supset S_1 \supset \ldots \supset S_k \supset S_{k+1} = \emptyset $; the sequence is strictly decreasing.
\item for every $i\in \{0\LDOTS k\}$ and any open subset $U\subset X$ that contains $S_{i+1}$, it holds that $\COV(S_i\setminus U) \leq d$.~%
\footnote{For $i=k$, this condition is equivalent to $\COV(S_k)\leq d$.}
\end{OneLiners}
We call such sequence a \emph{$d$-fatness decomposition} of length $k$.

We generalize the quotas in an obvious way: for a phase of duration $T$ and a fixed $d>d^*$,
\begin{align}
&\forall i\in\{0 \LDOTS k\}: \nonumber\\
&\qquad|\{ \emph{active arms $x \in S_i\setminus S_{i+1}$}:\, r_t(x) \geq \rho \}| \leq \rho^{-d},
\quad \rho= T^{- 1/(d+2)}. \label{eq:quota-fatSet}
\end{align}
This completes the specification of the algorithm. The invariant~\refeq{eq:quota-fatSet} holds for each round; this is because during a given phase the confidence radius of every active arm does not increase over time.

\begin{remark}
Essentially, the algorithm ``knows" the decomposition 
    $S_0 \LDOTS S_{k+1}$
and the parameter $d$.
To implement the algorithm, it suffices to use the decomposition via a covering oracle for each subset
    $S_{i+1}\setminus S_i$, $i\in \{0 \LDOTS k\}$.
Here a covering oracle for subset $Y\subset X$ 
takes a finite collection of open balls, where each ball is represented as a (center, radius) pair, and either declares that these balls cover $Y$ or outputs an uncovered point.
\end{remark}

\begin{theorem}\label{lm:zooming-fatSet}
Consider the \problem on a compact metric space $(X,\mD)$ which contains a $d^*$-fatness decomposition $(S_0 \LDOTS S_{k+1})$ of finite length $k \geq 1$. Let $\A$ be the zooming algorithm (Algorithm~\ref{alg:zooming-fat}) with quotas~\refeq{eq:quota-fatSet}, for some known parameter $d>d^*$. Then the regret dimension of $\A$ is at most $d$. Moreover, the instance-specific regret dimension of $\A$ is bounded from above by the zooming dimension.
\end{theorem}

The remainder of this section presents the full proof of
Theorem~\ref{lm:zooming-fatSet}. The following rough outline of the proof
may serve as a useful guide.

\begin{proof}[Proof Outline]
For simplicity, assume there is a unique optimal arm, and call it $x^*$. The desired regret bounds follow from the analysis in Section~\ref{subsec:zooming-analysis} as long as $x^*$ is covered w.h.p.\ throughout any sufficiently long phase. All arms in $S=S_k$ are covered eventually because $\COV(S)<d$, so if $x^*\in S$ then we are done. If $x^*\not\in S$ then pick the largest $\ell$ such that $x^*\in S_{\ell}\setminus S_{\ell+1}$. Then there is some $\eps>0$ such that all arms in $S_{\ell+1}$ are suboptimal by at least $\eps$. Letting $U$ be an $\frac{\eps}{2}$-neighborhood of $S_{\ell+1}$, note that each arm in $U$ is suboptimal by at least $\frac{\eps}{2}$. It follows that (w.h.p.)  the algorithm cannot activate too many arms in $U$. On the other hand, $S_\ell \setminus U$ has a low covering dimension, so (w.h.p.) the algorithm cannot activate too many arms in $S_{\ell} \setminus U$, either. It follows that (for a sufficiently long phase, w.h.p.) the algorithm stays within the quota, in which case $x^*$ is covered.
\end{proof}

To prove Theorem~\ref{lm:zooming-fatSet}, we incorporate the analysis from Section~\ref{subsec:zooming-analysis} as follows. We state a general lemma which applies to the zooming algorithm with a modified activation rule (and no other changes). We assume that the new activation rule is at least as selective as the original one: an arm is activated only if it is not covered, and at most one arm is activated in a given round. We call such algorithms \emph{zooming-compatible}.

A phase of a zooming-compatible algorithm is called \emph{clean} if the property in Claim~\ref{cl:conf-rad} holds for each round in this phase. Claim~\ref{cl:conf-rad} carries over: each phase $i$ is clean with probability at least $1-4^{-i}$. Let us say that a given round in the execution of the algorithm is \emph{well-covered} if after the activation step in this round some optimal arm is covered. (We focus on compact metric spaces, so the supremum $\mu^* = \sup(\mu,X)$ is achieved by some arm; we will call such arm \emph{optimal}.) A phase of the algorithm is called well-covered if all rounds in this phase are well-covered.

An algorithm is called \emph{$(k,d)$-constrained} if in every round $t$ it holds that
\begin{align*}
|\{ \emph{active arms $x \in X$}:\, r_t(x) \geq \rho \}| \leq (k+1) \,\rho^{-d},
\quad \rho= T^{- 1/(d+2)},
\end{align*}
where $T$ is the duration of the current phase. Note that the zooming algorithm with quotas~\refeq{eq:quota-fatSet} is $(k,d)$-constrained by design.

\begin{lemma}[immediate from Section~\ref{subsec:zooming-analysis}]
\label{lm:zooming-compatible}
Consider the \problem on a compact metric space. Let $\A_T$ be one phase of a zooming-compatible algorithm, where $T$ is the duration of the phase. Consider a clean run of $\A_T$.
\begin{itemize}
\item[(a)] Consider some round $t\leq T$ such that all previous rounds are well-covered. Then
\begin{align}
	\mD(x,y) > \min(r_t(x), r_t(y)) \geq \tfrac13 \min(\Delta(x), \Delta(y)).
\end{align}
for any two distinct active arms $x,y \in X$.
\item[(b)] Suppose the phase is well-covered. If $\A_T$ is $(c,d)$-constrained, $d\geq 0$, then it has regret
    $$R(T) \leq O(c\log T)^{\frac{1}{d+2}}\,\times T^{\frac{d+1}{d+2}}.$$
This regret bound also holds if $d$ is the zooming dimension with multiplier $c>0$.
\end{itemize}
\end{lemma}

An algorithm is called \emph{eventually well-covered} if for every
problem instance $(X,\mD,\mu)$ there is a constant $i_0$ such that every
clean phase $i > i_0$ is guaranteed to be well-covered,
as long as the preceding phase $i-1$ is also clean.%
\footnote{The last clause (``as long as the preceding phase is also clean'') is not needed for this subsection; it is added for compatibility with the analysis of the per-metric optimal algorithm in Section~\ref{sec:fatness-PMO}.}

\begin{corollary}[immediate from Section~\ref{subsec:zooming-analysis}]\label{cor:zooming-compatible}
Consider the \problem on a compact metric space. Let $\A$ be a zooming-compatible algorithm. Assume $\A$ is eventually well-covered. Then its (instance-specific) regret dimension is at most the zooming dimension. Further, if $\A$ is $(k,d)$-constrained, for some constant $k>0$, then the regret dimension of $\A$ is at most $d$.
\end{corollary}

Since our algorithm is $(k,d)$-constrained by design, to complete the proof of Theorem~\ref{lm:zooming-fatSet} it suffices to prove that the algorithm is eventually well-covered. This is the part of the analysis that is new, compared to Section~\ref{subsec:zooming-analysis}. The crux of the argument is encapsulated in the following claim.

\begin{claim}\label{lm:zooming-compatible-wellCovered}
Consider the \problem on a compact metric space. Fix $d>0$. Let $S' \subset S \subseteq  X$ (where $S'$ can be empty) be closed subsets such that
    $\COV(S\setminus U)<d$
for any open neighborhood $U$ of $S'$. Further, suppose $S$ contains some optimal arm and $S'$ does not. Let $\A_T$ be a clean phase of a zooming-compatible algorithm, where $T$ is the duration of the phase.  Suppose $\A_T$ activates an arm whenever some arm in $S\setminus S'$ is not covered and, for some $\rho_T>0$,
$$ |\{ \emph{active arms $x \in S\setminus S'$}:\, r_t(x) \geq \rho_T \}| < \rho_T^{-d} .$$
Here $\rho_T$ depends on $T$ so that $\rho_T\to 0$ as $T\to\infty$. Then the phase is well-covered whenever $T\geq T_0$, for some finite $T_0$ which may depend on the problem instance.
\end{claim}

\begin{proof}
Recall that a $\rho$-packing is a set $P\subset X$ such that any two points in this set are at distance at least $\rho$. For any $C>0$ and any subset $Y\subset X$ with $\COV(Y)<d$, there exists $\rho_0$ such that for any $\rho\leq \rho_0$ any $\rho$-packing of $Y$ consists of at most $C\,\rho^{-d}$ points.

We pick $\rho_0>0$ as follows.
\begin{itemize}

\item If $S'$ is empty, we pick $\rho_0>0$ such that any $\rho$-packing of $S$, $\rho\leq \rho_0$, consists of at most $\rho^{-d}$ points. Such $\rho_0$ exists because $\COV(S)<d$.

\item Now suppose $S'$ is not empty. Since the metric space is compact and $S'$ is closed, $\mu$ attains its supremum on $S'$. Since $S'$ does not contain any optimal arm, it follows that
    $\sup(\mu,S')<\mu^*-\eps$
for some $\eps>0$. Let
    $U  = \cup_{x\in S'}\; B(x,\tfrac{\eps}{2})$
be the $\tfrac{\eps}{2}$-neighborhood of $S'$. Then for each arm $x\in U$ we have $\Delta(x)>\tfrac{\eps}{2}$. Since the metric space is compact, there is $c_0<\infty$ such that any $\tfrac{\eps}{6}$-packing of $U$ consists of at most $c_0$ points. Moreover, we are given that
    $\COV(S\setminus U) <d$. Pick $\rho_0>0$ such that any $\tfrac{\rho}{4}$-packing of $S\setminus U$ consists of at most $\rho^{-d}-c_0$ points, for any $\rho\leq \rho_0$.
\end{itemize}

Supose $T$ is such that $\rho_T\leq \rho_0$; denote $\rho = \rho_T$. Let us prove that all rounds in this phase are well-covered. Let us use induction on round $t$. The first round of the phase is well-covered by design, because in this round some arm is activated, and the corresponding confidence ball covers the entire metric space. Now assume that for some round $t$, all rounds before $t$ are well-covered. Let $P$ be the set of all arms $x\in S$ that are active at time $t$ with $r_t(x)\geq \rho$. We claim that $|P|<\rho^{-d}$. Again, we consider two cases depending on whether $S'$ is empty.
\begin{itemize}

\item Suppose $S'$ is empty. By Lemma~\ref{lm:zooming-compatible}(a), $P$ is an $\rho$-packing, so $|P|<\rho^{-d}$ by our choice of $\rho_0$.

\item Suppose $S'$ is not empty. For any active arm $x\in U$ it holds that $\Delta(x)\geq \tfrac{\eps}{2}$. Then by Lemma~\ref{lm:zooming-compatible}(a) the active arms in $U$ form an $\tfrac{\eps}{6}$-packing of $U$. So $U$ contains at most $c_0<\infty$ active arms.

Further, let $P'$ be the set of all arms in $S\setminus U$ that are active at round $t$ with $r_t(x)\geq \rho$. By Lemma~\ref{lm:zooming-compatible}(a), $P$ is a $\rho$-packing, so $|P'|<\rho^{-d}-c_0$ by our choice of $\rho_0$. Again, it follows that $|P|<\rho^{-d}$.
\end{itemize}

Therefore by our assumption the algorithm activates an arm whenever some arm in $S\setminus S'$ is not covered. It follows that $S\setminus S'$ is covered after the activation step, so in particular some optimal arm is covered.
\end{proof}

\begin{corollary}
In the setting of Theorem~\ref{lm:zooming-fatSet}, any clean phase of algorithm $\A$ of duration $T\geq T_0$ is well-covered, for some finite $T_0$ which can depend on the problem instance.
\end{corollary}

\begin{proof}
Pick the largest $\ell\in\{0 \LDOTS k\}$ such that $S_\ell\setminus S_{\ell+1}$ contains some optimal arm. Then Claim~\ref{lm:zooming-compatible-wellCovered} applies with $S = S_{\ell}$ and $S' = S_{\ell+1}$.
\end{proof}

\OMIT{ 
\begin{claim}
In the setting of Theorem~\ref{lm:zooming-fatSet}, any clean phase of algorithm $\A$ of duration $T\geq T_0$ is well-covered, for some finite $T_0$ which can depend on the problem instance $(X,\mD,\mu)$.
\end{claim}
\begin{proof}
Recall that a $\rho$-packing is a set $P\subset X$ such that any two points in this set are at distance at least $\rho$. For any $C>0$ and any subset $Y\subset X$ with $\COV(Y)<d$, there exists $\rho_0$ such that for any $\rho\leq \rho_0$ any $\rho$-packing of $Y$ consists of at most $C\,\rho^{-d}$ points.

We pick
    $T_0 = \rho_0^{-(d+2)}$
as follows. Pick the largest $\ell\in\{0 \LDOTS k\}$ such that $S_\ell\setminus S_{l+1}$ contains some optimal arm.
\begin{itemize}

\item Suppose $\ell=k$, that is some optimal arm lies in $S=S_k$. Pick $\rho_0>0$ such that any $\rho$-packing of $S$, $\rho\leq \rho_0$, consists of at most $\rho^{-d}$ points. Such $\rho_0$ exists because $\COV(S)<d$.

\item Suppose $\ell<k$. Since the metric space is compact and $S_{\ell+1}$ is closed, it follows that $\sup(\mu,S_{\ell+1})<\mu^*-\eps$ for some $\eps>0$. Let
    $U  = \cup_{x\in S_{\ell+1}}\; B(x,\tfrac{\eps}{2})$
be the $\tfrac{\eps}{2}$-neighborhood of $S_{\ell+1}$. Then for each arm $x\in U$ we have $\Delta(x)>\tfrac{\eps}{2}$. Since the metric space is compact, any $\tfrac{\eps}{6}$-packing of $U$ consists of at most $c_0<\infty$ points. Moreover, $\COV(S_{\ell}\setminus U) <d$, because $S$ is a $d^*$-fat subset. Pick $\rho_0>0$ such that any $\tfrac{\rho}{4}$-packing of $S_{\ell}\setminus U$ consists of at most $\rho^{-d}-c_0$ points, for any $\rho\leq \rho_0$.
\end{itemize}

Consider a clean phase of $\A$ of duration $T\geq T_0$, and let $\rho = T^{-1/(d+2)}$. Let us prove that all rounds in this phase are well-covered use induction on round $t$. The first round of the phase is well-covered by design, because in this round some arm is activated, and the corresponding confidence ball covers the entire metric space. Now assume that for some round $t$, all rounds before $t$ are well-covered. Again, we consider two cases depending on $\ell$.

\begin{itemize}

\item Suppose $\ell=k$. Then $S=S_{\ell}$ contains some optimal arm $x^*$. Let $P$ be the set of all arms $x\in S$ that are active at time $t$ with $r_t(x)\geq \rho$.
By Lemma~\ref{lm:zooming-compatible}(a), $P$ is an $\rho$-packing, so $|P|<\rho^{-d}$ by our choice of $\rho_0$. Therefore in this round there is room under the quota for arms in $S$, so $S\subset \Eligible$. It follows that $S$ is covered after the activation step (if $S$ is not covered before the activation step then some arm is activated, after which all of $X$ is covered). So in particular $x^*$ is covered.

\item Suppose $\ell<k$. Then there exists an optimal arm $x^*\in S_{\ell}\setminus U$, where $U$ is the open neighborhood defined above. For any active arm $x\in U$ it holds that $\Delta(x)\geq \tfrac{\eps}{2}$.
Then by Lemma~\ref{lm:zooming-compatible}(a) the active arms in $U$ form an $\tfrac{\eps}{6}$-packing of $U$. So $U$ contains at most $c_0<\infty$ active arms.

Further, let $P$ be the set of all arms in $S_{\ell}\setminus U$ that are active at round $t$ with $r_t(x)\geq \rho$. By Lemma~\ref{lm:zooming-compatible}(a), $P$ is an $\rho$-packing, so $|P|<\rho^{-d}-c_0$ by our choice of $\rho_0$. Therefore in this round there is room under the quota for arms in $S_{\ell}\setminus S_{\ell+1}$, so
    $S_{\ell}\setminus S_{\ell+1}\subset \Eligible$.
It follows that $S_{\ell}\setminus S_{\ell+1}$ is covered after the activation step (if it is not covered before the activation step, then some arm is activated, in which case all of $X$ is covered). So in particular $x^*$ is covered.
\end{itemize}
So in both cases round $t$ is well-covered, and we are done.
\end{proof}
} 

In passing, let us give an example of a fatness decomposition of length $>1$. Start with a metric space $(X,\mD)$ with a $d$-fat subset $S$. Consider the product metric space $(X\times X,\mD^*)$ defined by
	$$\mD^*( (x_1, x_2),(y_1, y_2)) = \mD(x_1,y_1) + \mD(x_2, y_2).$$
This metric space admits a $2d$-fatness decomposition
    $$ (S_0, S_1, S_2, S_3) = (X\times X,\; (S\times X) \cup (X\times S),\; S\times S,\; \emptyset).$$

\subsection{Transfinite fatness decomposition}
\label{sec:fatness-PMO}

\newcommand{\PhaseAlg}{\ensuremath{\A_{\mathtt{ph}}}}

The fact that $d=\MaxMinCOV(X) < \COV(X)$ does not appear to imply the existence of a $d$-fatness decomposition of any finite length. Instead, we prove the existence of a much more general structure which we then use to design the per-metric optimal algorithm. This structure is a \emph{transfinite} sequence of subsets of $X$, i.e. a sequence indexed by ordinal numbers rather than integers.%
\footnote{Formally, a transfinite sequence of length $\beta$ (where $\beta$ is an ordinal) is a mapping from
$\{\text{ordinals }\lambda: 0\leq \lambda \leq \beta \}$
to the corresponding domain, in this case the power set of $X$.}
 	
\begin{definition}\label{def:fatness-transfinite}
Fix a metric space $(X,\mD)$. A \emph{transfinite $d$-fatness decomposition} of length $\beta$, where $\beta$ is an ordinal,
is a  transfinite sequence
	$\{S_\lambda\}_{0 \leq \lambda \leq \beta+1}$
of closed subsets of $X$ such that:
\begin{OneLiners}
\item[(a)] $S_0 = X$, $S_{\beta+1} = \emptyset$, and
	$S_\nu \supseteq S_\lambda$ whenever $\nu < \lambda$.
\item[(b)] for any ordinal $\lambda \leq \beta$ and any open set
$U\subset X$ containing $S_{\lambda+1}$ it holds that
	$\COV(S_\lambda \setminus U) \leq d$.~%
\footnote{For $\lambda = \beta$, this is equivalent to $\COV(S_{\beta})\leq d $.}
\item[(c)] If $\lambda$ is a limit ordinal then
	$S_{\lambda} = \bigcap_{\nu < \lambda} S_\nu$.
\end{OneLiners}
\end{definition}

\OMIT{ 
\item[(b)] if $V\subset X$ is closed, then the set
    $\{\text{ordinals } \nu \leq \beta$:\, $V \mbox{ intersects } S_\nu \}$
has a maximum element.
}

For finite length $\beta$ this is the same as (non-transfinite) $d$-fatness decomposition. The smallest infinite length $\beta$ is a countable infinity $\beta = \omega$. Then the transfinite sequence
	$\{S_\lambda\}_{0 \leq \lambda \leq \beta+1}$
consists of subsets
    $ \{ S_i\}_{i\in \N}$
followed by
    $S_{\omega} = \cap_{i\in \N}\, S_i$
and
    $S_{\omega+1} = \emptyset $.

\begin{proposition} \label{prop:fatness-dim}
For every compact metric space $(X,\mD)$, the
max-min-covering dimension 
is equal to the infimum of all
$d$ such that $(X,\mD)$ has a transfinite
$d$-fatness decomposition.
\end{proposition}

\newcommand{\ThinPt}{{\mathtt{Thin}}}
\newcommand{\FatPt}{{\mathtt{Fat}}}

\begin{proof}
Assume there exists a transfinite $d$-fatness decomposition
    	$\{S_\lambda\}_{0 \leq \lambda \leq \beta+1}$,
for some ordinal $\beta$. Let us show that $\MaxMinCOV(X)\leq d$. Suppose not, then there exists a non-empty subset $Y \subseteq X$ with $\MinCOV(Y) > d$. Let us use transfinite induction on $\lambda$ to prove that $Y\subseteq S_\lambda$ for all $\lambda \leq \beta$. This would imply $Y \subseteq S_\beta$ and consequently $\COV(S_\beta)>d$, contradiction.

The transfinite induction consists of three cases: ``zero case'', ``limit case'', and ``successor case''. The zero case is $Y\subseteq S_0=X$. The limit case is easy: if $\lambda\leq \beta$ is a limit ordinal and $Y\subseteq S_\nu$ for every $\nu<\lambda$ then
    $Y\subseteq S_\lambda = \cap_{\nu<\lambda} S_\nu $.
For the successor case, we assume $Y\subseteq S_\lambda$, $\lambda+1 \leq \beta$, and we need to show that
    $Y\subseteq S_{\lambda+1}$.
Suppose not, and pick some
    $x\in Y\cap (S_\lambda \setminus S_{\lambda+1})$.
Since $S_{\lambda+1}$ is closed, $x$ is at some positive distance $2\eps$ from $S_{\lambda+1}$. Then an $\eps$-neighborhood $U$ of $S_{\lambda+1}$ is disjoint with a ball $B = B(x,\eps)$. So $B\subseteq S_\lambda \setminus U$, which implies $\COV(B)\leq d$ by definition of transfinite $d$-fatness decomposition. However, since $B\cap Y$ is open in the metric topology induced by $Y$, by definition of the min-covering dimension we have $\COV(B)>d$. We obtain a contradiction, which completes the successor case.

Now given any $d > \MaxMinCOV(X)$, let us construct a transfinite $d$-fatness decomposition of length $\beta$, where $\beta$ is any ordinal whose cardinality exceeds that of $X$. For a metric space $(Y,\mD)$, a point is called $d$-thin if it is contained in some open $U\subset Y$ such that $\COV(Y)<d$, and $d$-thick otherwise. Let $\FatPt(Y,d)$ be the set of all $d$-thick points; note that  $\FatPt(Y,d)$ is a closed subset of $Y$. For every ordinal $\lambda \leq \beta+1$, we define a set $S_{\lambda}\subset X$
using transfinite induction as follows:
\begin{OneLiners}
\item[1.] $S_0 = X$ and  $S_{\lambda+1} = \FatPt(S_{\lambda},d)$
for each ordinal $\lambda$.
\item[2.]  If $\lambda$ is a limit ordinal then
	$S_{\lambda} = \bigcap_{\nu < \lambda} S_\nu$.	
\end{OneLiners}
This completes the construction of a sequence $\{ S_\lambda \}_{\lambda \leq \beta+1}$.

Note that each $S_\lambda$ is closed, by transfinite induction. It remains to show that the sequence satisfies the properties (a)-(c) in Definition~\ref{def:fatness-transfinite}. It follows immediately from the construction that $S_0 = X$ and $S_{\nu} \supseteq S_{\lambda}$ when $\nu < \lambda.$  To prove that $S_{\beta} = \emptyset$, observe first that the sets $S_{\lambda} \setminus S_{\lambda+1} \; (\mbox{for } 0 \leq \lambda < \beta)$ are disjoint subsets of $X$, and the number of such sets is greater than the cardinality of $X$, so at least one of them is empty.  This means that $S_{\lambda} = S_{\lambda+1}$ for some $\lambda < \beta.$ If $S_{\lambda} = \emptyset$ then $S_{\beta} = \emptyset$ as desired.  Otherwise, the relation $\FatPt(S_{\lambda},d) = S_{\lambda}$ implies that the metric space $(S_{\lambda},\, \mD)$ contains no open set $U\subset S_{\lambda}$ with $\COV(U)<d$. It follows that $\MinCOV(S_{\lambda}) \geq d$, contradicting the assumption that $\MaxMinCOV(X) < d.$ This completes the proof of property (a). To prove property (b), note that if $U$ is an open neighborhood of $S_{\lambda+1}$ then the set $T = S_{\lambda} \setminus U$ is closed (hence compact) and is contained in $\ThinPt(S_{\lambda},d)$. Consequently $T$ can be covered by open sets $V$ satisfying $\COV(V) < d$.  By compactness of $T$, this covering has a finite subcover $V_1,\ldots,V_m$, and consequently $\COV(T) = \max_{1 \leq i \leq m} \COV(V_i) < d.$ Finally, property (c) holds by design.
\end{proof}

\begin{theorem}\label{thm:PMO}
Consider the Lipschitz MAB problem on a compact metric space $(X,\mD)$ with a transfinite $d^*$-fatness decomposition, $d^*\geq 0$. Then for each $d>d^*$ there exists an algorithm $\A$ (parameterized by $d$) such that
	$\RegretDim(\A) \leq d$.
Moreover, the instance-specific regret dimension of $\A$ is bounded from above by the zooming dimension.
\end{theorem}

\newcommand{\DesiredLambda}{\lambda_{\mathtt{max}}}
\newcommand{\Quota}[1]{\mathtt{Q}_{#1}}

In the rest of this section we design and analyze an algorithm for Theorem~\ref{thm:PMO}. The algorithm from the previous subsection has regret proportional to the length of the fatness decomposition, so it does not suffice even if the fatness decomposition has countably infinite length. As it turns out, the main algorithmic challenge in dealing with fatness decompositions of transfinite length is to handle the special case of \emph{finite} length $k$ so that the regret bound does not depend on $k$.

In what follows, let
    $\{S_\lambda\}_{0 \leq \lambda \leq \beta+1}$,
be a transfinite $d^*$-fatness decomposition of length $\beta$, for some ordinal $\beta$ and $d^*\geq 0$. Fix some $d>0$.

\begin{proposition}\label{prop:fatness-decomposition-closed}
For any closed $V\subset X$, there is a maximal ordinal $\lambda$ such that $V$ intersects $S_\lambda$.
\end{proposition}

\begin{proof}
Let
    $\Omega = \{\text{ordinals } \nu \leq \beta$:\, $V \mbox{ intersects } S_\nu \}$,
and let $\nu = \sup (\Omega)$.  Then 
    \[ S_{\nu} \cap V = \bigcap_{\lambda\in \Omega} (S_{\lambda} \cap V),\]
and this set is nonempty because $X$ is compact and the closed sets $\{S_{\lambda} \cap V \,:\, \lambda \in \Omega\}$ have the finite intersection property. 
(To derive the latter, consider a finite subset $\Omega'\subset \Omega$ and let
    $\nu' = \max(\Omega')\in \Omega$.
Then
    $\bigcap_{\lambda\in \Omega'} (S_{\lambda} \cap V) = S_{\nu'} \cap V$,
which is not empty by definition of $\Omega$.)
\end{proof}

Recall that the supremum $\mu^* = \sup(\mu,X)$ is attained because the metric space is compact. Further, recall that the arms $x$ such that $\mu(x) = \mu$ are called optimal. Let $\DesiredLambda$ be the maximal $\lambda$ such that $S_\lambda$ contains an optimal arm. Such $\DesiredLambda$ exists by Proposition~\ref{prop:fatness-decomposition-closed} because the set $V = \mu^{-1}(\mu^*)$ is non-empty and closed. Note that $S_{\DesiredLambda}$ contains an optimal arm, whereas $S_{\DesiredLambda+1}$ does not.

Our algorithm is a version of Algorithm~\ref{alg:zooming-fat} from the previous subsection, with a different ``eligibility rule'' -- the definition of $\Eligible$. For phase duration $T$ and an ordinal $\lambda\leq \beta$, define the quota as the following condition:
\begin{align*}
\Quota{\lambda} \triangleq
\left[ \quad |\{ \emph{active arms $x \in S_\lambda$}:\, r_t(x) \geq \rho \}| < \rho^{-d} \quad\right],
\quad \rho= T^{- 1/(d+2)}.
\end{align*}
The algorithm maintains the \emph{target ordinal} $\lambda^*$, recomputed after each phase, so that some arm in $S_{\lambda^*}$ is activated as long as the quota $\Quota{\lambda^*}$ is satisfied. Further, there is a subset $\mathcal{N}$ of cardinality at most $T^{d/(d+2)}$, chosen in the beginning of each phase, such that all arms in $\mathcal{N}$ are always eligible and all arms not in $S_{\lambda^*} \cup \mathcal{N}$ are never eligible.

Note that such algorithm is $(1,d)$-constrained by design, because in any round $t$ there can be at most $\rho^{-d}$ active arms in $S_{\lambda^*} \setminus \mathcal{N} $ with confidence radius less than $\rho = T^{-1/(d+2)}$.

The analysis hinges on proving that after any sufficiently long clean phase the target ordinal is $\DesiredLambda$, and then the subsequent phase (assuming it is also clean) is well-covered, and then the desired regret bounds follow from Corollary~\ref{cor:zooming-compatible}.  Any sufficiently long clean phase with target ordinal $\DesiredLambda$ is well-covered by Claim~\ref{lm:zooming-compatible-wellCovered}. So the only new thing to prove is that that after any sufficiently long clean phase the target ordinal is $\DesiredLambda$.

We also change the definition of index to
\begin{align}\label{eq:index-pmo}
I_t(x) = \mu_t(x) + 3\, r_t(x),
\end{align}
where, as before, $\mu_t(x)$ denotes the average payoff from arm $x$ in rounds $1$ to $t-1$ of the current phase, and $r_t(x)$ is the current confidence radius of this arm. It is easy to check that the analysis in Section~\ref{subsec:zooming-analysis}, and therefore also Lemma~\ref{lm:zooming-compatible} and Corollary~\ref{cor:zooming-compatible}, carry over to any index of the form
    $I_t(x) = \mu_t(x) + c_0\, r_t(x)$
for some absolute constant $c_0\geq 2$ (the upper bound on regret increases by the factor of $c_0$).

The pseudocode is summarized as Algorithm~\ref{alg:zooming-pmo}. In the beginning of each phase, the subset $\mathcal{N}\subset X$ is defined as follows. We choose $\mathcal{N}$ to be an $\eps_0$-net%
\footnote{Recall from Section~\ref{sec:dim-notions} that an $\eps$-net of a metric space $(X,\mD)$ is a subset $S\subset X$ such that any two points in $S$ are at distance at least $\eps$ from one another, and any point in $X$ is within distance $\eps$ from some point in $S$.}
of $X$ which consists of at most $T^{d/(d+2)}$ points, for (essentially) the smallest possible $\eps_0>0$. More precisely, we compute an $\eps_0>0$ and an $\eps_0$-net $\mathcal{N}$ using a standard \emph{greedy heuristic}. For a given $\eps>0$ we construct an $\eps$-net $S\subset X$ as follows: while there exists a point $x\in X$ such that
    $\mD(S,x) \triangleq \inf_{y\in S} \mD(x,y) <\eps$, add any such point to $S$, and abort if $|S|> T^{d/(d+2)}$. We consecutively try
    $\eps = 2^{-j}$ for each $j=1,2,3, \ldots$,
and pick the smallest $\eps$ which results in an $\eps$-net of at most $T^{d/(d+2)}$ points.

In the end of each phase, the new target ordinal $\lambda^*$ is defined as follows.  We pick an $\eps^*$ according to $T$ and $\eps_0$, and focus on arms whose confidence radius is less than $\eps^*$. Let $A$ be the set of all such arms. We define $\lambda^*$ as the largest ordinal $\lambda$ such that $S_\lambda$ intersects
   $\bar{B}(A,\eps^*) \triangleq \{x\in X:\; \mD(A,x) \leq \eps^* \}$,
the the closed $\eps^*$-neighborhood of $A$. Such ordinal exists by Proposition~\ref{prop:fatness-decomposition-closed}.

\begin{algorithm}[h]
\caption{(the per-metric optimal algorithm)}
\label{alg:zooming-pmo}

\newcommand{\TAB}{\hspace{5mm}}

\begin{algorithmic}
\STATE Target ordinal $\lambda^* \leftarrow 0$.
\FOR{phase $i=1,2,3, \ldots$}
\STATE \COMMENT{Phase duration is $T=2^i$}
\STATE  Compute an $\eps_0>0$ and an $\eps_0$-net $\mathcal{N}$ of $X$ such that $|\mathcal{N}|<T^{d/(d+2)}$.
\STATE \TAB\COMMENT{use greedy heuristic}
\STATE Initially, no arms are active.
\FOR{round $t=1,2,3, \ldots, T$}
\STATE
$\Eligible = \begin{cases}
\mathcal{N} \cup S_{\lambda^*} & \text{if constraint $\Quota{\lambda^*}$ is satisfied}, \\
\mathcal{N}  & \text{otherwise}.
\end{cases}$
\STATE \emph{Activation rule:} if some arm $x\in \Eligible$ is not covered,
\STATE\TAB pick any such arm and activate it.
\STATE \emph{Selection rule:} play any active arm with the maximal index~\refeq{eq:index-pmo}.
\ENDFOR
\STATE \COMMENT{Recompute the target ordinal $\lambda^*$}
\STATE $\eps^*      =  6\,\max(\eps_0, 4 T^{-1/(d+2)} \sqrt{\log T})$.
\STATE  $\lambda^*   = \max\{ \lambda:\; \text{$S_\lambda$ intersects $\bar{B}(A,\eps^*)$ } \}$,
    where $A = \{ \text{active arms $x$: $r_T(x)<\eps^*$ } \}$.
\ENDFOR
\end{algorithmic}
\end{algorithm}

\xhdr{Implementation details.}
To implement Algorithm~\ref{alg:zooming-pmo}, it suffices to use the following oracles:
\begin{itemize}
\item For any finite set of open balls
    $B_1 \LDOTS B_n$
(given via the centers and the radii) whose union is denoted by $B$, the \emph{depth oracle} returns
    $\sup \{\lambda:\, \text{$S_{\lambda}$ intersects the closure of $B$}  \} $.

\item Given balls $B_1 \LDOTS B_n$ as above, and an ordinal $\lambda$, the \emph{enhanced covering oracle}
either reports that $B$ covers $S_{\lambda}$, or it returns an arm
	$x \in S_{\lambda} \setminus B$.
\end{itemize}
To avoid the question of how arbitrary ordinals are represented on the oracle's output tape, we can instead say that the depth oracle outputs a point $u \in S_{\lambda} \setminus S_{\lambda+1}$ instead of outputting $\lambda.$  In this case, the definition of the covering should be modified so that it inputs a point $u \in S_{\lambda} \setminus S_{\lambda+1}$ rather than the ordinal $\lambda$ itself.

\xhdr{Analysis.}
We bring in the machinery developed in the previous subsection. Note that Algorithm~\ref{alg:zooming-pmo} is zooming-compatible and $(1,d)$-constrained by design. Therefore by Corollary~\ref{cor:zooming-compatible} we only need to prove that it is eventually well-covered. If in a given clean phase the target ordinal is $\DesiredLambda$, this phase satisfies the assumptions in Claim~\ref{lm:zooming-compatible-wellCovered} for $S = S_{\DesiredLambda}$ and $S' = S_{\DesiredLambda+1}$. It follows that any sufficiently long clean phase with target ordinal $\DesiredLambda$ is well-covered. Thus it remains to show that after any sufficiently long clean phase of Algorithm~\ref{alg:zooming-pmo} the target ordinal is $\DesiredLambda$. (This is where we use the new definition of index.)

\begin{claim}\label{cl:pmo-analysis-targetOrdinal}
After any sufficiently long clean phase of Algorithm~\ref{alg:zooming-pmo} the target ordinal is $\DesiredLambda$.
\end{claim}

\newcommand{\xFreq}{x_{\mathtt{freq}}}
\newcommand{\xNet}{x_{\mathtt{net}}}

To prove Claim~\ref{cl:pmo-analysis-targetOrdinal} we need to ``open up the hood'' and analyze the internal workings of the algorithm. (We have been avoiding this so far by using Corollary~\ref{cor:zooming-compatible}.) Such analysis is encapsulated in the following claim. Note that we cannot assume that the phase is well-covered.

\begin{claim}\label{cl:pmo-analysis-crux}
Consider a clean phase of Algorithm~\ref{alg:zooming-pmo} of duration $T$, with $\eps_0$-net $\mathcal{N}$. Let $y$ be an arm that has been played at least once in this phase. Then
\begin{OneLiners}
\item[(a)]    $\Delta(y) \leq 4\, r_T(y) + \eps_0$.
\item[(b)] For any optimal arm $x^*$, there exists an active arm $x$ such that
    $$ \min(\mD(x,x^*),\; r_T(x)) \leq 4\, r_T(y) + 2\, \eps_0.$$
\end{OneLiners}
\end{claim}

\begin{proof}
Let $\xNet \in \mathcal{N}$ be such that $\mD(x^*, \xNet)\leq \eps_0$. Let $t$ be the last time arm $y$ is played in this phase. Let $x$ be an arm that covers $\xNet$ at time $t$. (Since $\mathcal{N} \subset \Eligible$, all points in $\mathcal{N}$ are covered at all times.) Then:
\begin{align*}
I_t(x)  &\geq \mu(x) + 2 r_t(x)         & \text{by definition of index and confidence radius} \\
        &\geq \mu(\xNet) + r_t(x)       & \text{because $x$ covers $\xNet$ at time $t$} \\
        &\geq \mu^* - \eps_0 + r_t(x)   & \text{because $\mD(x^*, \xNet)\leq \eps_0$} \\
I_t(x)  &\leq I_t(y)                    & \text{because arm $y$ is played at time $t$} \\
        &\leq \mu(y) + 4\, r_t(y)       & \text{by definition of index and confidence radius} \\
        &= \mu^* - \Delta(y) + 4\, r_t(y).
\end{align*}
Combining the two inequalities, we obtain:
$$ \mu^* - \Delta(y) + 4\, r_t(y) \geq I_t(x) \geq \mu^* - \eps_0 + r_t(x). $$
Noting that $r_T(y) = r_t(y)$, we obtain
    $$\Delta(y) + r_t(x) \leq 4\, r_T(y) + \eps_0. $$
This immediately implies part (a) of the claim. Part (b) follows by triangle inequality, because
    $$\mD(x,x^*) \leq \mD(x,\xNet) + \mD(\xNet,x^*) \leq r_t(x) + \eps_0. \qedhere $$
\end{proof}

We also need a simple and well-known fact about compact metric spaces.

\begin{claim}[Folklore]\label{cl:pmo-analysis-folklore}
For any given $\delta>0$ there exists $T_0<\infty$ such that in any phase of Algorithm~\ref{alg:zooming-pmo} of duration $T>T_0$, the algorithm computes an $\eps_0$-net $\mathcal{N}$ such that $\eps_0<\delta$.
\end{claim}
\begin{proof}
Fix $\delta>0$. Since the metric space is compact, there exists a covering of $X$ with finitely many subsets $S_1 \LDOTS S_n\subset X$ of diameter less than $\tfrac{\delta}{2}$. Suppose $T$ is large enough so that $n<T^{d/(d+2)}$. Suppose
the algorithm computes an $\eps_0$-net $\mathcal{N}$ such that $\eps_0\geq \delta$. Then the following iteration of the greedy heuristic (if not aborted) would construct an $\eps_0/2$-net $\mathcal{N}'$ for $X$ with more than $T^{d/(d+2)}$ points. However, any two points in $\mathcal{N}'$ lie at distance $\geq \delta/2$ from one another, so they cannot lie in the same set $S_i$. It follows that $|\mathcal{N}'| \leq n$, contradiction.
\end{proof}

\begin{proofof}{Claim~\ref{cl:pmo-analysis-targetOrdinal}}
Consider a clean phase of Algorithm~\ref{alg:zooming-pmo} of duration $T$, with an $\eps_0$-net $\mathcal{N}$. Let $\eps^*$ and $A$ be defined as in Algorithm~\ref{alg:zooming-pmo}, so that
    $A = \{ \text{active arms $x$: $r_t(x)<\eps^*$ } \}$.
We need to show that for any sufficiently large $T$ two things happen: $\bar{B}(A,\eps^*)$ intersects $S_{\DesiredLambda}$ and it does not intersect $S_{\DesiredLambda+1}$.

Let $\xFreq$ be the most frequently played arm by the end of the phase. We claim that
    $$r_T(\xFreq) < 4 T^{-1/(d+2)} \sqrt{\log T}.$$
Suppose not. By our choice of $\xFreq$, at time $T$ all arms have confidence radius at least $r_T(\xFreq)$. Since the algorithm is $(1,d)$-constrained, it follows that at most $n = 2T^{d/(d+2)}$ arms are activated throughout the phase. So by the pigeonhole principle $n_T(\xFreq) \geq T/n = \tfrac12 T^{2/(d+2)}$, which implies the desired inequality. Claim proved.

Let $x^* \in S_{\DesiredLambda}$ be some optimal arm. Taking $y = \xFreq$ in Claim~\ref{cl:pmo-analysis-crux}(b) and noting that
$4\,r_T(\xFreq)+ 2\,\eps_0 \leq \eps^*$, we derive that there exists an active arm $x$ such that $\mD(x,x^*)\leq \eps^*$ and $r_T(x) \leq \eps^*$. It follows that $x\in A$ and
    $x^* \in \bar{B}(A,\eps^*)$.
Therefore $\bar{B}(A,\eps^*)$ intersects $S_{\DesiredLambda}$.

Since the metric space is compact and $S_{\DesiredLambda+1}$ is a closed subset that does not contain an optimal arm, it follows that any arm in this subset has expected payoff at most $\mu^*-\eps$, for some $\eps>0$. Assume $T$ is sufficiently large so that $\eps^*<\eps/6$.
(We can make sure that $\eps_0<\eps/6$ by Claim~\ref{cl:pmo-analysis-folklore}).

To complete the proof, we need to show that $\bar{B}(A,\eps^*)$ does not intersect $S_{\DesiredLambda+1}$. Suppose this is not the case. Then there exists $x\in S_{\DesiredLambda+1}$ and active $y\in X$ such that $\mD(x,y)\leq \eps^*$ and $r_T(y)\leq \eps^*$. Then by Claim~\ref{cl:pmo-analysis-crux} (a) we have that
    $\Delta(y)\leq 4\,\eps^* + \eps_0 \leq 5\,\eps^*$,
which implies that
    $\Delta(x) \leq \Delta(y) + \mD(x,y) \leq 6\,\eps^* < \eps$,
contradicting our assumption that every arm in $S_{\DesiredLambda+1}$ has expected payoff at most $\mu^*-\eps$. Claim proved.
\end{proofof}

\section{The (sub)logarithmic vs. $\sqrt{t}$ regret dichotomy}
\label{sec:dichotomies}

This section concerns the dichotomy between (sub)logarithmic and $\sqrt{t}$ regret for Lipschitz bandits and Lipschitz experts (Theorem~\ref{thm:intro-dichotomy-MAB} and Theorem~\ref{thm:main-experts}, respectively). We focus on the restriction of these results to compact metric spaces:

\begin{theorem}\label{thm:dichotomy-compact}
Fix a compact metric space $(X,\mD)$. The following dichotomies hold:
\begin{itemize}
\item[(a)] The \problem on $(X,\mD)$ is either $f(t)$-tractable for every $f\in \omega(\log t)$, or it is not $g(t)$-tractable for any $g\in o(\sqrt{t})$.

\item[(b)] The \FFproblem on $(X,\mD)$ is either $1$-tractable, even with double feedback, or it is not $g(t)$-tractable for any $g\in o(\sqrt{t})$, even with full feedback and uniformly Lipschitz payoffs.
\end{itemize}
In both cases, (sub)logarithmic tractability occurs if and only if $X$ is countable.
\end{theorem}

We also prove two auxiliary results: the $(\log t)$-intractability for Lipschitz bandits on infinite metric spaces (Theorem~\ref{thm:logT}), and an algorithmic result via a more intuitive oracle access to the metric space (for metric spaces of finite \emph{Cantor-Bendixson rank}, a classic notion from point-set topology).

The section is organized as follows. We provide a joint analysis for Lipschitz bandits and Lipschitz experts: an overview in Section~\ref{sec:dichotomies-overview}, the lower bound is in Section~\ref{sec:lower-bound}, and the algorithmic result is in Section~\ref{sec:tractability}. The two auxiliary results are, respectively, in Section~\ref{sec:logT} and Section~\ref{sec:simpler-alg}.

\subsection{Regret dichotomies: an overview of the proof}
\label{sec:dichotomies-overview}

We identify a simple topological property (existence of a topological well-ordering) which entails the algorithmic result, and another topological property (existence of a perfect subspace) which entails the lower bound.

\begin{definition}\label{def:topology}
Consider a topological space $X$.
$X$ is called \emph{perfect} if it contains no isolated points.
A \emph{topological well-ordering} of $X$ is a well-ordering
$(X,\prec)$ such that every initial segment thereof is an open
set.  If such $\prec$ exists, $X$ is called \emph{well-orderable}.
A metric space $(X,\mD)$ is called well-orderable if and only if
its metric topology is well-orderable.
\end{definition}

Perfect spaces are a classical notion in point-set topology.
Topological well-orderings are implicit in
the work of \citet{Cantor83}, but the particular definition
given here is new, to the best of our knowledge.

The proof of Theorem~\ref{thm:dichotomy-compact} consists of three parts: the algorithmic result for a compact, well-orderable metric space, the lower bound for a metric space with a perfect subspace, and the following lemma that ties together the two topological properties.

\begin{lemma}\label{lm:topological-equivalence}
For any compact metric space $(X,\mD)$, the following are equivalent: (i) $X$ is a countable set, (ii) $(X,\mD)$ is well-orderable, and (iii) no metric subspace of $(X,\mD)$ is perfect.\footnote{For arbitrary metric spaces we have (ii)$\iff$(iii) and (i)$\Rightarrow$(ii), but not (ii)$\Rightarrow$(i). }
\end{lemma}
Lemma~\ref{lm:topological-equivalence} follows from classical theorems of \citet{Cantor83} and \citet{MazSier}.  We provide a proof in Appendix~\ref{sec:topological}
for the sake of making our exposition self-contained.

\xhdr{Extension to arbitrary metric spaces.}
We extend Theorem~\ref{thm:dichotomy-compact} to the corresponding dichotomies for arbitrary metric spaces using the reduction to complete metric spaces in Appendix~\ref{sec:reduction}, and the $o(t)$-intractability result for non-compact metric spaces in Theorem~\ref{thm:boundary-of-tractability} (which is proved independently in Section~\ref{sec:boundary-body}).

For Lipschitz MAB, the argument is very simple. First, we reduce from arbitrary metric spaces to complete metric spaces: we show that the \problem is $f(t)$-tractable on a given metric space if and only if it is $f(t)$-tractable on the completion thereof (see Appendix~\ref{sec:reduction}). Second, we reduce from complete metric spaces to compact metric spaces using Theorem~\ref{thm:boundary-of-tractability}: by this theorem, the \problem is not $o(t)$-tractable if the metric space is complete but not compact. Thus, we obtain the desired dichotomy for Lipschitz MAB on arbitrary metric spaces, as stated in Theorem~\ref{thm:intro-dichotomy-MAB}.

For Lipschitz experts, the argument is slightly more complicated because the reduction to complete metric spaces only applies to the lower bound. Let $(X,\mD)$ be an arbitrary metric space, and let $(X^*,\mD^*)$ denote the metric completion thereof. First, if $(X^*,\mD^*)$ is not compact then by Theorem~\ref{thm:boundary-of-tractability} the \FFproblem is not $o(t)$-tractable. Therefore, it remains to consider the case that $(X^*,\mD^*)$ is compact. Note that Theorem~\ref{thm:dichotomy-compact} applies to $(X^*,\mD^*)$. If $X^*$ is not countable, then by Theorem~\ref{thm:dichotomy-compact} the problem is not $o(\sqrt{t})$-tractable on $(X^*,\mD^*)$, and therefore it is not $o(\sqrt{t})$-tractable on $(X,\mD)$ (see Appendix~\ref{sec:reduction}). If $X^*$ is countable, then the algorithm and analysis in Section~\ref{sec:tractability} apply to $X$, too, and guarantee $O(1)$-tractability. Thus, we obtain the desired dichotomy for Lipschitz experts on arbitrary metric spaces, as stated in Theorem~\ref{thm:main-experts}.

\OMIT{Same is true for the double-feedback \FFproblem, and the ``only if" direction holds for the \FFproblem. Then the main dichotomy results follow from the lower bound in Theorem~\ref{thm:boundary-of-tractability}.}

\subsection{Lower bounds via a perfect subspace}
\label{sec:lower-bound}

In this section we prove the following lower bound:

\begin{theorem}\label{thm:lower-bound}
Consider the uniformly Lipschitz experts problem on a metric space $(X,\mD)$ which has a perfect subspace. Then the problem is not $g$-tractable for any $g\in o(\sqrt{t})$.
In particular, for any such $g$ there exists a distribution $\mathcal{P}$ over problem instances $\mu$ such that for any experts algorithm \A\ we have
\begin{align}\label{eq:lower-bound}
\Pr_{\mu\in\mathcal{P}}
	\left[ R_{(\A,\,\mu)}(t) = O_{\mu}(g(t)) \right] = 0.
\end{align}
\end{theorem}

Let us construct the desired distribution over problem instances. First, we use the existence of a perfect subspace to construct a ball-tree (cf.\ Definition~\ref{def:ball-tree}).

\begin{lemma}\label{lm:ball-tree}
For any metric space with a perfect subspace there exists a ball-tree
in which each node has exactly two children.
\end{lemma}

\begin{proof}
Consider a metric space $(X,\mD)$ with a perfect subspace $(Y,\mD)$. Let us construct the ball-tree recursively, maintaining the invariant that for each tree node $(y,r)$ we have $y\in Y$. Pick an arbitrary $y\in Y$ and let the root be $(y,1)$. Suppose we have constructed a tree node $(y,r)$, $y\in Y$.  Since $Y$ is perfect, the ball $B(y,r/3)$ contains another point  $y'\in Y$. Let $r' = \mD(y,y')/2$ and define the two children of $(y,r)$ as $(y,r')$ and $(y',r')$.
\end{proof}

Now let us use the ball-tree to construct the distribution on payoff functions. (We will re-use this construction in Section~\ref{sec:MaxMinLCD-LB}.) In what follows, we consider a metric space $(X,\mD)$ with a fixed ball-tree $T$.  For each $i\geq 1$, let $D_i$ be the set of all depth-$i$ nodes in the ball-tree. Recall that an \emph{end} in a ball-tree is an infinite path from the root:
	$\mathbf{w} = (w_0, w_1, w_2,\,\ldots)$, where $ w\in D_i$ for all $i$.
For each tree node $w = (x_0, r_0)$ 
define the ``bump function''
	$F_{w}: X \rightarrow [0,1]$ as in \eqref{eq:bump-fn}:
\begin{align}\label{eq:needle}
 F_{w}(x) = \begin{cases}
	\min\{ r_0 - \mD(x,x_0),\, r_0/2  \} & \text{if $x\in B(x_0, r_0)$}, \\
	0	& \text{otherwise.}
\end{cases}
\end{align}
The construction is parameterized by a sequence
    $\delta_i, \delta_2, \delta_3, \, \ldots \in(0,1)$
which we will specify later.

\begin{definition}
A \emph{lineage} in a ball-tree is a
set of tree nodes containing at most
one child of each node; if it contains
\emph{exactly} one child of each node
then we call it a \emph{complete lineage}.
For each complete lineage
$\lambda$ there is an associated end $\mathbf{w}(\lambda)$
defined by $\mathbf{w}=(w_0,w_1,\, \ldots)$ where
$w_0$ is the root and for $i>0$, $w_i$ is the unique
child of $w_{i-1}$ that belongs to $\lambda$.
\end{definition}

\begin{construction}\label{con:LB-payoffs}
For any lineage $\lambda$ let us define a problem instance $\prob_{\lambda}$
(probability measure on payoff functions) via the following sampling rule.  First every tree node $w$ independently
samples a random sign $\sigma(w) \in \{+1,-1\}$
so that $\E[\sigma(w)] = \delta_i$ if
$w$ is the depth $i\geq 1$ node in $\mathbf{w}(\lambda)$,
and choosing the sign
uniformly at random otherwise.  Define the payoff function $\pi$ associated with a particular sign pattern $\sigma(\cdot)$ as follows:
\begin{align}\label{eq:LB-pi-defn}
    \payoff = \frac12 + \frac13\, \sum_{w \in T \setminus D_0} \sigma(w)\, F_{w}.
\end{align}
Let
    $\mu_\lambda(x) = \E_{\pi\sim \prob_\lambda} [\pi(x)]$
denote the expectation of $\pi(x)$ under distribution $\prob_\lambda$.

Let $\mathcal{P}_{T}$ be the distribution over problem instances
$\prob_{\lambda}$ in which  $\lambda$ is a complete lineage
sampled uniformly at random; that is, each node samples one of its children
independently and uniformly at random, and $\lambda$ is the set of
sampled children.
\end{construction}

\begin{note}{Remark.}
By Lemma~\ref{lm:LB-Lipschitz}, the payoff function in \refeq{eq:LB-pi-defn} is Lipschitz on $(X,\mD)$ for any sign pattern $\sigma(\cdot)$. Therefore $\prob_{\lambda}$ is an instance of uniformly Lipschitz experts problem, for each lineage $\lambda$.
\end{note}

To complete Construction~\ref{con:LB-payoffs}, it remains to specify the $\delta_i$'s. Fix function $g()$ from Theorem~\ref{thm:lower-bound}. For each $i\geq 1$, let
	$r^*_i = \min \{r: (x,r)\in D_i   \}$
be the smallest radius among depth-$i$ nodes in the ball-tree. Note that $r^*_i\leq 4^{-i}$.
Choose a number $n_i$ large enough that
    $g(n) < \tfrac{1}{24\,i}\, r^*_i \sqrt{n}$
for all $n > n_i$; such $n_i$ exists because $g\in o(\sqrt{t})$.
Let $\delta_i = n_i^{-1/2}$.

\OMIT{Note that the radius decreases from parents to children in the ball tree, so the sequence $(r^*_i:\, i\in\N)$ is decreasing. Consequently, the sequence $(n_i:\,i\in\N$ is non-increasing, and therefore the sequence
    $(\delta_i:\,i\in\N)$
is non-decreasing, as required. }

\begin{note}{Discussion.}
For a complete lineage $\lambda$, the expected payoffs are given by
\begin{align}\label{eq:experts-mu}
\mu_\lambda = \frac12 + \frac13\,\sum_{i=1}^\infty \delta_i\, F_{w_i},
\end{align}
where $\mathbf{w}(\lambda) = (w_0,w_1, \ldots)$ is the end associated with $\lambda$. For the special case of MAB it would suffice to construct a problem instance with  expected payoffs given by \refeq{eq:experts-mu}, without worrying about lineages or random sign patterns. This would be a ``weighted" version of the lower-bounding construction from Section~\ref{sec:PMO-LB}.

However, for the full-feedback problem it is essential that the sum in \eqref{eq:LB-pi-defn} is over all tree nodes (except the root), rather than the end $\mathbf{w}(\lambda)$. If the sum were over $\mathbf{w}(\lambda)$, then a single sample of the payoff function $\pi$ would completely inform the learner of the location of $\mathbf{w}(\lambda)$ in the tree. (Just look for the nested rings on which $\pi$ varies, and they form a target whose bulls-eye is $\mathbf{w}(\lambda)$.) Instead, we fill the whole metric space with ``static" in the form of a hierarchically nested set of rings on which $\pi$ varies, where the only special distinguishing property of the rings that zero in on $\mathbf{w}(\lambda)$ is that there is a slightly higher probability that $\pi$ increases on those rings. Thus, $\mathbf{w}(\lambda)$ is well-hidden, and in particular is impossible to learn from a single sample of $\pi$.
\end{note}

Let us state and prove a salient property of Construction~\ref{con:LB-payoffs} which we use to derive the regret lower bound. (This property holds for an arbitrary non-increasing sequence of $\delta_i$'s; we will re-use it in Section~\ref{sec:MaxMinLCD-LB}.)

\begin{lemma}\label{lm:LB-experts-salient}
Fix a complete lineage $\lambda$ and tree node $v\in \lambda$. To fix the notation, let us say that $v$ is depth-$i$ node with corresponding ball $B$ of radius $r$.

\begin{itemize}
\item[(i)] For every event $\mE$ in the Borel $\sigma$-algebra on $[0,1]^X$,
$$ \prob_\lambda(\mE)/ \prob_{\lambda\setminus \{v\}}(\mE)
    \in [1-\delta_i,\, 1+\delta_i].
$$

\item[(ii)] If $v\in \mathbf{w}(\lambda)$, then
$\sup(\mu_\lambda, B) - \sup(\mu_\lambda, X\setminus B) \geq r \delta_i/6$

\end{itemize}
\end{lemma}

\begin{proof}
For part (i), let us treat $\mE$ as a set of sign patterns $\sigma:V\to\{\pm 1\}$, where $V$ is the set of all nodes in the ball-tree, and let us treat $\prob_\lambda$ as a measure on these sign patterns. For each sign $\beta\in\{\pm 1\}$, let
    $\mE_\beta = \{\sigma\in\mE:\, \sigma(v)=\beta \}$
be the set of all sign patterns in $\mE$ with a given sign on node $u$. Note that
\begin{align*}
\prob_\lambda(\mE) = \sum_{\beta\in \{\pm 1\}}
    \prob_\lambda(\mE_\beta)\cdot \prob_\lambda\left( \sigma(v) = \beta \right).
\end{align*}
This equality holds for any lineage, in particular for lineage $\lambda\setminus\{v\}$.

For brevity, denote $\prob_0 = \prob_{\lambda \setminus \{v\}}$. Observe that $\prob_0$ and $\prob_\lambda$ differ only in how they set $\sigma(v)$. We can state this property rigorously as follows:
\begin{align*}
\prob_\lambda(\mE_\beta) = \prob_0(\mE_\beta)
    \qquad \text{for each sign $\beta\in \{\pm 1\}$}.
\end{align*}
Now, recalling that the event $\{\sigma(v)=1\}$ is assigned probability $\tfrac12$ under measure $\prob_0$, and probability $\tfrac12+\delta_i/2$ under measure $\prob_\lambda$, it follows that
\begin{align*}
\prob_\lambda(\mE) - \prob_0(\mE)
    &= (\delta_i/2)\, \left(
        \prob_0(\mE_+) - \prob_0(\mE_-)
    \right) \\
|\prob_\lambda(\mE) - \prob_0(\mE)|
    &\leq (\delta_i/2)\, \left(
        \prob_0(\mE_+) + \prob_0(\mE_-)
    \right)
    = \delta_i\, \prob_0(\mE).
\end{align*}

For part (ii), write $\mathbf{w}(\lambda) = (w_0,w_1, \ldots)$ be the end corresponding to $\lambda$. Recall that $v=w_i$. For each $w_j$, let $B_j$ be the corresponding ball, and let $r_j$ be its radius. Using \refeq{eq:experts-mu} and the fact that the sequence $(B_j:\,j\in\N)$ is decreasing, it follows that
\begin{align*}
\sup(\mu_\lambda,B)
    &= \frac12 + \frac16\,\sum_{j=1}^\infty \delta_j\, r_j, \\
\sup(\mu_\lambda,X\setminus B)
    &= \frac12 + \frac16\,\sum_{j=1}^{i-1} \delta_j\, r_j, \\
\sup(\mu_\lambda,B) -\sup(\mu_\lambda,X\setminus B)
    &= \frac16\,\sum_{j=i}^\infty \delta_j\, r_j
    \geq \delta_i\, r_i/6. \qedhere
\end{align*}
\end{proof}

\begin{lemma}\label{lm:ball-tree-LB}
Consider a metric space $(X,\mD)$ with a ball-tree $T$. Then~\refeq{eq:lower-bound} holds with $\mathcal{P} = \mathcal{P}_T$.
\end{lemma}

To prove this lemma, we define a notion called an
$(\eps,\delta,k)$-ensemble,
analogous to the $(\eps,k)$-ensembles
defined in Section~\ref{sec:PMO-LB}. As before,
it is convenient to articulate this definition in the more general
setting of the \emph{feasible experts problem}, in which one
is given a set of arms $X$ (not necessarily a metric
space) along with a collection $\mF$ of Borel
probability measures on the set $[0,1]^X$ of
functions $\payoff : X \rightarrow [0,1].$  A problem
instance of the feasible experts problem
consists of a triple $(X,\mF,\prob)$
where $X$ and $\mF$ are known to the
algorithm, and $\prob \in \mF$ is not.

\begin{definition}\label{def:ensemble}
Consider a set $X$ and a $(k+1)$-tuple
$\vec{\prob} = (\prob_0,\prob_1 \LDOTS \prob_k)$
of Borel probability measures on $[0,1]^X$, the
set of $[0,1]$-valued payoff functions $\payoff$
on $X$.  For $0 \leq i \leq k$ and $x \in X$, let
$\mu_i(x)$ denote the expectation of $\payoff(x)$
under measure $\prob_i$.
We say that $\vec{\prob}$ is an \emph{$(\eps,\delta,k)$-ensemble}
if there exist pairwise disjoint subsets $S_1,S_2,\ldots,S_k \subseteq X$
for which the following properties hold:
\begin{itemize}
\item[(1)] 
for every $i>0$ and every event $\mathcal{E}$ in the Borel
$\sigma$-algebra of $[0,1]^X$, we have
    $$1-\delta < \prob_0(\mathcal{E}) / \prob_i(\mathcal{E}) < 1+\delta.$$
\item[(ii)] 
for every $i > 0$, we have
    $\sup(\mu_i, S_i) - \sup(\mu_i,\, X \setminus S_i) \geq \eps.$
\end{itemize}
\end{definition}

Essentially, the measures $\prob_1 \LDOTS \prob_k$ correspond to the children of any given node in the ball-tree. The precise connection to Construction~\ref{con:LB-payoffs} is stated below, derived as corollary of Lemma~\ref{lm:LB-experts-salient}.

\begin{corollary}\label{cor:ensemble}
Fix an arbitrary complete lineage $\lambda$ in a ball-tree $T$ and a tree node $u\in \mathbf{w}(\lambda)$. Let $u_1 \LDOTS u_k$ be the children of $u$. Let $u'$ be the unique child of $u$ contained in $\lambda$. Define lineage $\lambda_0 = \lambda \setminus \{u'\}$, and complete lineages
    $\lambda_i = \lambda_0\cup \{u_i\}$
for each $i\in [1,k]$. Then the tuple
    $\vec{\prob} = (\prob_{\lambda_0}, \prob_{\lambda_1} \LDOTS \prob_{\lambda_k})$
of probability measures from Construction~\ref{con:LB-payoffs}
constitutes a $(\eps,2\,\delta_j,k)$-ensemble where $j$ is the depth of the tree nodes $u_i$, $r$ is their radius, and $\eps = r \delta_j / 6$.
\end{corollary}

\begin{proof}
Let $S_1 \LDOTS S_k$ be the balls that correspond to $u_1 \LDOTS u_k$. Fix $u_i$, and apply Lemma~\ref{lm:LB-experts-salient} with lineage $\lambda_i$ and tree node $u_i$. Then both parts of Definition~\ref{def:ensemble} are satisfied for a given $i$. (For part (i), note that $\lambda_0 = \lambda_i \setminus \{u_i \}$. Observe that Lemma~\ref{lm:LB-experts-salient}(i) bounds
    $\prob_\lambda(\mE)/ \prob_{\lambda\setminus \{v\}}(\mE)$,
whereas for Definition~\ref{def:ensemble} we need to bound the inverse ratio; hence, the bound increases from $\delta_i$ to $2\cdot \delta_i$.)
\end{proof}

\begin{theorem}\label{thm:LB-technique}
Consider the feasible experts problem on $(X,\mF)$. Let  $\vec{\prob}$ be an $(\eps,\delta,k)$-ensemble with $\{\prob_1,\ldots,\prob_k\} \subseteq
\mF$ and $0<\eps,\delta<1/2$. Then for any
    $t < \ln(17k)/(2 \delta^2)$
and any experts algorithm \A, at least half of the measures $\prob_i$ have the property that
	$R_{(\A,\,\prob_i)}(t) \geq \eps t/2$.
\end{theorem}

\begin{note}{Remarks.}
To preserve the flow of the paper,
the proof of this theorem is deferred until
Appendix~\ref{sec:KL-divergence}, where the
relevant KL-divergence techniques are developed.
The proof of Theorem~\ref{thm:lower-bound} uses Theorem~\ref{thm:LB-technique} for $k=2$, and  the proof of Theorem~\ref{thm:intro-MaxMinLCD}  will use it again for large $k$.
\end{note}

\begin{proofof}{Lemma~\ref{lm:ball-tree-LB}}
Let us fix an experts algorithm \A\ and a function $g\in o(\sqrt{t})$, and consider the distribution over problem instances in Construction~\ref{con:LB-payoffs}. For each complete lineage
$\lambda$ and tree node $w\in \mathbf{w}(\lambda)$, let $w_1,w_2$ denote
the children of $w$ in the ball-tree,  and let $w'$ denote the unique child that belongs to $\lambda$. The three lineages
    $\lambda_0 = \lambda \setminus \{w'\}, \,
     \lambda_1 = \lambda_0 \cup \{w_1\}, \,
     \lambda_2 = \lambda_0 \cup \{w_2\}$
define a triple of probability measures
    $\vec{\prob} = (\prob_{\lambda_0}, \prob_{\lambda_1}, \prob_{\lambda_2})$.
By Corollary~\ref{cor:ensemble}, this triple constitutes an $(\eps,2\,\delta_i,2)$-ensemble where $i$ is the depth of $w_1,w_2$ in the ball-tree, $r$ is their radius, and
$\eps = r\, \delta_i / 6$. By Theorem~\ref{thm:LB-technique}
there exists a problem instance
    $\alpha(w) \in \{\prob_{\lambda_1},\prob_{\lambda_2}\}$
such that for any $t_i < \tfrac{1}{8}\,\ln(34) \cdot \delta_i^{-2}$ one has
	$$R_{(\A,\, \alpha(w))}(t_i) \geq \eps t_i / 2. $$
Taking $t_i \in (\tfrac14, \tfrac{1}{8}\,\ln(34))\cdot \delta_i^{-2}$ one has
	$$R_{(\A,\, \alpha(w))}(t_i) \geq \eps t_i / 2
           = r \delta_i t_i / 12
           > \tfrac{1}{24} r^*_i \sqrt{t_i},$$
where $r^*_i$ is the smallest radius among all depth-$i$ nodes in the ball-tree. Recalling that we chose $n_i$ large enough that
$g(n_i) < \tfrac{1}{24\,i}\, r_i^* \sqrt{n}$
for all $n > n_i$, and that $n_i = \delta_i^{-2}$,
 we see that
 $$ i \cdot g(t_i)
    < \tfrac{1}{24} r^*_i \sqrt{t_i}
        <R_{(\A, \, \alpha(w))}(t_i).
 $$

For each depth $i$, let us define $\mathcal{E}_i$ to be the set of
input distributions $\prob_{\lambda}$ such that $\lambda$ is
a complete lineage whose associated end
$\mathbf{w}(\lambda) = (w_0,w_1,\ldots)$ satisfies
$w_i = \alpha(w_{i-1})$.  Interpreting these sets as random
events under the probability distribution $\mathcal{P}_T$,
they are mutually independent events each having
probability $\tfrac12$. Furthermore, we have proved that there exists
a sequence of times
$t_i \rightarrow \infty$ such that for each $i$ we have
	$R_{(\A,\,\prob)}(t_i) > i \cdot g(t_i)$
for any $\prob \in \mathcal{E}_i$.

For each complete lineage  $\lambda$, define the ``smallest possible constant" if we were to characterize the algorithm's regret on problem instance $\prob_\lambda$ using function $g$:
\begin{align*}
 C_\lambda := \inf \{ C \leq \infty:\,
        R_{(\A,\,\prob_{\lambda})}(t) \leq C\, g(t) \text{~for all $t$} \}.
\end{align*}
Note that
        $R_{(\A,\,\prob_{\lambda})}(t) = O_\mu(g(t)) $
if and only if $C_\lambda < \infty$. We claim that
    $\Pr [C_\lambda < \infty] = 0$,
where the probability is over the random choice of complete
lineage $\lambda$.  Indeed,
if infinitely many events $\mathcal{E}_i$ happen,
then event $\{C_\lambda=\infty\}$ happens as well.  But the
probability that infinitely many events $\mathcal{E}_i$
happen is 1, because for every positive integer $n$,
   $\Pr \left[ \cap_{i=n}^{\infty}
      \overline{\mathcal{E}_i} \right] =
     \prod_{i=n}^{\infty} \Pr \left[ \overline{\mathcal{E}_i}
     \,\left|\, \cap_{j=n}^{i-1} \overline{\mathcal{E}_j} \right. \right]
     = 0.$
\end{proofof}

\subsection{Tractability for compact well-orderable metric spaces}
\label{sec:tractability}

In this section we prove the main algorithmic result.

\begin{theorem}\label{thm:main-alg}
Consider a compact well-orderable metric space $(X,\mD)$. Then:
\begin{OneLiners}
\item[(a)] the Lipschitz MAB problem on $(X,\mD)$ is $f$-tractable for every $f\in\omega(\log t)$;
\item[(b)] the Lipschitz experts problem on $(X,\mD)$ is 1-tractable, even with a double feedback.
\end{OneLiners}
\end{theorem}

We present a joint exposition for both the bandit and the experts version. Let us consider the Lipschitz MAB/experts problem on a compact metric space $(X,\mD)$ with a topological well-ordering $\prec$ and a payoff function $\mu$. For each strategy $x\in X$, let
	$S(x) = \{y\preceq x: y\in X\}$
be the corresponding initial segment of the well-ordering $(X,\prec)$. Let
	$\mu^* = \sup(\mu, X)$
denote the maximal payoff. Call a strategy $x\in X$ \emph{optimal} if $\mu(x) = \mu^*$.
We rely on the following structural lemma:

\begin{lemma}\label{lm:structural}
There exists an optimal strategy $x^*\in X$ for which it holds that
	$\sup(\mu, X\setminus S(x^*)) < \mu^*$.
\end{lemma}

\begin{proof}
Let $X^*$ be the set of all optimal strategies. Since $\mu$ is a continuous real-valued function on a compact space $X$, it attains its maximum, i.e. $X^*$ is non-empty, and furthermore $X^*$ is closed. Note that $\{S(x): x\in X^*\}$ is an open cover for $X^*$. Since $X^*$ is compact (as a closed subset of a compact set) this cover contains a finite subcover, call it $\{S(x): x\in Y^*\}$. Then the $\prec$-maximal element of $Y^*$  is the $\prec$-maximal element of $X^*$. The initial segment $S(x^*)$ is open, so its complement $ Y = X\setminus S(x^*)$ is closed and therefore compact. It follows that $\mu$ attains its maximum on $Y$, say at a point $y^*\in Y$. By the choice of $y^*$ we have $x^*\prec y^* $, so by the choice of $x^*$ we have $\mu(x^*)> \mu(y^*)$.
\end{proof}

In the rest of this section we let $x^*$ be the strategy from Lemma~\ref{lm:structural}. Our algorithm is geared towards finding $x^*$ eventually, and playing it from then on. The idea is that if we cover $X$ with balls of a sufficiently small radius, any strategy in a ball containing  $x^*$ has a significantly larger payoff than any strategy in a ball that overlaps with $X\setminus S(x^*)$.

The algorithm accesses the metric space and the well-ordering via the following two oracles.

\begin{definition}\label{def:covering-oracle}
A \emph{$\delta$-covering set} of a metric space $(X,\mD)$ is a subset $S\subset X$ such that each point in $X$ lies within distance $\delta$ from some point in $S$. An oracle $\mathcal{O} = \mathcal{O}(k)$ is a \emph{covering oracle} for $(X,\mD)$ if it inputs $k\in\N$ and outputs a pair $(\delta, S)$ where $\delta = \delta_\mathcal{O}(k)$ is a positive number and  $S$ is a $\delta$-covering set of $X$ consisting of at most $k$ points. Here $\delta_\mathcal{O}(\cdot)$ is any function such that
	$\delta_\mathcal{O}(k)\rightarrow 0$ as $k\rightarrow\infty$.
\end{definition}

\begin{definition}\label{def:ordering-oracle}
Given a metric space $(X,\mD)$ and a total order $(X,\prec)$, the \emph{ordering oracle} inputs a finite collection of balls (given by the centers and the radii), and returns the $\prec$-maximal element covered by the closure of these balls, if such element exists, and an arbitrary point in $X$ otherwise.
\end{definition}

\newcommand{\oracleX}{\ensuremath{x_{\mathrm{or}}}}
\newcommand{\muAv}{\ensuremath{\mu_{\mathrm{av}}}}
\newcommand{\explSub}{\ensuremath{\mathtt{EXPL}}}

Our algorithm is based on the following \emph{exploration subroutine} $\explSub()$.

\begin{algorithm}\label{alg:PMO-explore}
Subroutine $\explSub(k,n,r)$: inputs $k, n \in \N$ and $r\in (0,1)$, outputs a point in $X$.

First it calls the covering oracle $\mathcal{O}(k)$ and receives a $\delta$-covering set $S$ of $X$ consisting of at most $k$ points. Then it plays each strategy $x\in S$ exactly $n$ times; let $\muAv(x)$ be the sample average. Let us say that $x$ a \emph{loser} if
	$\muAv(y) - \muAv(x)> 2r  + \delta$
for some $y\in S$. Finally, it calls the ordering oracle with the collection of all closed balls 	 $\Bar{B}(x,\delta)$ such that $x$ is not a loser, and outputs the point $\oracleX\in X$ returned by this oracle call.
\end{algorithm}

Clearly, $\explSub(k,n,r)$ takes at most $kn$ rounds to complete. We show that for sufficiently large $k,n$ and sufficiently small $r$ it returns $x^*$ with high probability.

\begin{lemma}\label{lm:tractability}
Fix a problem instance and let $x^*$ be the optimal strategy from Lemma~\ref{lm:structural}.
Consider increasing functions $k,n,T: \N\to \N$ such that
    $r(t) := 4\sqrt{ (\log T(t))\, / n(t)}  \to 0$.
Then for any sufficiently large $t$, with probability at least $1-T^{-2}(t)$, the subroutine
    $\explSub(k(t),\,n(t),\, r(t))$ returns $x^*$.
\end{lemma}
\begin{proof}
Let us use the notation from Algorithm~\ref{alg:PMO-explore}. Fix $t$ and consider a run of
    $\explSub(k(t),\,n(t),\, r(t))$.
Call this run \emph{clean} if for each $x\in S$ we have $|\muAv(x)-\mu(x)| \leq r(t)$. By Chernoff Bounds, this happens with probability at least $1-T^{-2}(t)$. In the rest of the proof, let us assume that the run is clean.

Let $\Bar{B}$ be the union of the closed balls $\Bar{B}(x,\delta)$, $x\in S^*$. Then the ordering oracle returns the $\prec\mbox{-maximal}$ point in $\Bar{B}$ if such point exists. We will show that
	$x^*\in \Bar{B}\subset S(x^*)$
for any sufficiently large $t$, which will imply the lemma.

We claim that $x^* \in \Bar{B}$. Since $S$ is a $\delta$-covering set, there exists $y^*\in S$ such that
	$\mD(x^*, y^*)\leq \delta$.
Let us fix one such $y^*$. It suffices to prove that $y^*$ is not a loser. Indeed, if 	
	$\muAv(y) - \muAv(y^*)> 2\,r(t) + \delta $
for some $y\in S$ then
	$\mu(y) > \mu(y^*) + \delta \geq \mu^*$,
contradiction. Claim proved.

Let
	$\mu_0 = \sup(\mu, X\setminus S(x^*))$
and let $r_0 = (\mu^*-\mu_0)/7$. Let us assume that $t$ is sufficiently large so that $r(t) < r_0$ and
	$\delta = \delta_{\mathcal{O}}(k(t))<r_0$,
where $\delta_{\mathcal{O}}(\cdot)$ is from the definition of the covering oracle.

We claim that $\Bar{B}\subset S(x^*)$. Indeed, consider
	$x\in S$ and $y\in X\setminus S(x^*)$ such that $\mD(x, y) \leq \delta$.
It suffices to prove that $x$ is a loser. Consider some $y^*\in S$ such that $\mD(x^*,y^*)\leq \delta$. Then by the Lipschitz condition
\begin{align*}
\muAv(y^*)
	& \geq \mu(y^*) - r_0 \geq \mu^* -  2 r_0, \\
\muAv(x)
	&\leq \mu(x) + r_0 \leq \mu(y) + r_0
	\leq \mu_0 + 2r_0 \leq \mu^* - 5 r_0 \\
\muAv(y^*) - \muAv(x)
	& \geq 3 r_0 > 2 r(t) + \delta. \qedhere
\end{align*}
\end{proof}

\begin{proofof}{Theorem~\ref{thm:main-alg}}
Let us fix a function $f\in \omega(\log t)$. Then $f(t) = \alpha(t) \log(t)$ where $\alpha(t)\to\infty$. Without loss of generality, assume that $\alpha(t)$ is non-decreasing. (If not, then instead of $f(t)$ use
    $g(t) = \beta(t) \log(t)$,
where
    $\beta(t) = \inf \{ \alpha(t'):\, t'\geq t \}$.)

For part (a), define
    $k_t = \flr{\sqrt{g(t)/ \log t}} $,
    $n_t = \flr{k_t \log t}$,
and
    $r_t = 4\sqrt{(\log t)/ n_t}$.
Note that $r_t \to 0$.

The algorithm proceeds in phases of a doubly exponential length\footnote{The doubly exponential phase length is necessary in order to get $f$-tractability.  If we employed the more familiar \emph{doubling trick} of using phase length $2^i$ (as in \citep{bandits-exp3,Bobby-nips04,LipschitzMAB-stoc08}, for example) then the algorithm would only be $f(t)\, \log t$-tractable.}. A given phase  $i=1,2,3,\ldots$ lasts for $T = 2^{2^i}$ rounds. In this phase, first we call the exploration subroutine
    $\explSub(k_T,\,n_T,\, r_T)$.
Let $\oracleX\in X$ be the point returned by this subroutine. Then we play \oracleX\ till the end of the phase. This completes the description of the algorithm.

Fix a problem instance $\mathcal{I}$. Let $W_i$ be the total reward accumulated by the algorithm in phase $i$, and let
	$R_i = 2^{2^i}\,\mu^* - W_i$
be the corresponding share of regret.
By Lemma~\ref{lm:tractability} there exists $i_0 = i_0(\mathcal{I})$ such that for any phase $i\geq i_0$  we have, letting $T=2^{2^i}$ be the phase duration, that $R_i \leq k_T\, n_T \leq g(T)$ with probability at least $1-T^{-2}$, and therefore
	$E[R_i] \leq g(T) + T^{-1}$.
For any $t > t_0 = 2^{2^{i_0}}$ it follows by summing over
$i \in \{i_0,i_0+1,\ldots,\lceil \log \log t \rceil\}$
that
        $R_{\A,\,\mathcal{I}}(t) = O(t_0 + g(t)).$
Note that we have used the fact that $\alpha(t)$ is non-decreasing.

For part (b), we separate exploration and exploitation. For exploration, we run  $\explSub()$ on the \emph{free peeks}. For exploitation, we use the point returned by $\explSub()$ in the previous phase. Specifically, define
    $k_t = n_t = \flr{\sqrt{t}} $,
and
    $r_t = 4\sqrt{(t^{1/4})/ n_t}$.
The algorithm proceeds in phases of exponential length. A given phase $i=1,2,3,\ldots$ lasts for $T = 2^i$ rounds. In this phase, we run the exploration subroutine
    $\explSub(k_T,\,n_T,\, r_T)$
on the \emph{free peeks}. In each round, we \emph{bet} on the point returned by $\explSub()$ in the previous phase. This completes the description of the algorithm.

By Lemma~\ref{lm:tractability} there exists $i_0 = i_0(\mathcal{I})$ such that in any phase $i\geq i_0$ the algorithm incurs zero regret with probability at least $1-e^{\Omega(i)}$. Thus the total regret after $t>2^{i_0}$ rounds is at most $t_0 + O(1)$.
\end{proofof}

\subsection{The $(\log t)$-intractability for infinite metric spaces: proof of Theorem~\ref{thm:logT}}
\label{sec:logT}

Consider an infinite metric space $(X,\mD)$. In view of Theorem~\ref{thm:boundary-of-tractability}, we can assume that the completion $X^*$ of $X$ is compact. It follows that there exists $x^*\in X^*$ such that $x_i\to x^*$ for some sequence $x_1, x_2,\, \ldots\, \in X$. Let $r_i = \mD(x_i, x^*)$. Without loss of generality, assume that
    $r_{i+1} < \tfrac12\, r_i$
for each $i$, and that the diameter of $X$ is $1$.

Let us define an ensemble of payoff functions
    $\mu_i :X\to[0,1]$, $i\in\N$,
where $\mu_0$ is the ``baseline" function, and for each $i\geq 1$ function $\mu_i$ is the ``counterexample" in which a neighborhood of $x_i$ has slightly higher payoffs.
The ``baseline" is defined by
    $\mu_0(x) = \tfrac12 - \tfrac{\mD(x,x^*)}{8}$,
and the ``counterexamples" are given by
$$ \mu_i(x) = \mu_0(x)+ \nu_i(x),
\text{~~where~~}
 \nu_i(x) = \tfrac{3}{4} \max \left(0, \tfrac{r_i}{3} - \mD(x,x^*) \right).
$$
Note that both $\mu_0$ and $\nu_i$ are $\tfrac18$-Lipschitz
and $\tfrac34$-Lipschitz w.r.t. $(X,\mD)$, respectively,
so $\mu_i$ is $\tfrac78$-Lipschitz w.r.t $(X,\mD)$.
Let us fix a MAB algorithm \A\ and assume that it is $(\log t)$-tractable. Then for each $i\geq 0$ there exists a constant $C_i$ such that
    $R_{(\A,\, \mu_i)}(t) < C_i \log t$
for all times $t$. We will show that this is not possible.

Intuitively, the ability of an algorithm to distinguish between payoff functions $\mu_0$ and $\mu_i$, $i\geq 1$  depends on the number of samples in the ball
    $B_i = B(x_i,\, r_i/3)$.
(This is because $\mu_0 = \mu_i$ outside $B_i$.) In particular, the number of samples itself cannot be too different under $\mu_0$ and under $\mu_i$, \emph{unless it is large}. To formalize this idea, let  $N_i(t)$ be the number of times algorithm \A\ selects a strategy in the ball $B_i$ during the first $t$ rounds, and let $\sigma(N_i(t))$ be the corresponding $\sigma$-algebra. Let $\prob_i[\cdot]$ and $\mathbb{E}_i[\cdot]$ be, respectively, the distribution and expectation induced by $\mu_i$. Then we can connect $\mathbb{E}_0[N_i(t)]$ with the probability of any event $S\in \sigma(N_i(t))$ as follows.

\begin{claim}\label{cl:logT-KLdiv}
For any $i\geq 1$ and any event $S\in \sigma(N_i(t))$ it is the case that
\begin{align}\label{eq:logT-KLdiv}
 \prob_i[S] < \tfrac13 \leq \prob_0[S] \quad \Rightarrow \quad
 - \ln(\prob_i[S]) - \tfrac{3}{e}
    \leq O(r_i^2)\; \mathbb{E}_0[N_i(t)].
\end{align}
\end{claim}

\begin{note}{Remark.}
The reason our argument proves the regret lower bound in terms of $\log(t)$, rather than some other function of $t$, is the $\ln(\cdot)$ term in \refeq{eq:logT-KLdiv}, which in turn comes from the $\exp(\cdot)$ term in Claim~\ref{lem:kl-distinguishing} (which captures a crucial property of KL-divergence).
\end{note}

\noindent
Claim~\ref{cl:logT-KLdiv} is proved using KL-divergence techniques, see Appendix~\ref{sec:KL-divergence} for details. To complete the proof of the theorem, we claim that for each $i\geq 1$ it is the case that
    $\mathbb{E}_0 [N_i(t)] \geq \Omega(r_i^{-2}\, \log t)$
for any sufficiently large $t$. Indeed, fix $i$ and let
    $S = \{N_i(t) < r_i^{-2} \log t\}$.
Since
    $$C_i \log t > R_{(\A,\; \mu_i)}(t) \geq
        \prob_i(S)\,(t - r_i^{-2} \log t) \tfrac{r_i}{8} ,$$
it follows that
    $\prob_i(S) < t^{-1/2} < \tfrac13$
for any sufficiently large $t$. Then by Claim~\ref{cl:logT-KLdiv} either
    $\prob_0(S) < \tfrac13$
or the consequent in~\refeq{eq:logT-KLdiv} holds. In both cases
        $\mathbb{E}_0 [N_i(t)] \geq \Omega(r_i^{-2}\, \log t)$.
Claim proved.

Finally, the fact that $\mu_0(x^*) - \mu_0(x) \geq r_i/12$
for every $x \in B_i$ implies that
     $R_{(\A, \, \mu_0)}(t) \geq \tfrac{r_i}{12} \mathbb{E}_0[N_i(t)]
      \geq \Omega(r_i^{-1} \, \log t)$
which establishes Theorem~\ref{thm:logT} since $r_i^{-1} \rightarrow
\infty$ as $i \rightarrow \infty$.

\subsection{Tractability via more intuitive oracle access}
\label{sec:simpler-alg}

\newcommand{\LimSet}{\ensuremath{\text{\sc lim}}}
\newcommand{\Decomposable}{Cantor-Bendixson}

In Theorem~\ref{thm:main-alg}, the algorithm accesses the metric space via two oracles: a very intuitive \emph{covering oracle}, and a less intuitive \emph{ordering oracle}. In this section we show that for a wide family of metric spaces --- including, for example, compact metric spaces with a finite number of limit points --- the ordering oracle is not needed: we provide an algorithm which accesses the metric space via a finite set of covering oracles. We will consider metric spaces of finite \emph{Cantor-Bendixson rank}, a classic notion from point topology.

\begin{definition}\label{def:CB-rank}
Fix a metric space $(X,\mD)$. If for some $x\in X$ there exists a sequence of points in $X\setminus \{x\} $ which converges to $x$, then $x$ is called a \emph{limit point}. For $S\subset X$ let $\LimSet(S)$ denote the \emph{limit set}: the set of all limit points of $S$. Let
	$\LimSet(S,0) = S$,
and
	$\LimSet(S,i) = \LimSet(\LimSet(\cdots \LimSet(S)))$,
where $\LimSet(\cdot)$ is applied $i$ times.
The \emph{Cantor-Bendixson rank} of $(X,\mD)$ is defined as
    $\sup \{n: \LimSet(X,n) \neq \emptyset\}$.
\end{definition}

Let us say that a \emph{\Decomposable\ metric space} is one with a finite Cantor-Bendixson rank. In order to apply Theorem~\ref{thm:main-alg}, we show that any such metric space is well-orderable.

\begin{lemma}
Any \Decomposable\ metric space is well-orderable.
\end{lemma}

\begin{proof}
Any finite metric space is trivially well-orderable. To prove the lemma, it suffices to show the following: any metric space $(X,\mD)$ is well-orderable  if so is $(\LimSet(X),\mD)$.

Let $X_1 = X\setminus  \LimSet(X)$ and $X_2 = \LimSet(X)$. Suppose $(X_2,\mD)$ admits a topological well-ordering $\prec_2$. Define a binary relation $\prec$ on $X$ as follows. Fix an arbitrary well-ordering $\prec_1$ on $X_1$. For any $x,y\in X$ posit $x\prec y$ if either
(i) $x,y\in X_1$ and $x\prec_1 y$, or
(ii) $x,y\in X_2$ and $x\prec_2 y$, or
(iii) $x\in X_1$ and $y\in X_2$.
It is easy to see that $(X,\prec)$ is a well-ordering.

It remains to prove that an arbitrary initial segment
	$Y = \{ x\in X: x\prec y\}$
is open in $(X,\mD)$. We need to show that for each $x\in Y$ there is a ball $B(x,\eps)$, $\eps>0$ which is contained in $Y$. This is true if $x\in X_1$ since by definition each such $x$ is an isolated point in $X$. If $x\in X_2$ then
	$Y = X_1 \cup Y_2$
where $Y_2 = \{ x\in X_2: x \prec_2 y\}$
is the initial segment of $X_2$. Since $Y_2$ is open in $(X_2,\mD)$, there exists $\eps>0$ such that $B_{X_2}(x,\eps) \subset Y_2$. It follows that
	$B_X(x,\eps) \subset B_{X_2}(x,\eps) \cup X_1 \subset Y$.
\end{proof}

The structure of a \Decomposable\ metric space is revealed by a partition of $X$ into subsets
	$X_i = \LimSet(X,i) \setminus \LimSet(X,i+1)$, $0\leq i\leq n $.
For a point $x\in X_i$, we define the \emph{rank} to be $i$. The algorithm requires a covering oracle for each $X_i$.

\begin{theorem}\label{thm:lim-decomposable}
Consider the Lipschitz MAB/experts problem on a compact metric space $(X,\mD)$ such that
	$\LimSet_N(X)=\emptyset$
for some $N$.  Let $\mathcal{O}_i$ be the covering oracle for
	$X_i = \LimSet(X,i)\setminus \LimSet(X,i+1)$.
Assume that access to the metric space is provided only via the collection of oracles
	$\{\mathcal{O}_i \}_{i=0}^N$.
Then:
\begin{OneLiners}
\item[(a)] the Lipschitz MAB problem on $(X,\mD)$ is $f$-tractable for every $f\in\omega(\log t)$;
\item[(b)] the Lipschitz experts problem on $(X,\mD)$ is 1-tractable, even with a double feedback.
\end{OneLiners}
\end{theorem}

\OMIT{ 
\begin{theorem}\label{thm:lim-decomposable}
Consider a compact metric space $(X,\mD)$ such that
	$\LimSet_n(X)=\emptyset$
for some $n$.  Let $\mathcal{O}_i$ be the covering oracle for
	$X_i = \LimSet(X,i)\setminus \LimSet(X,i+1)$.
Then for each $f\in \omega(\log t)$ there exists an $f(t)$-tractable bandit algorithm which accesses $(X,\mD)$ via the collection of oracles
	$\{\mathcal{O}_i \}_{i=0}^n$.
\end{theorem}
} 

In the rest of this section, consider the setting in Theorem~\ref{thm:lim-decomposable}. We describe the \emph{exploration subroutine} $\explSub'()$, which is similar to $\explSub()$ in Section~\ref{sec:tractability} but does not use the ordering oracle. Then we prove a version of Lemma~\ref{lm:tractability} for $\explSub'()$. Once we have this lemma, the proof of Theorem~\ref{thm:lim-decomposable} is identical to that of Theorem~\ref{thm:main-alg} (and is omitted).

\begin{algorithm}\label{alg:LIM}
Subroutine $\explSub'(k,n,r)$: inputs $k, n \in \N$ and $r\in (0,1)$, outputs a point in $X$.

Call each covering oracle $\mathcal{O}_i(k)$ and receive a $\delta_i$-covering set $S_i$ of $X$ consisting of at most $k$ points. Let
	$S = \cup_{l=1}^n S_l$.
Play each strategy $x\in S$ exactly $n$ times; let $\muAv(x)$ be the corresponding sample average. For $x,y\in S$, let us say that \emph{$x$ dominates $y$} if
	$\muAv(x) - \muAv(y)> 2\,r$.
Call $x\in S$ a \emph{winner} if $x$ has a largest rank among the strategies that are not dominated by any other strategy. Output an arbitrary winner if a winner exists, else output an arbitrary point in $S$.
\end{algorithm}

Clearly, $\explSub(k,n,r)$ takes at most $knN$ rounds to complete. We show that for sufficiently large $k,n$ and sufficiently small $r$ it returns an optimal strategy with high probability.

\begin{lemma}\label{lm:LIM-tractability}
Fix a problem instance. Consider increasing functions $k,n,T: \N\to \N$ such that
    $r(t) := 4\sqrt{ (\log T(t))\, / n(t)}  \to 0$.
Then for any sufficiently large $t$, with probability at least $1-T^{-2}(t)$, the subroutine
    $\explSub'(k(t),\,n(t),\, r(t))$ returns an optimal strategy.
\end{lemma}

\begin{proof}
Use the notation from Algorithm~\ref{alg:LIM}. Fix $t$ and consider a run of
    $\explSub'(k(t),\,n(t),\, r(t))$.
Call this run \emph{clean} if for each $x\in S$ we have $|\muAv(x)-\mu(x)| \leq r(t)$. By Chernoff Bounds, this happens with probability at least $1-T^{-2}(t)$. In the rest of the proof, let us assume that the run is clean.

Let us introduce some notation. Let $\mu$ be the payoff function and let $\mu^* = \sup(\mu,X)$. Call $x\in X$ \emph{optimal} if $\mu(x) = \mu^*$. (There exists an optimal strategy since $(X,\mD)$ is compact.) Let $i^*$ be the largest rank of any optimal strategy. Let $X^*$ be the set of all optimal strategies of rank $i^*$. Let $Y = \LimSet(X,i^*)$. Since each point $x\in X_{i^*}$ is an isolated point in $Y$, there exists some $r(x)>0$ such that $x$ is the only point of $B(x,r(x))$ that lies in $Y$.

We claim that $\sup(\mu, Y\setminus X^*)<\mu^*$. Indeed, consider
	$C= \cup_{x\in X^*} B(x,r(x)) $.
This is an open set. Since $Y$ is closed, $Y\setminus C$ is closed, too, hence compact. Therefore there exists $y\in Y\setminus C$ such that
	$\mu(y) = \sup(\mu, Y\setminus C)$.
Since $X^* \subset C$, $\mu(y)$ is not optimal, i.e. $\mu(y)<\mu^*$. Finally, by definition of $r(x)$ we have $Y\setminus C = Y\setminus X^*$. Claim proved.

Pick any $x^*\in X^*$. Let $\mu_0 = \sup(\mu, Y\setminus X^*)$.
Assume that $t$ is large enough so that $r(t) < (\mu^*-\mu_0)/4$ and $\delta_{i^*}<r(x^*)$. Note that the $\delta_{i^*}$-covering set $S_{i^*}$ contains $x^*$.

Finally, we claim that in a clean phase, $x^*$ is a winner, and all winners lie in $X^*$. Indeed, note that $x^*$ dominates any non-optimal strategy $y\in S$ of larger or equal rank, i.e. any $y\in S\cap (Y\setminus X^*)$. This is because
$ \muAv(x^*) - \muAv(y) \geq \mu^* - \mu_0 - 2r > 2.
$
The claim follows since any optimal strategy cannot be dominated by any other strategy.
\end{proof}

\section{Boundary of tractability: proof of Theorem~\ref{thm:boundary-of-tractability}}
\label{sec:boundary-body}

We prove that Lipschitz bandits/experts are $o(t)$-tractable if and only if the completion of the metric space is compact. More formally, we prove Theorem~\ref{thm:boundary-of-tractability} (which subsumes  Theorem~\ref{thm:boundary-of-tractability-bandits}). We restate the theorem below for the sake of convenience.

\begin{theorem*}[Theorem~\ref{thm:boundary-of-tractability} restated]
The \FFproblem on metric space $(X,\mD)$ is either $f(t)$-tractable for some $f\in o(t)$, even in the bandit setting, or it is not $g(t)$-tractable for any $g\in o(t)$, even with full feedback. The former occurs if and only if the completion of $X$ is a compact metric space.
\end{theorem*}

First, we reduce the theorem to that on complete metric spaces, see Appendix~\ref{sec:reduction}. In what follows, we will use a basic fact that a complete metric space is compact if and only if for any $r>0$, it can be covered by a finite number of balls of radius $r$.

\xhdr{Algorithmic result.}
We consider a compact metric space $(X,\mD)$ and use an extension of algorithm \NaiveAlg (described in the Introduction). In each phase $i$ (which lasts for $t_i$ round) we fix a covering of $X$ with $N_i < \infty$ balls of radius $2^{-i}$ (such covering exists by compactness), and run a fresh instance of the $N_i$-armed bandit algorithm \UCB\ from \citet{bandits-ucb1} on the centers of these balls.  (This algorithm is for the ``basic'' MAB problem,
in the sense that it does not look at the distances in the metric
space.) The phase durations $t_i$ need to be tuned to the $N_i$'s. In the setting considered in \cite{Bobby-nips04} (essentially, bounded covering dimension) it suffices to tune each $t_i$ to the corresponding $t_i$ in a fairly natural way. The difficulty in the present setting is that there are no guarantees on how fast the $N_i$'s grow. To take this into account, we fine-tune each $t_i$ to (essentially) all covering numbers $N_1 \LDOTS N_{i+1}$.

Let $R_k(t)$ be the expected regret accumulated by the algorithm in the first $t$ rounds of phase $k$. Using the off-the-shelf regret guarantees for \UCB, it is easy to see \citep{Bobby-nips04} that
\begin{align}\label{eq:app-boundary-UB-Rk}
 R_k(t) \leq O(\sqrt{N_k\, t \log t}) + \eps_k\, t
\leq \eps_k\, \max(t^*_k,\, t),
    \text{~~where~~} t^*_k = 2\,\tfrac{N_k}{\eps_k^2} \log\tfrac{N_k}{\eps_k^2}.
\end{align}

Let us specify phase durations $t_i$. They are defined very differently from the ones in \citep{Bobby-nips04}. In particular, in \cite{Bobby-nips04} each $t_i$ is fine-tuned to the corresponding covering number $N_i$ by setting
    $t_i = t^*_i$,
and the analysis works out for metric spaces of bounded covering dimension. In our setting, we fine-tune each $t_i$ to (essentially) all covering numbers $N_1 \LDOTS N_{i+1}$. Specifically, we define the $t_i$'s inductively as follows:
$$ t_i = \min(t^*_i,\, t^*_{i+1},\, 2\,
    \textstyle{\sum_{j=1}^{i-1}} t_j).
$$
This completes the description of the algorithm, call it \A.

\begin{lemma}
Consider the Lipschitz MAB problem on a compact and complete metric space $(X,\mD)$. Then $R_\A(t) \leq 5\, \eps(t)\, t$, where
    $\eps(t) = \min\{ 2^{-k}:\, t\leq s_k\}$
and
    $s_k = \sum_{i=1}^k\, t_i$.
In particular, $R_\A(t) = o(t)$.
\end{lemma}

\begin{proof}
First we claim that
    $R_\A(s_k) \leq 2\, \eps_k\, s_k$
for each $k$. Use induction on $k$. For the induction base, note that
    $R_\A(s_1) = R_1 (t_1) \leq \eps_1 t_1$
by~\refeq{eq:app-boundary-UB-Rk}. Assume the claim holds for some $k-1$. Then
\begin{align*}
R_\A(s_k)
    &= R_\A(s_{k-1}) + R_k(t_k) \\
    &\leq 2\,\eps_{k-1}\, s_{k-1} + \eps_k\, t_k \\
    &\leq 2\, \eps_k (s_{k-1} + t_k)
    = 2\,\eps_k s_k,
\end{align*}
claim proved. Note that we have used~\refeq{eq:app-boundary-UB-Rk} and the facts that
    $t_k \geq t^*_k$ and $t_k \geq 2\, s_{k-1}$.

For the general case, let $T = s_{k-1} + t$, where $t\in (0, t_k)$. Then by~\refeq{eq:app-boundary-UB-Rk} we have that
\begin{align*}
R_k(t)
    &\leq \eps_k\, \max(t^*_k,\, t) \\
    &\leq \eps_k\, \max(t_{k-1},\, t)
    \leq \eps_k\, T,\\
R_\A(T) &= R_\A(s_{k-1}) + R_k(T)  \\
    &\leq 2\, \eps_{k-1}\, s_{k-1} + \eps_k\,T
    \leq 5\,\eps_k\, T. \qedhere
\end{align*}
\end{proof}

\xhdr{Lower bound: proof sketch.} For the lower bound, we consider a metric space $(X,\mD)$ with an infinitely many disjoint balls $B(x_i,r_*)$ for some $r_*>0$. For each ball $i$ we define the \emph{wedge function} supported on this ball:
$$ G_{(i,r)}(x) = \begin{cases}
    \min\{ r_* - \mD(x,x_i),\; r_* - r \} & \mbox{if $x \in B(x_i,r_*)$} \\
         0 & \mbox{otherwise}.
\end{cases} $$
The balls are partitioned into two infinite sets: the
\emph{ordinary} and \emph{special} balls.
The random payoff function is then defined by taking a
constant function, adding the wedge function on each
special ball, and randomly adding
or subtracting the wedge function on each
ordinary ball.
Thus, the expected payoff is constant throughout the
metric space except that it assumes higher values
on the special balls.  However,
the algorithm has no chance of ever finding these balls,
because at time $t$ they are statistically indistinguishable
from the $2^{-t}$ fraction of ordinary balls that randomly
happen to never subtract their wedge function during the
first $t$ steps of play.

\xhdr{Lower bound: full proof.}
Suppose $(X,\mD)$ is not compact. Fix $r>0$ such that $X$ cannot be covered by a finite number of balls of radius $r$. There exists a countably infinite subset $S\subset X$ such that the balls $B(x,r)$, $x\in S$ are mutually disjoint. (Such subset can be constructed inductively.)
Number the elements of $S$ as $s_1,s_2,\ldots,$ and denote
the ball $B(s_i,r)$ by $B(i)$.

Suppose there exists a Lipschitz experts algorithm
$\A$ that is $g(t)$-tractable for some $g\in o(t)$.
Pick an increasing sequence
    $t_1, t_2, \ldots \in \N$
such that $t_{k+1} > 2 t_k \geq 10$ and
    $g(t_k)< r_k\, t_k /k$
for each $k$, where
    $r_k = r/2^{k+1}$.
Let $m_0=0$ and
    $m_k = \sum_{i=1}^{k} 4^{t_i}$
for $k>0$, and let $I_k = \{m_{k}+1,\ldots,m_{k+1}\}.$
The intervals $I_k$ form a partition of $\N$ into
sets of sizes $4^{t_1},4^{t_2},\ldots$.
For every $i \in \N$, let $k$ be the unique value
such that $i \in I_k$ and define the following
Lipschitz function supported in $B(s_i,r)$:
    $$ G_i(x) = \begin{cases}
         \min\{ r - \mD(x,s_i), r - r_k \} & \mbox{if $x \in B(i)$} \\
         0 & \mbox{otherwise}.
       \end{cases} $$
If $J \subseteq \N$ is any set of natural numbers,
we can define a distribution $\prob_J$ on payoff
functions by sampling independent, uniformly-random
signs $\sigma_i \in \{ \pm 1 \}$ for every $i \in \N$
and defining the payoff function to be
    $$ \payoff = \tfrac12 + \textstyle{\sum_{i \in J}}\, G_i
       + \textstyle{\sum_{i \not\in J}}\, \sigma_i G_i.
    $$
Note that the distribution $\prob_J$ has expected
payoff function $\mu = \tfrac12 + \sum_{i \in J} G_i.$
Let us define a distribution $\mathcal{P}$ over problem
instances $\prob_J$ by letting $J$ be a random
subset of $\N$ obtained by sampling exactly one element $j_k$
of each set $I_k$ uniformly at random, independently
for each $k$.

Intuitively, consider an algorithm that
is trying to discover the value of $j_k$.  Every time a payoff function
$\pi_t$ is revealed, we get to see a random $\{ \pm 1 \}$
sample at every element of $I_k$ and we can eliminate
the possibility that $j_k$ is one of the elements that
sampled $-1$.  This filters out about half the elements
of $I_k$ in every time step, but $|I_k| = 4^{t_k}$ so
on average it takes $2 t_k$ steps before we can discover the identity
of $j_k$.  Until that time, whenever we play a strategy
in $\cup_{i \in I_k} B(i)$, there is a constant probability
that our regret is at least $r_k$.  Thus our regret
is bounded below by $r_k t_k \geq k g(t_k).$  This
rules out the possibility of a $g(t)$-tractable
algorithm.  The following lemma makes this argument
precise.

\begin{lemma}\label{lm:app-boundary}
$\Pr_{\prob \in\mathcal{P}} [ R_{(\A, \, \prob)} (t) = O_\mu( g(t)) ] = 0$.
\end{lemma}
\begin{proof}
Let $j_1,j_2,\ldots$ be the elements of the
random set $J$, numbered
so that $j_k \in I_k$ for all $k$.
For any $i,t \in \N$, let $\sigma(i,t)$ denote
the value of $\sigma_i$ sampled at time $t$
when sampling the sequence of i.i.d.~payoff
functions $\payoff_t$ from distribution $\prob_J$.
We know that $\sigma(j_k,t)=1$ for all $t$.
In fact if $S(k,t)$ denotes the set of all
$i \in I_k$ such that
    $\sigma(i,1) = \sigma(i,2) = \cdots = \sigma(i,t) = 1$
then conditional on the value of the set $S(k,t)$, the
value of $j_k$ is distributed uniformly at random
in $S(k,t)$.  As long as this set $S(k,t)$
has at least $n$ elements, the probability
that the algorithm picks a strategy $x_t$ belonging
to $B(j_k)$ at time $t$ is bounded above by $\tfrac{1}{n}$,
even if we condition on the event that
    $x_t \in \cup_{i \in I_k} B(i).$
For any given $i \in I_k \setminus \{j_k\}$, we have
    $\prob_J(i \in S(k,t)) = 2^{-t}$
and these events are independent for different
values of $i$.  Setting $n=2^{t_k}$, so that
$|I_k|=n^2$, we have
    \begin{align}
      \nonumber
    \prob_J \left[\, |S(k,t)| \leq n \,\right] &\leq
    \textstyle{\sum_{R \subset I_k,\; |R|=n}}\; \prob_J[\, S(k,t) \subseteq R \,] \\
      \nonumber
    &= \binom{n^2}{n}
    \left( 1-2^{-t} \right)^{n^2-n}
      \nonumber
    < \left( n^2 \cdot \left( 1 - 2^{-t} \right)^{n-1} \right)^{n} \\
      \label{eq:yucky-1}
    &< \exp \left( n (2 \ln(n) - (n-1)/2^t) \right).
    \end{align}
As long as $t \leq t_{k-1}$, the relation $t_k > 2t$
implies $(n-1)/2^t > \sqrt{n}$ so the expression
\eqref{eq:yucky-1} is bounded above by
$\exp \left( -n \sqrt{n} + 2 n \ln(n) \right)$,
which equals
$\exp \left( -8^{t_k} + 2 \ln(4) t_k 4^{t_k} \right)$
and is in turn bounded above by $\exp \left( -8^{t_k}/2 \right).$

Let $B(j_{>k})$ denote the union
    $B(j_{k+1}) \cup B(j_{k+2}) \cup \ldots,$
and let $N(t,k)$ denote the random variable
that counts the number of times \A\ selects
a strategy in $B(j_{>k})$ during rounds $1,\ldots,t$.
We have already demonstrated that for all $t \leq t_k$,
    \begin{equation} \label{eq:yucky-2}
       \Pr_{\prob_J \in \mathcal{P}}(x_t \in B(j_{>k}))
       \leq 2^{-t_{k+1}} + \sum_{\ell > k}
       \exp \left( -8^{t_\ell}/2 \right) < 2^{1-t_{k+1}},
    \end{equation}
where the term $2^{-t_{k+1}}$ accounts for the
event that $S(\ell,t)$ has at least $2^{t_{k+1}}$ elements,
where $\ell$ in the index of the set $I_{\ell}$ containing
the number $i$ such that $x_t \in B(i)$, if such an $i$
exists.
\eqref{eq:yucky-2} implies the bound
    $\mathbb{E}_{\prob_J \in \mathcal{P}}[N(t_k,k)] < t_k \cdot 2^{1-t_{k+1}}.$
By Markov's inequality, the probability that $N(t_k,k) > t_k/2$
is less than $2^{2-t_{k+1}}$.  By Borel-Cantelli, almost
surely the number of $k$ such that $N(t_k,k) \leq t_k/2$
is finite.
The algorithm's expected regret at time $t$ is
bounded below by $r_k(t_k - N(t_k,k))$, so with
probability $1$, for all but finitely many $k$ we have
$R_{(\A, \, \prob_J)}(t_k) \geq r_k t_k / 2 \geq (k/2) g(t_k).$
This establishes that \A\ is not $g(t)$-tractable.
\end{proof}

\section{Lipschitz experts in a (very) high dimension}
\label{sec:FFproblem}

This section concerns polynomial regret results for Lipschitz experts in metric spaces of (very) high dimension: Theorem~\ref{thm:intro-LCD}, Theorem~\ref{thm:intro-LCD-uniform}, and Theorem~\ref{thm:intro-MaxMinLCD}, as outlined in Section~\ref{sec:intro-experts}.

\OMIT{ 
Fix a metric space $(X,\mD)$. For a subset $Y\subset X$ and $\delta>0$, a \emph{$\delta$-covering} of $Y$ is a collection of sets of diameter at most $\delta$ whose union contains $Y$. A subset $S\subset X$ is a \emph{$\delta$-hitting set} for $Y$ if
    $Y \subset \cup_{x\in S}\, B(x,\,\delta)$.
(So if $S$ is a hitting set for some $\delta$-covering of $Y$ then it is a $\delta$-hitting set for $Y$.)
} 

\subsection{The uniform mesh (proof of Theorem~\ref{thm:intro-LCD})}
\label{sec:FFproblem-LCD}

We start with a version of algorithm \NaiveAlg discussed in the Introduction.%
\footnote{A similar algorithm has been used by \citet{Anupam-experts07} to obtain regret $R(T) = O(\sqrt{T})$ for metric spaces of finite covering dimension.}
This algorithm, called $\NaiveExp(b)$, is parameterized by $b>0$. It runs in phases. Each phase $i$ lasts for $T = 2^i$ rounds, and outputs its \emph{best guess} $x^*_i\in X$, which is played throughout phase $i+1$. During phase $i$, the algorithm picks a $\delta$-hitting set%
\footnote{A subset $S\subset X$ is a \emph{$\delta$-hitting set} for $Y\subset X$ if
    $Y \subset \cup_{x\in S}\, B(x,\,\delta)$.}
for $X$ of size at most $N_\delta(X)$, for
    $\delta = T^{-1/(b+2)}$.
By the end of the phase, $x^*_i$ as defined as the point in $S$ with the highest sample average (breaking ties arbitrarily). This completes the description of the algorithm.

It is easy to see that the regret of \NaiveExp\ is naturally described in terms of the log-covering dimension (see \eqref{eq:LCD}). The proof is based the argument from \citet{Bobby-nips04}. We restate it here for the sake of completeness, and to explain how the new dimensionality notion is used.

\begin{theorem}\label{thm:ffproblem-naive}
Consider the \FFproblem\ on a metric space $(X,\mD)$. For each $b>\LCD(X)$, algorithm $\NaiveExp(b)$ achieves regret
    $R(t) = O(t^{1-1/(b+2)})$.
\end{theorem}

\begin{proof}
Let $N_\delta = N_\delta(X)$, and let $\mu$ be the expected payoff function. Consider a given phase $i$ of the algorithm. Let $T=2^i$ be the phase duration. Let $\delta = T^{-1/(b+2)}$, and let $S\subset X$ the $\delta$-hitting set chosen in this phase. Note that for any sufficiently large $T$ it is the case that
    $N_\delta < 2^{\delta^{-b}}$.
For each $x\in S$, let $\mu_T(x)$ be the sample average of the feedback from $x$ by the end of the phase. Then by Chernoff bounds,
\begin{align}\label{eq:ffproblem-naive}
 \Pr[ | \mu_T(x) - \mu(x)| <  r_T ] > 1-  (T N_\delta)^{-3},
    \quad\text{where}\quad
        r_T = \sqrt{ 8\,\log (T\, N_\delta)\, / T} < 2\delta.
\end{align}
Note that $\delta$ is chosen specifically to ensure that $r_T \leq O(\delta)$.

We can neglect the regret incurred when the event in~\refeq{eq:ffproblem-naive} does not hold for some $x\in S$. From now on, let us assume that the event in~\refeq{eq:ffproblem-naive} holds for all $x\in S$. Let $x^*$ be an optimal strategy, and $x^*_i = \argmax_{x\in S} \mu_T(x) $ be the ``best guess". Let $x\in S$ be a point that covers $x^*$. Then
$$ \mu(x^*_i)
    \geq \mu_T(x^*_i) - 2\delta
    \geq \mu_T(x) - 2\delta
    \geq \mu(x) - 4 \delta
    \geq \mu(x^*) - 5 \delta.
$$
Thus the total regret $R_{i+1}$ accumulated in phase $i+1$ is
$$ R_{i+1}\leq 2^{i+1}\, (\mu(x^*) - \mu(x^*_i)) \leq O(\delta T)
    = O(T^{1-1/(2+b)}).
$$
Thus the total regret summed over phases is as claimed.
\end{proof}

\subsection{Uniformly Lipschitz experts (proof of Theorem~\ref{thm:intro-LCD-uniform})}
\label{sec:FFproblem-uniform}

We now turn our attention to the \emph{\ULproblem}, a restricted version of the \FFproblem\ in which a problem instance $(X,\mD,\prob)$ satisfies a further property that each function $f\in \mathtt{support}(\mathbb{P})$ is itself a Lipschitz function on $(X,\mD)$. We show that for this version, \NaiveExp\ obtains a significantly better regret guarantee, via a more involved analysis. As we will see in the next section, for a wide class of metric spaces including \eps-uniform tree metrics there is a matching upper bound.

\OMIT{ 
The lower bound in Theorem~\ref{thm:experts-MaxMinLCD} holds for the \emph{\ULproblem},
Let us show the matching upper bound for the metric spaces such that
    $\MaxMinLCD(X) = \LCD(X)$.
In fact, this bound is achieved by algorithm $\NaiveExp$ via a more involved analysis.
} 

\begin{theorem}\label{thm:ulproblem-naive}
Consider the \ULproblem\ with full feedback. Fix a metric space $(X,\mD)$. For each $b>\LCD(X)$ such that $b\geq 2$, $\NaiveExp(b-2)$ achieves regret
    $R(t) = O(t^{1-1/b})$.
\end{theorem}

\begin{proof}
The preliminaries are similar to those in the proof of Theorem~\ref{thm:ffproblem-naive}. For simplicity, assume $b\geq 2$. Let $N_\delta = N_\delta(X)$, and let $\mu$ be the expected payoff function. Consider a given phase $i$ of the algorithm. Let $T=2^i$ be the phase duration. Let $\delta = T^{-1/b}$, and let $S$ be the $\delta$-hitting set chosen in this phase. (The specific choice of $\delta$ is the only difference between the algorithm here and the algorithm in Theorem~\ref{thm:ffproblem-naive}.) Note that $|S|\leq N_\delta$, and for any sufficiently large $T$ it is the case that
    $N_\delta < 2^{\delta^{-b}}$.

The rest of the analysis holds for any set $S$ such that $|S|\leq N_\delta$. (That is, it is not essential that $S$ is a $\delta$-hitting set for $X$.) For each $x\in S$, let $\nu(x)$ be the sample average of the feedback from $x$ by the end of the phase. Let $y^*_i = \argmax (\mu,S)$ be the optimal strategy in the chosen sample, and let $x^*_i = \argmax(\nu,S)$ be the algorithm's ``best guess". The crux is to show that
\begin{align}\label{eq:thm-ulproblem-naive-crux}
 \Pr[\, \mu(y^*_i) - \mu(x^*_i) \leq O(\delta \log T) \,]
     > 1-  T^{-3}.
\end{align}
Once~\refeq{eq:thm-ulproblem-naive-crux} is established, the remaining steps is exactly as the proof of Theorem~\ref{thm:ffproblem-naive}.

Proving~\refeq{eq:thm-ulproblem-naive-crux} requires a new technique. The obvious approach -- to use Chernoff Bounds for each $x\in S$ separately and then take a Union Bound -- does not work, essentially because one needs to take the Union Bound over too many points. Instead, we will use a more efficient version tail bound: for each $x,y\in X$, we will use Chernoff Bounds applied to the random variable $f(x)-f(y)$, where $f \sim \prob$ and  $(X,\mD,\prob)$ is the problem instance. For a more convenient notation, we define
$$\Delta(x,y)
    = \left[\, \mu(x)-\mu(y) \,\right]
    + \left[\, \nu(y)-\nu(x) \,\right],$$
Then for any $N\in \N$ we have
\begin{align}\label{eq:thm-ulproblem-naive-Chernoff}
 \Pr\left[\,  |\Delta(x,y)| \leq
        \mD(x,y)\,\sqrt{8\,\log (T\, N)/T}
 \right]
 > 1-  (T N)^{-3}.
\end{align}
The point is that the ``slack" in the Chernoff Bound is scaled by the factor of $\mD(x,y)$. This is because each $f\in \texttt{support}(\prob)$ is a Lipschitz function on $(X,\mD)$,

In order to take advantage of~\refeq{eq:thm-ulproblem-naive-Chernoff}, let us define the following structure that we call the \emph{covering tree} of the metric space $(X,\mD)$. This structure consists of a rooted tree $\mathcal{T}$ and non-empty subsets $X(u)\subset X$ for each internal node $u$. Let $V_\mathcal{T}$ be the set of all internal nodes. Let $\mathcal{T}_j$ be the set of all level-$j$ internal nodes (so that $\mathcal{T}_0$ is a singleton set containing the root). For each
    $u\in V_\mathcal{T}$,
let $\mathcal{C}(u)$ be the set of all children of $u$. For each node $u\in \mathcal{T}_j$ the structure satisfies the following two properties: (i) set $X(u)$ has diameter at most $2^{-j}$, (ii) the sets $X(v)$, $v\in \mathcal{C}(u)$ form a partition of $X(u)$. This completes the definition.

By definition of the covering number $N_\delta(\cdot)$ there exist a covering tree $\mathcal{T}$ in which each node $u\in \mathcal{T}_j$ has fan-out $N_{2^{-j}}( X(u))$. Fix one such covering tree. For each node $u\in V_\mathcal{T}$, define
\begin{align}
   \sigma(u) &= \argmax( \mu,\, \mathcal{X}(u) \cap S)
        \label{eq:thm-ulproblem-naive-sigma}\\
     \rho(u) &= \argmax( \nu,\, \mathcal{X}(u) \cap S), \nonumber
\end{align}
where the tie-breaking rule is the same as in the algorithm.

Let
    $n = \cel{\log\tfrac{1}{\delta}}$.
Let us say that phase $i$ is \emph{clean} if the following two properties hold:
\begin{OneLiners}
\item[(i)] for each node $u\in V_\mathcal{T}$  any two children $v,w\in \mathcal{C}(u)$ we have
    $|\, \Delta( \sigma(v),\, \sigma(w))\, | \leq 4\delta $.

\item[(ii)] for any $x,y\in S$ such that $\mD(x,y)\leq \delta$ we have
    $|\Delta(x,y)| \leq 4\delta $.
\end{OneLiners}

\begin{claim}
For any sufficiently large $i$,
phase $i$ is clean with probability at least $1-T^{-2}$.
\end{claim}
\begin{proof}
To prove (i), let $j$ be such that $u\in \mathcal{T}_j$. We consider each $j$ separately. Note that (i) is trivial for $j>n$. Now fix $j\leq n$ and apply the Chernoff-style bound~\refeq{eq:thm-ulproblem-naive-Chernoff} with
    $N = |\mathcal{T}_j|$ and $(x,y) = (\sigma(v), \sigma(w))$.
Since
	$|\mathcal{T}_l| \leq 2^{2^{lb}}\, |\mathcal{T}_{l-1}|$
for each sufficiently large $l$, it follows that
$\log |\mathcal{T}_j|
	\leq C + \textstyle{\sum_{l=1}^j}\; 2^{lb}
	\leq C+ \tfrac{4}{3}\, 2^{jb},
$
where $C$ is a constant that depends only on the metric space and $b$.
It is easy to check that for any sufficiently large phase $i$ (which, in turn, determines $T$, $\delta$ and $n$),  the ``slack" in ~\refeq{eq:thm-ulproblem-naive-Chernoff} is at most $4\delta$:
\begin{align*}
\mD(x,y)\,\sqrt{8\,\log (T\, N)/T}
	&\leq  3\, \mD(x,y)\,\sqrt{\log (N)/T}
	\leq 4\,2^{-j}\, \sqrt{2^{bj}/ 2^{bn}}
	= 4 \delta\, 2^{-(n-j)(b-2)/2}
	\leq 4\delta.
\end{align*}
Interestingly, the right-most inequality above is the only place in the proof where it is essential that $b\geq 2$.

To prove (ii), apply ~\refeq{eq:thm-ulproblem-naive-Chernoff} with $N = |S|$ similarly. Claim proved.
\end{proof}

From now on we will consider clean phase. (We can ignore regret incurred in the event that the phase is not clean.) We focus on the quantity
    $ \Delta^*(u) = \Delta( \sigma(u),\, \rho(u))$.
Note that by definition $ \Delta^*(u)\geq 0$. The central argument of this proof is the following upper bound on $\Delta^*(u)$.

\begin{claim}\label{cl:thm-ulproblem-naive}
In a clean phase,
    $\Delta^*(u) \leq O(\delta)(n-j)$
for each $j\leq n$ and each $u\in \mathcal{T}_j$.
\end{claim}
\begin{proof}
Use induction on $j$. The base case $j=n$ follows by part (ii) of the definition of the clean phase, since for $u\in \mathcal{T}_n$ both $\sigma(u)$ and $\rho(u)$ lie in $X(u)$, the set of diameter at most $\delta$. For the induction step, assume the claim holds for each $v\in \mathcal{T}_{j+1}$, and let us prove it for some fixed  $u\in \mathcal{T}_j$.

Pick children $u,v\in \mathcal{C}(u)$ such that
    $\sigma(u) \in X(v)$ and $\rho(u)\in X(w)$.
Since the tie-breaking rules in~\refeq{eq:thm-ulproblem-naive-sigma} is fixed for all nodes in the covering tree, it follows that
    $\sigma(u) = \sigma(v)$ and $\rho(u) = \rho(w)$.
Then
\begin{align*}
\Delta^*(w) + \Delta( \sigma(v),\, \sigma(w) )
    & = \Delta( \sigma(w),\, \rho(u) ) + \Delta( \sigma(u),\, \sigma(w) ) \\
    & =  \mu(\sigma(w)) - \mu(\rho(u)) + \nu(\rho(u)) - \rho(\sigma(w)) \;+\\
    & \quad \;
        \mu(\sigma(u)) - \mu(\sigma(w)) + \nu(\sigma(w)) - \nu(\sigma(u)) \\
    & = \Delta^*(u).
\end{align*}
Claim follows since
    $\Delta^*(w)\leq O(\delta)(n-j-1)$
by induction, and $\Delta( \sigma(v),\, \sigma(w) )\leq 4\delta$ by part (i) in the definition of the clean phase.
\end{proof}

To complete the proof of~\refeq{eq:thm-ulproblem-naive-crux}, let $u_0$ be the root of the covering tree. Then
    $y^*_i = \sigma(u_0)$ and $x^*_i = \rho(u_0)$.
Therefore by Claim~\ref{cl:thm-ulproblem-naive} (applied for $\mathcal{T}_0 = \{u_0\}$) we have
$$ O(\delta n) \geq \Delta^*(u_0) = \Delta^*(y^*_i, \, x^*_i)
    \geq \mu(y^*_i) - \mu(x^*_i). \qedhere
$$
\end{proof}


\subsection{Regret characterization (proof of Theorem~\ref{thm:intro-MaxMinLCD})}
\label{sec:FFproblem-characterization}

As it turns out, the log-covering dimension is not the right notion to characterize optimal regret for arbitrary metric spaces. We need a more refined version: the \emph{max-min-log-covering dimension}, defined in \eqref{eq:MaxMinLCD},
 similar to the max-min-covering dimension.

\begin{theorem}\label{thm:experts-MaxMinLCD-body}
Fix a metric space $(X,\mD)$ and let $b=\MaxMinLCD(X)$. The \FFproblem\ on $(X,\mD)$ is $(t^\gamma)$-tractable for any $\gamma> \tfrac{b+1}{b+2}$, and not $(t^\gamma)$-tractable for any $\gamma < \tfrac{b-1}{b}$.
\end{theorem}

For the lower bound, we use a suitably ``thick'' version of the ball-tree from Section~\ref{sec:lower-bound} in conjunction with the $(\eps,\delta,k)$-ensemble
idea from Section~\ref{sec:lower-bound}, see Section~\ref{sec:MaxMinLCD-LB}. For the algorithmic result, we combine the ``naive" experts algorithm (\NaiveExp) with (an extension of) the \emph{transfinite fat decomposition} technique from Section~\ref{sec:pmo}, see Section~\ref{sec:MaxMinLCD-UB}.

The lower bound in Theorem~\ref{thm:experts-MaxMinLCD-body} holds for the \ULproblem. It follows that the upper bound in Theorem~\ref{thm:ulproblem-naive} is optimal for metric spaces such that $\MaxMinLCD(X) = \LCD(X)$, e.g. for \eps-uniform tree metrics. In fact, we can plug the improved analysis of \NaiveExp\ from Theorem~\ref{thm:ulproblem-naive} into the algorithmic technique from Theorem~\ref{thm:experts-MaxMinLCD-body} and obtain a matching upper bound in terms of the \MaxMinLCD. Thus (in conjunction with Theorem~\ref{thm:main-experts}) we have a complete characterization for regret:

\OMIT{ 
The extension to the corresponding result for $\MaxMinLCD$ (the upper bound in Theorem~\ref{thm:experts-MaxMinLCD-unif}) proceeds exactly as in the proof of Theorem~\ref{thm:experts-MaxMinLCD}, except we use a more efficient analysis of \NaiveExp. We omit the details from this version.
} 

\begin{theorem}\label{thm:experts-MaxMinLCD-unif-body}
Consider the \ULproblem\ with full feedback. Fix a metric space $(X,\mD)$ with uncountably many points, and let $b=\MaxMinLCD(X)$. The problem on $(X,\mD)$ is $(t^\gamma)$-tractable for any
    $\gamma> \max(\tfrac{b-1}{b},\, \tfrac12)$,
and not $(t^\gamma)$-tractable for any
    $\gamma < \max(\tfrac{b-1}{b},\, \tfrac12)$.
\end{theorem}

The proof of the upper bound in Theorem~\ref{thm:experts-MaxMinLCD-unif-body} proceeds exactly that in Theorem~\ref{thm:experts-MaxMinLCD-body}, except that we use a more efficient analysis of \NaiveExp.

\subsubsection{The $\MaxMinLCD$ lower bound: proof for Theorem~\ref{thm:experts-MaxMinLCD-unif-body}}
\label{sec:MaxMinLCD-LB}

If $\MaxMinLCD(X) = d,$ and $\gamma < \tfrac{d-1}{d},$
let us first fix constants $b$ and $c$ such that $b < c < d$ and
$\gamma < \tfrac{b-1}{b}$.  Let $Y \subseteq X$
be a subspace such that $c \leq \inf \{\LCD(Z) : \mbox{open,
nonempty } Z \subseteq Y \}.$
We will repeatedly use the following packing lemma that
relies on the fact that $b < \LCD(U)$ for all
nonempty subsets $U \subseteq Y$.

\begin{lemma} \label{lem:packing-LCD}
For any nonempty
open $U \subseteq Y$ there exists $r_0>0$
such that for all $r \in (0,r_0)$,
$U$ contains more than $2^{r^{-b}}$
disjoint balls of radius $r$.
\end{lemma}
\begin{proof}
Let $r_0$ be a positive number such that
for all positive $r < r_0$, every
covering of $U$ requires more than $2^{r^{-b}}$ balls of
radius $2r$.  Such an $r_0$ exists, because
$\LCD(U) > b$.
Now for any positive $r < r_0$
let $\mathcal{P} = \{B_1,B_2,\ldots,B_M\}$ be any
maximal collection of disjoint $r$-balls.  For
every $y \in Y$ there must exist some ball $B_i
\; (1 \leq i \leq M)$ whose center is within
distance $2r$ of $y$, as otherwise $B(y,r)$
would be disjoint from every element of $\mathcal{P}$
contradicting the maximality of that collection.
If we enlarge each ball $B_i$ to a ball $B_i^+$
of radius $2r$, then every $y \in Y$ is contained
in one of the balls $\{B_i^+ \,|\, 1 \leq i \leq M\}$,
i.e. they form a covering of $Y$.  Hence
$M \geq 2^{r^{-b}}$ as desired.
\end{proof}

Using the packing lemma we recursively construct a ball-tree on metric space $(Y,\mD)$ with very high node degrees. Specifically, let us say that a ball-tree has \emph{log-strength} $b$ if each tree node with children of radius $r$ has at least $2^{r^{-b}}$ children. For convenience, all tree nodes of the same depth will have the same radius $r_i$. Then each node at depth $i-1$ has at least
    $n_i = \lceil 2^{r_i^{-b}} \rceil$
children.

\begin{claim}
There exists a ball-tree $T$ on $(Y,\mD)$ with log-strength $b$, in which all tree nodes of the same depth $i$ have the same radius $r_i$.
\end{claim}

\begin{proof}
The root of the ball tree is centered at any point in $Y$ and has radius $r_0=\tfrac14$. For each successive $i\geq 1$, let $r_i \in(0, r_{i-1}/4)$ be a positive number small enough that for every depth $i-1$ tree node $w=(x,r_{i-1})$, the sub-ball $B(x,r_{i-1}/2)$ contains $n_i = \lceil 2^{r_i^{-b}} \rceil$ disjoint balls
of radius $r_i$. (Denote by $\mB_w$ the collection of the corresponding disjoint  \myballs.) Such $r_i$ exists by Lemma~\ref{lem:packing-LCD}.
The set of children of $w$ is defined to be $\mB_w$.
\end{proof}

We re-use Construction~\ref{con:LB-payoffs} for metric space $(Y,\mD)$ and ball-tree $T$, with $\delta_i\equiv \tfrac13$.
Thus, we construct a problem instance $\prob_\lambda$ for each lineage over $\lambda$, and a distribution $\mP_T$ over problem instances $\prob_\lambda$. Recall that a problem instance is a distribution over (deterministic) payoff functions $\pi:X\to [0,1]$, which are Lipschitz by Lemma~\ref{lm:LB-Lipschitz}.

Fix a complete lineage $\lambda$, and let
    $\mathbf{w}(\lambda)=(w_0,w_1,\, \ldots)$
be the associated end of the ball-tree. For each $i\geq 1$, let $B_i$ be the ball in $(Y,\mD)$ corresponding to tree node $w_i$. Let
    $\mu = \E_{\pi\sim \prob_\lambda}[\pi]$
be the expected payoff function corresponding to $\prob_{\mathcal{Q}}$. Then
then $\mu$ achieves its maximum value
$\tfrac12 + \tfrac{1}{18} \sum_{i=0}^{\infty} r_i$
at the unique point $x^* \in \cap_{i=0}^{\infty} B_i$.
At any point $x \not\in B_j$, we have
\begin{align*}
\mu(x^*) - \mu(x) \;
    \geq \; \textstyle{ \left(\tfrac{1}{18}\, \sum_{i=j}^{\infty} r_i\right)}
  -  \textstyle{ \left( \tfrac{1}{18}\, \sum_{i=j+1}^{\infty} r_i \right) \; = \; \tfrac{1}{18}\, r_j.}
\end{align*}

We now finish the lower bound proof as in the proof of Lemma~\ref{lm:ball-tree-LB}. Fix depth $i-1$ node $w$ in the ball-tree, and let
    $w^1,w^2 \LDOTS w^{n_i}$
be the children of $w$ in the ball-tree. Let $\lambda(w)$ be the unique child of $w$ contained in the lineage $\lambda$. Consider the
sets
   $\lambda_0 = \lambda \setminus \lambda(w)$
and
   $\lambda_j = \lambda_0 \cup \{w^j\}$
for $j=1,2,\ldots,n_i$.  By Corollary~\ref{cor:ensemble}, the distributions
$\left(
\prob_{\lambda_0},\prob_{\lambda_1},\ldots,\prob_{\lambda_{n_i}}
\right)$
constitute an $(\eps,\delta,k)$-ensemble
for $\eps=r_i/18$, $\delta=\tfrac13,$
and $k=n_i$.
Consequently, for $t_i = r_i^{-b}$, the inequality
$t_i < \ln(17k)/2\delta^2$ holds, and we
obtain a lower bound of
  $$ R_{(\A, \, \prob_{\lambda_j})}(t_i) > \eps\, t_i / 2 = \Omega( r_i^{1-b})
 = \Omega(t_i^{(b-1)/b}) $$
for at least half of the distributions $\prob_{\lambda_j}$
in the ensemble.  Recalling that $\gamma < \tfrac{b-1}{b}$,
we see that the problem is not $t^{\gamma}$-tractable.

\subsubsection{The $\MaxMinLCD$ upper bound: proofs for
Theorem~\ref{thm:experts-MaxMinLCD-body} and Theorem~\ref{thm:experts-MaxMinLCD-unif-body}
}
\label{sec:MaxMinLCD-UB}

\newcommand{\NaiveSample}{\ensuremath{\text{{\sc NaiveSample}}}}

First, let us incorporate the analysis of $\NaiveExp(b)$ via the following lemma.

\begin{lemma}\label{lm:ULproblem-recap}
Consider an instance $(X,\mD,\prob)$ of the \FFproblem, and let $x^*\in X$ be an optimal point. Fix subset $U\subset X$ which contains $x^*$, and let  $b>\LCD(U)$. Then for any sufficiently large $T$ and $\delta = T^{-1/(b+2)}$ the following holds:
\begin{itemize}
\item[(a)] Let $S$ be a $\delta$-hitting set for $U$ of cardinality $|S|\leq N_\delta(U)$. Consider the feedback of all points in $S$ over $T$ rounds; let $x$ be the point in $S$ with the largest sample average (break ties arbitrarily). Then
	$$ \Pr[\mu(x^*) - \mu(x) <O(\delta\log T)]> 1-T^{-2}.$$

\item[(b)] For a \ULproblem\ and $b\geq 2$, property (a) holds for $\delta = T^{-1/b}$.

\end{itemize}
\end{lemma}

\xhdr{Transfinite LCD decomposition.}
We redefine the \emph{transfinite fat decomposition} from Section~\ref{sec:pmo} with respect to the log-covering dimension rather than the covering dimension.

\begin{definition}\label{def:fatness-transfinite-experts}
Fix a metric space $(X,\mD)$. Let $\beta$ denote an arbitrary ordinal.
A \emph{transfinite LCD decomposition} of depth $\beta$ and dimension $b$ is a  transfinite sequence
	$\{S_\lambda\}_{0 \leq \lambda \leq \beta}$
of closed subsets of $X$ such that:
\begin{description}
\item[(a)] $S_0 = X$, $S_\beta = \emptyset$, and
	$S_\nu \supseteq S_\lambda$ whenever $\nu < \lambda$.
\item[(b)] if $V\subset X$ is closed, then the set
    $\{\text{ordinals } \nu \leq \beta$:\, $V \mbox{ intersects } S_\nu \}$
has a maximum element.
\item[(c)] for any ordinal $\lambda \leq \beta$ and any open set
$U\subset X$ containing $S_{\lambda+1}$ we have
	$\LCD(S_\lambda \setminus U) \leq b$.
\end{description}
\end{definition}

The existence of suitable decompositions and the connection to $\MaxMinLCD$ is derived exactly as in 
Proposition~\ref{prop:fatness-dim}

\begin{lemma} \label{prop:fatness-dim-experts}
For every compact metric space $(X,\mD)$, $\MaxMinLCD(X)$ is equal to the infimum of all $b$ such that $X$ has a transfinite LCD decomposition of dimension $b$.
\end{lemma}



In what follows, let us fix metric space $(X,\mD)$ and $b>\MaxMinLCD(X)$, and let
	$\{S_\lambda\}_{0 \leq \lambda \leq \beta}$
be a transfinite LCD decomposition of depth $\beta$ and dimension $b$. For each $x\in X$, let the \emph{depth} of $x$ be the maximal ordinal $\lambda$ such that $x\in S_\lambda$. (Such an ordinal exists by Definition~\ref{def:fatness-transfinite-experts}(b).)

\xhdr{Access to the metric space.}
The algorithm requires two oracles: the \emph{depth oracle} $\DpthOracle(\cdot)$ and the \emph{covering oracle} $\DCovOracle(\cdot)$. Both oracles input a finite collection $\F$  of open balls $B_0, B_1, \ldots, B_n$, given via the centers and the radii, and return a point in $X$. Let $B$ be the union of these balls, and let $\overline{B}$ be the closure of $B$. A call to oracle $\DpthOracle(\F)$ returns an arbitrary point $x\in \overline{B} \cap S_\lambda$, where $\lambda$ is the maximum ordinal such that $S_{\lambda}$ intersects $\overline{B}$. (Such an ordinal exists by Definition~\ref{def:fatness-transfinite-experts}(b).) Given a point $y^*\in X$ of depth $\lambda$, a call to oracle $\DCovOracle(y^*,\F)$ either reports that $B$ covers $S_{\lambda}$, or it returns an arbitrary point $x \in S_{\lambda} \setminus B$. A call to $\DCovOracle(\emptyset, \F)$ is equivalent to the call $\DCovOracle(y^*,\F)$ for some $y^*\in S_0$.

The covering oracle will be used to construct $\delta$-nets as follows. First, using successive calls to $\DCovOracle(\emptyset, \F)$ one can construct a $\delta$-net for $X$. Second, given a point $y^*\in X$ of depth $\lambda$ and a collection of open balls whose union is $B$, using successive calls to $\DCovOracle(y^*,\,\cdot)$ one can construct a $\delta$-net for $S_\lambda \setminus B$. The second usage is geared towards the scenario when $S_{\lambda+1} \subseteq B$ and  for some optimal strategy $x^*$ we have
	$x^*\in S_\lambda \setminus B$.
Then by Definition~\ref{def:fatness-transfinite-experts}(c) we have
	$\LCD(S_\lambda\setminus B)<b$,
and one can apply Lemma~\ref{lm:ULproblem-recap}.

\xhdr{The algorithm.}
Our algorithm proceeds in phases $i=1,2,3,\ldots$ of $2^i$ rounds each. Each phase $i$ outputs two strategies:  $x^*_i, y^*_i\in X$ that we call the \emph{best guess} and the \emph{depth estimate}. Throughout phase $i$, the algorithm plays the best guess $x^*_{i-1}$ from the previous phase. The depth estimate $y^*_{i-1}$ is used ``as if" its depth is equal to the depth of some optimal strategy. (We show that for a large enough $i$ this is indeed the case with a very high probability.)

In the end of the phase, an algorithm selects a finite set $A_i\subset X$ of \emph{active points}, as described below. Once this set is chosen, $x^*_i$  is defined simply as a point in $A_i$ with the largest sample average of the feedback (breaking ties arbitrarily). It remains to define $y^*_i$ and $A_i$ itself.

Let $T=2^i$ be the phase duration. Using the covering oracle, the algorithm constructs (roughly) the finest $r$-net containing at most $2^{\sqrt{T}}$ points. Specifically, the algorithm constructs  $2^{-j}$-nets $\mathcal{N}_j$, for $j = 0,1,2,\ldots$, until it finds the largest $j$ such that
	$\mathcal{N}_j$
contains at most $2^{\sqrt{T}}$  points. Let
	$r = 2^{-j}$ and $\mathcal{N} = \mathcal{N}_j$.

For each $x\in X$, let $\mu_T(x)$ be the sample average of the feedback during this phase. Let
\begin{align*}
 \Delta_T(x) &= \mu^*_T - \mu_T(x),
        \text{~~~where~~~}
        \mu^*_T = \max(\mu_T, \mathcal{N})
\end{align*}

\noindent Define the depth estimate $y^*_i$ to be the output of the oracle call $\DpthOracle(\F)$, where
$$ \F = \{ B(x,r):\; x\in \mathcal{N} \text{~~and~~} \Delta_T(x)< r \}.
$$

Finally, let us specify $A_i$. Let $B$ be the union of balls
\begin{align}\label{eq:ffproblem-PMO-B}
\{ B(x,r):\; x\in \mathcal{N} \text{~~and~~} \Delta_T(x)> 2(r_T + r)\, \},
\end{align}
where
            $r_T = \sqrt{8\log(T\,|\mathcal{N}|)/T}$
is chosen so that by Chernoff Bounds we have
\begin{align}\label{eq:ffproblem-PMO-Chernoff}
 \Pr[ | \mu_T(x) - \mu(x)| <  r_T ] > 1-  (T \,|\mathcal{N}|)^{-3}
    \quad\text{for each $x\in \mathcal{N}$}.
\end{align}

\noindent Let  $\delta = T^{-1/b}$ for the \ULproblem, and
	$\delta = T^{-1/(b+2)}$
otherwise. Let
    $Q_T =  2^{\delta^{-b}} $
be the \emph{quota} on the number of active points. Given a point $y^*_{i-1}$ whose depth is (say) $\lambda$, algorithm uses the covering oracle to construct a $\delta$-net $\mathcal{N'}$ for $S_\lambda \setminus B$. Define $A_i$ as $\mathcal{N}'$ or an arbitrary $Q_T$-point subset thereof, whichever is smaller.\footnote{The interesting case here is $|\mathcal{N}'| \leq Q_T$.  If $\mathcal{N}'$ contains too many points, the choice of $A_i$ is not essential for the analysis.}

\xhdr{Sketch of the analysis.}
The proof roughly follows that of
Theorem~\ref{thm:pmo}.
Call a phase \emph{clean} if the event in~\refeq{eq:ffproblem-PMO-Chernoff} holds for all $x\in \mathcal{N}_i$ and the appropriate version of this event holds for all $x\in A_i$. (The regret from phases which are not clean is negligible). On a very high level, the proof consists of two steps. First we show that for a sufficiently large $i$, if phase $i$ is clean then the depth estimate $y^*_i$ is correct, in the sense that it is indeed equal to the depth of some optimal strategy. The argument is similar to the one in Lemma~\ref{lm:tractability}. Second, we show that for a sufficiently large $i$, if the depth estimate $y^*_{i-1}$ is ``correct" (i.e. its depth is equal to that of some optimal strategy), and phase $i$ is clean, then the ``best guess" $x^*_i$ is good, namely $\mu(x^*_i)$ is within $O(\delta log T)$ of the optimum. The reason is that, letting $\lambda$ be the depth of $y^*_{i-1}$, one can show that for a sufficiently large $T$ the set $B$ (defined in~\refeq{eq:ffproblem-PMO-B}) contains	$S_{\lambda+1}$ and does not contain some optimal strategy. By definition of the transfinite LCD decomposition we have
    $\LCD(S_\lambda \setminus U) < b$,
so in our construction the quota $Q_T$ on the number of active points permits $A_i$ to be a $\delta$-cover of $S_\lambda \setminus U$. Now we can use Lemma~\ref{lm:ULproblem-recap} to guarantee the ``quality" of $x^*_i$. The final regret computation is similar to the one in the proof of Theorem~\ref{thm:ffproblem-naive}.

\section{Conclusions}
\label{sec:conclusions}

\cite{LipschitzMAB-stoc08} (i.e., Sections~\ref{sec:adaptive-exploration} and Section~\ref{sec:pmo} of this paper) introduced the \problem and motivated a host of open questions. Many of these questions have been addressed in the follow-up work, including~\cite{DichotomyMAB-soda10} (i.e., the rest of this paper), and the work described in Section~\ref{sec:related-followup}. Below we describe the current state of the open questions.

First, the adaptive refinement technique from Section~\ref{sec:adaptive-exploration} can potentially be used in other settings in explore-exploit learning where one has side information on similarity between arms. Specific potential applications include adversarial MAB, Gaussian Process Bandits, and dynamic pricing. Also, stronger analysis of this technique appears possible in the context of ranked bandits (see \cite{ZoomingRBA-icml10} for details).

Second, it is desirable to consider MAB with more general structure on payoff functions. A particularly attractive target would be structures that subsume Lipschitz MAB and Linear MAB.

\OMIT{Third, the topological characterization of the max-min-covering dimension (existence of a ball-tree and a transfinite fat decomposition, see Section~\ref{sec:pmo}), and a similar result for max-min-log-covering-dimension in Section~\ref{sec:FFproblem} can potentially be extended to an abstract notion of dimensionality, which may be of independent mathematical interest.}

Third, a recurring theme in algorithm design is structural results that assert that a problem instance either has simple structure, or it contains a specific type of complex substructure that empowers the lower bound analysis. Our work contributes another example of this theme, in the form of dichotomy results in point-set topology (e.g. existence of a transfinite fat decomposition versus existence of a ball tree). It would potentially be interesting to find other applications of this technique.


\begin{small}
\bibliographystyle{plainnat}
\bibliography{bib-abbrv,bib-bandits,bib-random,bib-extras,bib-math,bib-embedding,bib-slivkins,bib-AGT,bib-nodeLabeling,bib-networking,bib-ML}
\end{small}

\appendix

\section{KL-divergence techniques}
\label{sec:KL-divergence}

\newcommand{\diff}{\ensuremath{\operatorname{d}\,}}

\newcommand{\KL}{\mathtt{KL}}

\newcommand{\given}{\,|\,}

All lower bounds in this paper heavily use the notion of Kullback-Leibler divergence (\emph{KL-divergence}). Our usage of the KL-divergence techniques is encapsulated in several statements in the body of the paper (Theorem~\ref{thm:LB-technique-MAB}, Theorem~\ref{thm:LB-technique}, and Claim~\ref{cl:logT-KLdiv}), whose proofs are fleshed out in this appendix and may be of independent interest. To make this appendix more  self-contained, we restate the relevant definitions and theorem statements from the body of the paper, and provide sufficient background.


\subsection{Background}

\begin{definition}  \label{def:kldiv}
Let $\Omega$ be a finite set with
two probability measures $p,q$.  Their \emph{KL-divergence} is the sum
\[
\KL(p;q) = \sum_{x \in \Omega} p(x) \ln \left(
\frac{p(x)}{q(x)} \right),
\]
with the convention that $p(x) \ln(p(x)/q(x))$
is interpreted to be $0$ when $p(x)=0$ and $+\infty$
when $p(x)>0$ and $q(x)=0$.  If $Y$ is a random
variable defined on $\Omega$ and taking values in
some set $\Gamma$, the \emph{conditional
KL-divergence} of $p$ and $q$
given $Y$ is the sum
\[
\KL(p;q \given Y) = \sum_{x \in \Omega} p(x)
\ln \left( \frac{p(x \given Y = Y(x))}{q(x \given Y = Y(x))} \right),
\]
where terms containing $\log(0)$ or $\log(\infty)$ are
handled according to the same convention as above.
\end{definition}

The definition can be applied to an infinite sample
space $\Omega$ provided that $q$ is absolutely
continuous with respect to $p$.  For details,
see \cite{Bobby-thesis}, Chapter 2.7.
The following lemma summarizes some standard facts about
KL-divergence; for proofs, see \citep{CoverThomas,Bobby-thesis}.
\begin{lemma} \label{lem:kl}
Let  $p,q$ be two probability measures on
a measure space $(\Omega,\mathcal{F})$ and
let $Y$ be a random variable defined on
$\Omega$ and taking values in some finite set
$\Gamma$.  Define a
pair of probability measures $p_Y,q_Y$ on $\Gamma$ by
specifying that $p_Y(y) = p(Y=y), q_Y(y) = q(Y=y)$ for
each $y \in \Gamma$.  Then
\[
\KL(p;q) = \KL(p;q \given Y) + \KL(p_Y;q_Y),
\]
and $\KL(p;q \given Y)$ is non-negative.
\end{lemma}
An easy corollary is the following lemma
which expresses the KL-divergence of two
distributions on sequences as a sum of
conditional KL-divergences.
\begin{lemma} \label{lem:kl-chainrule}
Let $\Omega$ be a sample space, and suppose $p,q$
are two probability measures on $\Omega^n$,
the set of $n$-tuples of elements of $\Omega$.
For a sample point $\vec{\omega} \in \Omega^n$,
let $\omega^i$ denote its  first $i$ components.
If $p^i,q^i$ denote the probability
measures induced on $\Omega^i$ by $p$
(resp. $q$) then
\[
\KL(p;q) = \textstyle{\sum_{i=1}^n}\, \KL(p^i;q^i \given \omega^{i-1}).
\]
\end{lemma}
\begin{proof}  For $m=1,2,\ldots,n$, the formula
$\KL(p^m;q^m) = \sum_{i=1}^m \KL(p^i;q^i \given \omega^{i-1})$
follows by induction on $m$, using Lemma~\ref{lem:kl}.
\end{proof}

The following three lemmas will also be useful
in our lower bound argument. They may have appeared in the literature, but we cannot provide specific citations. We provide proofs for the sake of completeness. Here and henceforth
we will use the following notational convention:
for real numbers $a,b \in [0,1]$,
$\KL(a;b)$ denotes the KL-divergence
$\KL(p;q)$ where $p,q$ are probability
measures on $\{0,1\}$ such that $p(\{1\})=a,
\, q(\{1\}) = b.$  In other words,
\[
\KL(a;b) = a \ln \left( \tfrac{a}{b} \right) +
(1-a) \ln \left( \tfrac{1-a}{1-b} \right).
\]
\begin{lemma} \label{lem:kl-bernoulli}
For any $0 < \eps < y \leq 1$,
$\KL(y-\eps;y) < \eps^2/y(1-y).$
\end{lemma}
\begin{proof}
A calculation using the inequality $\ln(1+x)<x$
(valid for $x > 0$) yields
\begin{align*}
\KL(y-\eps;y) &= (y-\eps) \ln \left( \tfrac{y-\eps}{y} \right)
+ (1-y+\eps) \ln \left( \tfrac{1-y+\eps}{1-y} \right) \\
&< (y-\eps) \left( \tfrac{y-\eps}{y} - 1 \right)
+ (1-y+\eps) \left( \tfrac{1-y+\eps}{1-y} - 1 \right) \\
& = \tfrac{-\eps(y-\eps)}{y} + \tfrac{\eps(1-y+\eps)}{1-y}
 = \tfrac{\eps^2}{y(1-y)}.\qedhere
\end{align*}
\end{proof}

\begin{lemma} \label{lem:kl-distinguishing}
Let $\Omega$ be a sample space with
two probability measures
$p,q$ whose KL-divergence is $\kappa.$  For
any event $\mathcal{E}$, the probabilities
$p(\mathcal{E}), \, q(\mathcal{E})$ satisfy
\[
q(\mathcal{E}) \geq
p(\mathcal{E}) \exp \left(
- \tfrac{\kappa + 1/e}{p(\mathcal{E})} \right).
\]
\end{lemma}
A consequence of the lemma, stated in less quantitative
terms, is the following: if $\kappa=\KL(p;q)$ is bounded above
and $p(\mathcal{E})$ is bounded away from zero then
$q(\mathcal{E})$ is bounded away from zero.
\begin{proof}
Let $a = p(\mathcal{E}), \, b = q(\mathcal{E}),
c = (1-a)/(1-b)$.
Applying Lemma~\ref{lem:kl} with $Y$ as the indicator
random variable of $\mathcal{E}$ we obtain
\[
\kappa = \KL(p;q) \geq \KL(p_Y;q_Y) =
a \ln \left( \tfrac{a}{b} \right) +
(1-a) \ln \left( \tfrac{1-a}{1-b} \right) =
a \ln \left( \tfrac{a}{b} \right) + (1-b)\, c \ln(c).
\]
Now using the inequality $c \ln(c) \geq -1/e$,
(valid for all $c \geq 0$) we obtain
\[
\kappa \geq a \ln(a/b) - (1-b)/e \geq a \ln(a/b) - 1/e.
\]
The lemma follows by rearranging terms.
\end{proof}
\begin{lemma} \label{lem:reverse-pinsker}
Let $p,q$ be two probability measures, and suppose that
for some $\delta \in (0,\tfrac{1}{2}]$ they satisfy
\[
\forall \mbox{\rm \ events } \mathcal{E}, \quad
1-\delta < \tfrac{q(\mathcal{E})}{p(\mathcal{E})}
< 1+\delta
\]
Then $\KL(p;q) < \delta^2.$
\end{lemma}
\begin{proof}
We will prove the lemma assuming the sample space
is finite.  The result for general measure spaces
follows by taking a supremum.

For every $x$ in the sample space $\Omega$, let
$r(x) = \frac{q(x)}{p(x)}-1$ and note that $|r(x)| < \delta$
for all $x$.  Now we make use of the inequality
$\ln(1+x) \leq x-x^2$, valid for $x \geq -\tfrac{1}{2}.$
\begin{align*}
\KL(p;q) &= \textstyle{\sum_{x}}\, p(x) \ln \left( \tfrac{p(x)}{q(x)} \right)  \quad\quad\;
 = \textstyle{\sum_{x}}\, p(x) \ln \left( \tfrac{1}{1 + r(x)} \right) \\
&= - \textstyle{\sum_{x}}\, p(x) \ln(1+r(x))
\;\leq - \textstyle{\sum_{x}}\, p(x) [r(x) - (r(x))^2 ] \\
& < - \left( \textstyle{\sum_{x}}\, p(x) r(x) \right) +
    \delta^2 \left( \textstyle{\sum_{x}}\, p(x) \right)
\\
& = - \left( \textstyle{\sum_{x}}\, q(x) - p(x) \right) + \delta^2
 = \delta^2. \qedhere
\end{align*}
\end{proof}

\subsection{Bandit lower bound via $(\eps,k)$-ensembles}

We consider an MAB problem with i.i.d. payoffs where the algorithm is given a set of arms $X$ and a collection $\F$ of feasible payoff functions $X\to [0,1]$. We call it the \emph{feasible MAB problem} on $(X,\F)$. We will consider 0-1 payoffs; then for a problem instance with payoff function $f\in\F$, the reward from each action $x\in X$ is $1$ with probability $f(x)$, and $0$ otherwise.

\begin{definition*}[Definition~\ref{def:eps-k-ensemble}, restated]
Consider the feasible MAB problem on $(X,\mF)$.
An \emph{$(\eps,k)$-ensemble} is
 a collection of subsets $\F_1 \LDOTS \F_k \subset \F$ such that there exist mutually disjoint subsets
        $S_1 \LDOTS S_k \subset X$ and a
function $\mu_0 : X \to [\tfrac13, \tfrac23]$
such that
for each $i=1 \ldots k$ and each function $\mu_i\in \F_i$ the following holds:
(i) $\mu_i \equiv \mu_0$ on each $S_\ell$, $\ell\neq i$, and
(ii) $\sup(\mu_i, S_i) - \sup(\mu_0,X) \geq \eps$, and
(iii) $0\leq \mu_i-\mu_0\leq 2\eps$ on $S_i$.
\end{definition*}

\begin{theorem*}[Theorem~\ref{thm:LB-technique-MAB}, restated]
Consider the feasible MAB problem with 0-1 payoffs. Let  $\F_1, \ldots, \F_k$ be an $(\eps,k)$-ensemble,
where $k\geq 2$ and $\eps\in(0,\,\tfrac{1}{24})$. Then for any
    $t \leq \tfrac{1}{128}\, k\,\eps^{-2}$
and any bandit algorithm there exist at least $k/2$ distinct $i$'s such that the regret of this algorithm on any payoff function from $\F_i$ is at least $\tfrac{1}{60}\,\eps t$.
\end{theorem*}


\begin{proof}
Let us specify the notation. Let $\Omega = X \times \{0,1\}$. Since we assume 0-1 payoffs, the $t$-step history of play
of a bandit algorithm \A\ can be expressed by
an element of $\Omega^t$ indicating the sequence
of arms selected and payoffs received.
Thus, an algorithm \A\ and a payoff function $\mu$
together determine a probability distribution on
$\Omega^t$ for every natural number $t$.  Fix
any (possibly randomized) algorithm \A\ and
consider the distribution $p$ determined by \A\ when
the payoff function is $\mu_0$.  Recall the
mutually disjoint sets $S_1,S_2,\ldots,S_k$ in the
definition of an $(\eps,k)$-ensemble.
For $1 \leq i \leq k$ and $1 \leq u \leq t$,
let $Y_{i,u}$ be the indicator random variable
of the event $x_u \in S_i$,
where $x_u$ denotes the arm selected by \A\ at
time $u$.  Let $Z_i = \sum_{u=1}^t Y_{i,u}$.

Since
    $\sum_{i=1}^k\; \E_p \left[ Z_i \right] \leq t$,
there must be at least $k/2$ indices $i$ such that
    $\E_p[Z_i] \leq t/k \leq 1/128\, \eps^2$.
Fix one such $i$, and an arbitrary $\mu_i \in \F_i$. In what follows, we will show
that $R_{(\A,\,\mu_i)}(t) \geq \eps t/60$.

Let $(x_u,y_u) \in X \times \{0,1\} = \Omega$
denote the arm selected and the payoff
received at time $u$, and
let $q$ denote the distribution on
$\Omega^t$ determined by \A\ and $\mu_i$.  We have

\begin{align*}
\KL(p^u;q^u \given \omega^{u-1}) &=
\sum_{\omega^u \in \Omega^u} p^u(\omega^u)
\;\ln \left( \frac{p^u(\omega^u \given \omega^{u-1})}{q^u(\omega^u \given \omega^{u-1})} \right) \\
&=
\sum_{\omega^u\in \Omega^u} p^u(\omega^u) \;\ln \left(
\frac{p^u(x_u \given \omega^{u-1})}{q^u(x_u \given \omega^{u-1})}
\cdot
\frac{p^u(y_u \given x_u,\omega^{u-1})}{q^u(y_u \given
x_u, \omega^{u-1})} \right) \\
&=
\sum_{\omega^u\in \Omega^u} p^u(\omega^u) \;\ln \left(
\frac{p^u(y_u \given x_u, \omega^{u-1})}{q^u(y_u \given
x_u, \omega^{u-1})} \right) \\
\intertext{[the distribution of $x_u$ given $\omega^{u-1}$
depends only on \A, not on distribution $p$ versus $q$.]}
&=
\sum_{\omega^{u-1}\in\Omega^{u-1}}\;\int_{x_u\in X}\;\sum_{y_u\in\{0,1\}}\;
p^u(y_u \given x_u, \omega^{u-1}) \;\ln \left(
\frac{p^u(y_u \given x_u, \omega^{u-1})}{q^u(y_u \given
x_u, \omega^{u-1})} \right)\;
\diff p^u(\,\cdot\,,\,\omega^{u-1}) \\
&=
\sum_{\omega^{u-1}\in\Omega^{u-1}}\;\int_{x_u\in X}\;
\KL( \mu_0(x_u);\, \mu_i(x_u) \given x_u, \omega^{u-1} )\;
\diff p^u(\,\cdot\,,\,\omega^{u-1}) \\
&=
\sum_{\omega^{u-1}\in\Omega^{u-1}}\;\int_{x_u\in S_i}\;
\KL( \mu_0(x_u)); \mu_i(x_u) \given x_u, \omega^{u-1} )\;
\diff p^u(\,\cdot\,,\,\omega^{u-1}) \\
\intertext{[because $\mu_0=\mu_i(x_u)$ when $x_u \not\in S_i$.]}
&\leq
\sum_{\omega^{u-1}\in\Omega^{u-1}}\;\int_{x_u\in S_i}\;
\frac{4\,\eps^2}{\mu_i(x_u) (1-\mu_i(x_u))}\;
\diff p^u(\,\cdot\,,\,\omega^{u-1}) \\
\intertext{[by Lemma~\ref{lem:kl-bernoulli}
and property (iii) in the definition of ``ensemble"}.]
&\leq
p^u(x_u \in S_i) \cdot \frac{4\eps^2}{3/16}.
\end{align*}
The last inequality holds because
    $\mu_i(x_u)\in [\tfrac13,\tfrac34]$.
The latter holds by property (iii) in the definition of the ``ensemble" and the assumptions that $\mu_0\in[\tfrac13,\tfrac23]$ and $\eps\leq \tfrac{1}{24}$.

Now we can write
\begin{align*}
\KL(p;q) = \sum_{u=1}^t \KL(p^u;q^u \given \omega^{u-1})
&\leq \left( \sum_{u=1}^t p^u(x_u \in S_i) \right) \cdot
\frac{64\;\eps^2}{3} \\
&= \E \left[ Z_i \right] \cdot \frac{64\;\eps^2}{3}\;
\leq \frac{1}{128\, \eps^2} \cdot \frac{64 \eps^2}{3} = \frac{1}{6}.
\end{align*}
Let $\mathcal{E}$ be the event that $Z_i \leq \frac{5t}{3k}.$
By Markov's inequality, $p(\mathcal{E}) \geq 0.4.$
Now using Lemma~\ref{lem:kl-distinguishing} along
with the bound $\KL(p;q) \leq 1/6,$ a short calculation
leads to the bound $q(\mathcal{E}) \geq 0.1,$ and
consequently,
\begin{align*}
\E_q[ t - Z_i ]
&\geq
q(\mathcal{E}) \E_q [ t - Z_i \given \mathcal{E} ] \\
&\geq
0.1 \cdot \left( t - \frac{5t}{3k} \right)
 \geq 0.1 \cdot \left( t - \frac{5t}{6} \right)
 = \frac{t}{60}.
\end{align*}
Assuming the payoff function is $\mu_i$,
the regret of algorithm \A\ increases by $\eps$
each time it chooses a arm $x_u \not\in S_i$.
Hence
\[
R_{(\A,\,\mu_i)}(t) \geq
\eps \E_q [ t - Z_i ] \geq \eps t / 60. \qedhere
\]
\end{proof}

\subsection{Experts lower bound via $(\eps,\delta,k)$-ensembles}

We consider the \emph{feasible experts problem}, in which one is given an action set $X$ along with a collection $\mF$ of Borel probability measures on the set $[0,1]^X$ of
functions $\payoff : X \rightarrow [0,1].$ A problem instance of the feasible experts problem consists of a triple $(X,\mF,\prob)$ where $X$ and $\mF$ are known to the
algorithm, and $\prob \in \mF$ is not. In each round the payoff function $\payoff$ is sampled independently from $\prob$, so that for each action $x\in X$  the (realized) payoff is $\payoff(x)$.

\begin{definition}[Definition~\ref{def:ensemble}, restated]
Consider a set $X$ and a $(k+1)$-tuple
$\vec{\prob} = (\prob_0,\prob_1 \LDOTS \prob_k)$
of Borel probability measures on $[0,1]^X$, the
set of $[0,1]$-valued payoff functions $\payoff$
on $X$.  For $0 \leq i \leq k$ and $x \in X$, let
$\mu_i(x)$ denote the expectation of $\payoff(x)$
under measure $\prob_i$.
We say that $\vec{\prob}$ is an \emph{$(\eps,\delta,k)$-ensemble}
if there exist pairwise disjoint subsets $S_1,S_2,\ldots,S_k \subseteq X$
for which the following properties hold:
\begin{itemize}
\item[(i)] 
for every $i>0$ and every event $\mathcal{E}$ in the Borel
$\sigma$-algebra of $[0,1]^X$, we have
    $$1-\delta < \prob_0(\mathcal{E}) / \prob_i(\mathcal{E}) < 1+\delta.$$
\item[(ii)] 
for every $i > 0$, we have
    $\sup(\mu_i, S_i) - \sup(\mu_i,\, X \setminus S_i) \geq \eps.$
\end{itemize}
\end{definition}

\begin{theorem}[Theorem~\ref{thm:LB-technique}, restated]
Consider the feasible experts problem on $(X,\mF)$. Let  $\vec{\prob}$ be an $(\eps,\delta,k)$-ensemble with $\{\prob_1,\ldots,\prob_k\} \subseteq
\mF$ and $0<\eps,\delta<1/2$. Then for any
    $t < \ln(17k)/(2 \delta^2)$
and any experts algorithm \A, at least half of the measures $\prob_i$ have the property that 
	$R_{(\A,\,\prob_i)}(t) \geq \eps t/2$.
\end{theorem}

\begin{proof}
Let $\Omega = [0,1]^X$.  Using Property (i)
of an $(\eps,\delta,k)$-ensemble combined with
Lemma~\ref{lem:reverse-pinsker}, we find that
$\KL(\prob_i;\prob_0) < \delta^2.$

Let \A\ be an experts algorithm whose random bits
are drawn from a sample space $\Gamma$
with probability measure $\nu$.  For any
positive integer $s < \ln(17k) / 2 \delta^2$,
let $p_i^s$ denote
the measure $\nu \times (\prob_i)^s$
on the probability space $\Gamma \times \Omega^s.$
By the chain rule for KL-divergence
(Lemma~\ref{lem:kl-chainrule}),
$\KL(p_i^s;p_0^s) < s \delta^2 <  \ln(17k) / 2.$
Now let $\mathcal{E}_i^s$ denote the event that
\A\ selects a point $x \in S_i$ at time $s$.
If $p_i^s(\mathcal{E}_i^s) \geq \tfrac{1}{2}$
then Lemma~\ref{lem:kl-distinguishing} implies
\begin{align*}
p_0^s(\mathcal{E}_i^s)
& \geq
p_i^s(\mathcal{E}_i^s) \exp \left(
- \frac{\ln(17k)/2 + 1/e}{p_i^s(\mathcal{E}_i^s)}
\right)
 \geq
\tfrac{1}{2} \exp \left(- \ln(k) + \ln(17) - \tfrac{2}{e} \right)
 > \frac{4}{k}.
\end{align*}
The events $\{\mathcal{E}_i^s \,|\, 1 \leq i \leq k\}$
are mutually exclusive, so fewer than $k/4$ of them
can satisfy $p_0^s(\mathcal{E}_i^s) > \frac{4}{k}.$
Consequently, fewer than $k/4$ of them can satisfy
$p_i^s(\mathcal{E}_i^s) \geq \tfrac{1}{2},$ a
property we denote in this proof by saying that
$s$ is \emph{satisfactory} for $i$.
Now assume $t < \ln(17k)/2 \delta^2$.
For a uniformly random $i \in \{1,\ldots,k\}$,
the expected number of satisfactory
$s \in \{1,\ldots,t\}$
is less than $t/4$, so by Markov's inequality, for
at least half of the $i \in \{1,\ldots,k\}$, the
number of satisfactory $s \in \{1,\ldots,t\}$ is
less than $t/2$.  Property (ii) of an
$(\eps,\delta,k)$-ensemble guarantees that
every unsatisfactory $s$ contributes at least
$\eps$ to the regret of \A\ when the problem
instance is $\prob_i$.  Therefore, at least half
of the measures $\prob_i$ have the property that
	$R_{(\A,\,\prob_i)}(t) \geq \eps t/2$.
\end{proof}

\subsection{Proof of Claim~\ref{cl:logT-KLdiv}}
\label{sec:logT-KLdiv}

Recall that in Section~\ref{sec:logT} we defined
a pair of payoff functions $\mu_0,\mu_i$ and
a ball $B_i$ of radius $r_i$ such that $\mu_0 \equiv \mu_i$
on $X \setminus B_i$, while for $x \in B_i$ we have
    $$ \tfrac38 \leq \mu_0(x) \leq \mu_i(x) \leq
       \mu_0(x) + \tfrac{r_i}{4} \leq \tfrac34.
    $$
Thus, by Lemma~\ref{lem:kl-bernoulli},
$\KL(\mu_0(x);\mu_i(x)) < r_i^2 / 3$ for
all $x \in X$, and $\KL(\mu_0(x);\mu_i(x)) = 0$
for $x \not\in B_i$.

Represent the algorithm's choice and the payoff
observed at any given time $t$ by a pair $(x_t,y_t).$
Let $\Omega = X \times [0,1]$ denote the set of
all such pairs.  When a given algorithm \A\
plays against payoff functions $\mu_0, \mu_i$,
this defines two different probability measures
$p_0^t, p_i^t$ on the set $\Omega^t$ of possible
$t$-step histories.  Let $\omega^t$ denote a
sample point in $\Omega^t$.  The bounds derived
in the previous paragraph imply that for any
non-negative integer $s$,
\begin{equation} \label{eq:logT-KLdiv-1}
\KL(p_0^{s+1}; p_i^{s+1} \,|\, \omega^s) <
\tfrac{1}{3} r_i^2 \prob_0(x_{s+1} \in B_i).
\end{equation}
Summing equation \eqref{eq:logT-KLdiv-1} for
$s=0,1,\ldots,t-1$ and applying Lemma~\ref{lem:kl-chainrule}
we obtain
\begin{equation} \label{eq:logT-KLdiv-2}
\KL(p_0^{t}; p_i^{t}) < \tfrac13 r_i^2 \;
    \textstyle{\sum_{s=1}^{t}}\, \prob_0(x_{s} \in B_i)
 = \tfrac13 r_i^2 \E_0(N_i(t)),
\end{equation}
where the last equation follows from the definition
of $N_i(t)$ as the number of times algorithm \A\
selects a arm in $B_i$ during the first $t$ rounds.

The bound stated in Claim~\ref{cl:logT-KLdiv}
now follows by applying Lemma~\ref{lem:kl-distinguishing}
with the event $S$ playing the role of $\mathcal{E}$,
$\prob_0$ playing the role
of $p$, and $\prob_i$ playing the role of $q$.


\section{Reduction to complete metric spaces}
\label{sec:reduction}

In this section we reduce the Lipschitz MAB problem to that on complete metric spaces.

\begin{lemma}\label{lm:reduction}
The \problem on a metric space $(X,d)$ is $f(t)$-tractable if and only if it is $f(t)$-tractable on the completion of $(X,d)$. Likewise for the \FFproblem with double feedback.
\end{lemma}

\begin{proof}
Let $(X,d)$ be a metric space with completion $(Y,d)$. Since $Y$ contain an isometric copy of $X$, we will abuse notation and consider $X$ as a subset of $Y$. We will present the proof the Lipschitz MAB problem; for the experts problem with double feedback, the proof is similar.

Given an algorithm $\A_X$ which is $f(t)$-tractable for $(X,d)$, we may use it as a Lipschitz MAB algorithm for $(Y,d)$ as well.  (The algorithm has the property that it never selects a point of $Y \setminus X$, but this doesn't prevent us from using it when the metric space is $(Y,d)$.)  The fact that $X$ is dense in $Y$ implies that for every Lipschitz payoff function $\mu$ defined on $Y$, we have
    $ \sup(\mu,X) = \sup(\mu,Y). $
From this, it follows immediately that the regret of $\A_X$, when considered a Lipschitz MAB algorithm for $(X,d)$, is the same as its regret when considered as a Lipschitz MAB algorithm for $(Y,d)$.

Conversely, given an algorithm $\A_Y$ which is $f(t)$-tractable for $(Y,d)$, we may design a Lipschitz MAB algorithm $\A_X$ for $(X,d)$ by running $\A_Y$ and perturbing its output slightly. Specifically, for each point $y\in Y$ and each $t\in \N$ we fix $x = x(y,t) \in X$ such that $d(x, y)<2^{-t}$. If $\A_Y$ recommends playing strategy $y_t \in Y$ at time $t$, algorithm $\A_X$ instead plays $x = x(y,t)$. Let $\pi$ be the observed payoff. Algorithm $\A_X$ draws an independent 0-1 random sample with expectation $\pi$, and reports this sample to $\A_Y$. This completes the description of the modified algorithm $\A_X$.

Suppose $\A_X$ is not $f(t)$-tractable. Then for some problem instance $\mathcal{I}$ on $(Y,d)$, letting $R_X(t)$ be the expected regret of $\A_X$ on this instance, we have that
    $\sup_{t\in\N} R_X(t)/f(t) =\infty$.
Let $\mu$ be the expected payoff function in $\mathcal{I}$.  Consider the following two problem instances of a MAB problem on $Y$, called $\mathcal{I}_1$ and $\mathcal{I}_2$, in which if point $y\in Y$ is played at time $t$, the payoff is an independent 0-1 random sample with expectation $\mu(y)$ and $\mu( x(y,t))$, respectively. Note that algorithm $\A_Y$ is $f(t)$-tractable on $\mathcal{I}_1$, and its behavior on $\mathcal{I}_2$ is identical to that of $\A_X$ on the original problem instance $\mathcal{I}$. It follows that by observing the payoffs of $\A_Y$ one can tell apart $\mathcal{I}_1$ and $\mathcal{I}_2$ with high probability. Specifically, there is a ``classifier" $\mathcal{C}$ which queries one point in each round, such that for infinitely many times $t$ it tell apart $\mathcal{I}_1$ and $\mathcal{I}_2$ with success probability $p(t) \to 1$. Now, the latter is information-theoretically impossible.

To see this, let $H_t$ be the $t$-round history of the algorithm (the sequence of points queried, and outputs received), and consider the distribution of $H_t$ under problem instances $\mathcal{I_1}$ and $\mathcal{I_2}$ (call these distributions $q_1$ and $q_2$). Let us consider  and look at their KL-divergence. By the chain rule (See Lemma~\ref{lem:kl}), we can show that $KL(q_1,q_2) <\tfrac12$. (We omit the details.) It follows that letting $S_t$ be the event that $\mathcal{C}$ classifies the instance as $\mathcal{I}_1$ after round $t$, we have
    $\mathbb{P}_{q_1}[S_t] - \mathbb{P}_{q_2}[S_t] \leq KL(q_1,q_2) \leq \tfrac12$.
For any large enough time $t$,
    $\mathbb{P}_{q_1}[S_t] <\tfrac14$,
in which case $\mathcal{C}$ makes a mistake (on $\mathcal{I}_2$) with constant probability.
\end{proof}

\OMIT{ 
If $\A_Y$ recommends playing strategy $y_t \in Y$ at time $t$, we instead find a point $x_t \in X$ such that $d(x_t,y_t) < 2^{-t}$ and we play $x_t$.  The regret of algorithm $\A_x$ exceeds that of $\A_Y$ by at most $\frac12 + \frac14 + \frac18 + \ldots = 1$, hence $\A_X$ is $f(t)$-tractable.
} 

\begin{lemma}\label{lm:reduction-2}
Consider the \FFproblem with full feedback. If it is $f(t)$-tractable on a metric space $(X,d)$ then it is $f(t)$-tractable on the completion of $(X,d)$.
\end{lemma}

\begin{proof}
Identical to the easy (``only if") direction of Lemma~\ref{lm:reduction}.
\end{proof}

\begin{note}{Remark.}
Lower bounds only require Lemma~\ref{lm:reduction-2}, or the easy (``only if") direction of Lemma~\ref{lm:reduction}. For the upper bounds (algorithmic results), we can either quote the ``if" direction of Lemma~\ref{lm:reduction}, or prove the desired property directly for the specific type algorithms that we use (which is much easier but less elegant).
\end{note}


\OMIT{
For any $B \in \mathcal{B}_{\infty}, \, t \in \N$,
let $\sigma(B,t)$ denote
the value of $\sigma(B)$ sampled at time $t$
when sampling the sequence of i.i.d.~payoff
functions $\payoff_t$ from distribution $\prob_{\mathcal{Q}}$.
We know that $\sigma(B_k,t)=1$ for all $k,t$.
In fact if $S(k,t)$ denotes the set of all
$B \in \mathcal{B}_k(B_{k-1})$ such that
    $\sigma(B,1) = \sigma(B,2) = \cdots = \sigma(B,t) = 1$
then conditional on the value of the set $S(k,t)$, the
value of $B_k$ is distributed uniformly at random
in $S(k,t)$.  As long as this set $S(k,t)$
has at least $n$ elements, the probability
that the algorithm picks a strategy $x_t$ belonging
to $B_k$ at time $t$ is bounded above by $\tfrac{1}{n}$,
even if we condition on the event that
    $x_t \in B_{k-1}.$
For any given $B \in \mathcal{B}_k(B_{k-1}) \setminus \{B_k\}$, we have
    $\prob_{\mathcal{Q}}(B \in S(k,t)) = 2^{-t}$
and these events are independent for different
values of $B$.  Setting $n=\sqrt{n_k}$, so that
$|\mathcal{B}_k(B_{k-1})|=n^2$, we have
    \begin{align}
      \nonumber
    \prob_{\mathcal{Q}} \left( |S(k,t)| \leq n \right) &\leq
    \sum_{R \subset \mathcal{B}_k(B_{k-1}), |R|=n} \prob_{\mathcal{Q}}(S(k,t)
\subseteq R) \\
      \nonumber
    &= \binom{n^2}{n}
    \left( 1-2^{-t} \right)^{n^2-n} \\
      \nonumber
    &< \left( n^2 \cdot \left( 1 - 2^{-t} \right)^{n-1} \right)^{n} \\
      \label{eq:mmlcdlb-3}
    &< \exp \left( n (2 \ln(n) - (n-1)/2^t) \right).
    \end{align}
STILL NEED TO FIX THIS.
Let $t_k = \tfrac14 \log_2(n_k) = \tfrac14 r_k^{-b}$
As long as $t < t_k$, the relation $\log_2(n) = \tfrac12
\log_2(n_k)$ implies $(n-1)/2^t > \sqrt{n}$ so the expression
\eqref{eq:mmlcdlb-3} is bounded above by
$\exp \left( -n \sqrt{n} + 2 n \ln(n) \right)$,
which equals
$\exp \left( -n_k^{3/4} + n_k^{1/2} \ln(n_k) \right)$
and is in turn bounded above by $\exp \left( -\tfrac12 n_k^{3/4} \right).$

For any Lipschitz experts algorithm \A,
let $N(t,k)$ denote the random variable
that counts the number of times \A\ selects
a strategy in $B_k$ during rounds $1,\ldots,t$.
We have already demonstrated that for all $t \leq t_k$,
    \begin{equation} \label{eq:mmlcdlb-4}
       \Pr_{\prob_\mathcal{Q} \in \mathcal{P}}(x_t \in B_k)
       \leq n_k^{-1/2} + \exp \left( -\tfrac12 n_{k}^{3/4} \right)
       < 2 n_k^{-1/2},
    \end{equation}
where the term $n_k^{-1/2}$ accounts for the
event that $S(k,t)$ has at least $\sqrt{n_k}$ elements.
Equation \eqref{eq:mmlcdlb-4} implies the bound
    $\mathbb{E}_{\prob_{\mathcal{Q}} \in \mathcal{P}}[N(t_k,k)] < 2 t_k n_k^{-1/2}.$
By Markov's inequality, the probability that $N(t_k,k) > t_k/2$
is less than $4 n_k^{-1/2}$.
By Borel-Cantelli, almost
surely the number of $k$ such that $N(t_k,k) \leq t_k/2$
is finite.
The algorithm's expected regret at time $t$ is
bounded below by $r_k(t_k - N(t_k,k))$, so with
probability $1$, for all but finitely many $k$ we have
  $$R_{(\A, \, \prob_\mathcal{Q})}(t_k) \geq r_k t_k / 2
    = r_k^{1-b}$
This establishes that \A\ is not $g(t)$-tractable.
}

\section{Topological equivalences: proof of Lemma~\ref{lm:topological-equivalence}}
\label{sec:topological}

Let us restate the lemma, for the sake of convenience. Recall that it includes an equivalence result for compact metric spaces, and two implications for arbitrary metric spaces:

\begin{lemma}\label{lm:topological-equivalence-appendix}
For any compact metric space $(X,d)$, the following are equivalent: (i) $X$ is a countable set, (ii) $(X,d)$ is well-orderable, (iii) no metric subspace of $(X,d)$ is perfect. For an arbitrary metric space we have (ii)$\iff$(iii) and (i)$\Rightarrow$(ii), but not (ii)$\Rightarrow$(i).
\end{lemma}

\begin{proof}[(compact metric spaces)]
Let us prove the assertions in the circular order.

\xhdr{(i) implies (iii).} Let us prove the contrapositive: if $(X,d)$ has a
perfect subspace $Y$, then $X$ is uncountable. We have
seen that if $(X,d)$ has a perfect subspace $Y$ then it has
a ball-tree. Every end $\ell$ of the ball-tree (i.e. infinite path
starting from the root) corresponds to a nested sequence
of balls. The closures of these balls
have the finite intersection property, hence their
their intersection is non-empty. Pick an arbitrary point of
the intersection and call if $x(\ell)$. Distinct ends $\ell$, $\ell'$
correspond to distinct points $x(\ell)$, $x(\ell')$ because if
$(y,r_y), \, (z,r_z)$ are siblings in the ball-tree
which are ancestors of $\ell$ and $\ell'$, respectively,
then the closures of $B(y,r_y)$ and $B(z,r_z)$ are disjoint
and they contain $x(\ell), x(\ell')$ respectively.
Thus we have constructed
a set of distinct points of $X$, one for each end of the
ball-tree. There are uncountably many ends, so $X$ is
uncountable.

\xhdr{(iii) implies (ii).}  Let $\beta$ be some ordinal of strictly larger cardinality than $X$. Let us define a transfinite sequence
	$\{x_\lambda\}_{\lambda\leq \beta}$
of points in $X$ using transfinite recursion\footnote{''Transfinite recursion" is a theorem in set theory which asserts that in order to define a function $F$ on ordinals, it suffices to specify, for each ordinal $\lambda$, how to determine $F(\lambda)$ from $F(\nu)$, $\nu<\lambda$.}, by specifying that $x_0$ is any isolated point of $X$, and that for any ordinal $\lambda > 0$, $x_\lambda$ is any isolated point of the subspace $(Y_\lambda, d)$, where
	$Y_\lambda = X\setminus \{ x_\nu:\, \nu < \lambda \}$,
as long as $Y_\lambda$ is nonempty. (Such isolated point exists since by our assumption subspace $(Y_\lambda, d)$ is not perfect.) If $Y_\lambda$ is empty define e.g. $x_\lambda = x_0$.
Now, $Y_\lambda$ is empty for some ordinal $\lambda$ because otherwise we obtain a mapping from $X$ onto an ordinal $\beta$ whose cardinality exceeds the cardinality of $X$. Let
	$\beta_0 = \min\{ \lambda:\, Y_\lambda=\emptyset \}$.
Then every point in $X$ has been indexed by an ordinal number $\lambda<\beta_0$, and so we obtain a well-ordering of $X$. By construction, for
every $x=x_\lambda$ we can define  a radius $r(x)>0$ such that $B(x, r(x))$ is disjoint from the set of points $\{ x_\nu : \nu > \lambda \}$.
Any initial segment $S$ of the well-ordering is equal to
the union of the balls $\{B(x,r(x)) : x \in S\}$, hence is an
open set in the metric topology. Thus we have constructed
a topological well-ordering of X.

\xhdr{(ii) implies (i).} Suppose we have a binary relation $\prec$
which is a topological well-ordering of $(X,d)$. Let $S(n)$ denote
the set of all $x \in X$ such that $B(x, \tfrac{1}{n})$ is contained in
the set $P(x) = \{ y : y \preceq x \}$. By the definition of a
topological well-ordering we know that for every $x$, $P(x)$ is
an open set, hence $x\in S(n)$ for sufficiently
large $n$. Therefore $X = \cup_{n\in\N} S(n)$. Now, the definition
of $S(n)$ implies that every two points of $S(n)$ are
separated by a distance of at least $1/n$. (If $x$ and $z$
are distinct points of $S(n)$ and $x\prec z$, then $B(x,\tfrac{1}{n})$
is contained in the set $P(x)$ which does not contain
$z$, hence $d(x,z)\geq \tfrac{1}{n}$.) Thus by compactness of $(X,d)$ set $S(n)$ is finite.
\end{proof}

\begin{proof}[(arbitrary metric spaces)]
For implications {\em (i)$\Rightarrow$(ii)}  and {\em (iii)$\Rightarrow$(ii)}, the proof above does not in fact use compactness. An example of an uncountable but well-orderable metric space is $(\R,d)$, where $d$ is a uniform metric. It remains to prove that {\em (ii)$\Rightarrow$(iii)}.

Suppose there exists a topological well-ordering $\prec$. For each subset $Y\subseteq X$ and an element $\lambda\in Y$ let
	$Y_\prec(\lambda) = \{ y\in Y: y\preceq \lambda\}$
be the corresponding initial segment.

We claim that $\prec$ induces a topological well-ordering on any subset $Y\subseteq X$. We need to show that for any $\lambda\in Y$ the initial segment $Y_\prec(\lambda)$ is open in the metric topology of $(Y,d)$. Indeed, fix $y\in Y_\prec(\lambda)$. The initial segment $X_\prec(\lambda)$ is open by the topological well-ordering property of $X$, so
	$B_X(y,\eps) \subset X_\prec(\lambda)$
for some $\eps>0$. Since
	$Y_\prec(\lambda) =  X_\prec(\lambda) \cap Y$
and
	$B_Y(y,\eps) = B_X(y,\eps) \cap Y $,
it follows that
	$B_Y(y,\eps) \subset Y_\prec(\lambda)$.
Claim proved.

Suppose the metric space $(X,d)$ has a perfect subspace $Y\subset X$. Let $\lambda$ be the $\prec$-minimum element of $Y$. Then $Y_\prec(\lambda) = \{\lambda \}$. However, by the previous claim $\prec$ is a topological well-ordering of $(Y,d)$, so the initial segment $Y_\prec(\lambda)$ is open in the metric topology of $(Y,d)$. Since $(Y,d)$ is perfect, $Y_\prec(\lambda)$ must be infinite, contradiction. This completes the {\em (ii)$\Rightarrow$(iii)} direction.
\end{proof}

\section{Log-covering dimension: the Earthmover distance example}
\label{app:earthmover}

We flesh out the example from Section~\ref{sec:intro-experts}. Fix a metric space $(X,\mD)$ of finite diameter and covering dimension $\kappa<\infty$. Let $\mathcal{P}_X$ denote the set of all probability measures over $X$. Let $(\mathcal{P}_X,W_1)$ be the space of all probability measures over $(X,\mD)$ under the Wasserstein $W_1$ metric, a.k.a. the Earthmover distance:
\begin{align*}
    W_1(\nu,\nu') = \inf \E\left[ \|Y-Y'\|_2 \right],
\end{align*}
where the infimum is taken over all joint distributions $(Y,Y')$ on $X\times X$
with marginals $\nu$ and $\nu'$ respectively (for any two
$\nu,\nu'\in \mathcal{P}_X$).

\begin{theorem}\label{thm:LCD-earthmover}
The log-covering dimension of $(\mathcal{P}_X,W_1)$ is $\kappa$.
\end{theorem}

For the sake of completeness: for any $\mu,\mu'\in \mathcal{P}_X$, the Wasserstein $W_1$ metric, a.k.a. the Earthmover distance, is defined as
    $W_1(\nu,\nu') = \inf \E[\, \mD(Y,Y') \,],$
where the infimum is taken over all joint distributions $(Y,Y')$ on $X\times X$ with marginals $\nu$ and $\nu'$ respectively.

In the remainder of this subsection we prove Theorem~\ref{thm:LCD-earthmover}.

\begin{proof}[(Theorem~\ref{thm:LCD-earthmover}: upper bound)]
Let us cover $(\mathcal{P}_X,W_1)$ with balls of radius $\tfrac{2}{k}$ for some $k\in\N$. Let S be a $\tfrac{1}{k}$-net in $(X,d)$; note that $|S|=O(k^\kappa)$ for a sufficiently large $k$. Let $P$ be the set of all probability distributions $p$ on $(X,d)$ such that $\texttt{support}(p) \subset S$ and for every point $x \in S$, $p(x)$ is a rational number with denominator $k^{d+1}$. The cardinality of $P$ is bounded above by $(k^{\kappa+1})^{k^\kappa}$. It remains to show that balls of radius $\tfrac{2}{k}$ centered at the points of $P$ cover the entire space $(\mathcal{P}_X,W_1)$. This is true because:
\begin{itemize}
\item every distribution $q$ is $\tfrac{1}{k}$-close to a distribution $p$ with support contained in $S$ (let $p$ be the distribution defined by randomly sampling a point of $(X,d)$ from $q$ and then outputting the closest point of $S$);
\item every distribution with support contained in $S$ is $\tfrac{1}{k}$-close to a distribution in $P$ (round all probabilities down to the nearest multiple of $k^{-(\kappa+1)}$; this requires moving only $\tfrac{1}{k}$ units of stuff). \qedhere
\end{itemize}
\end{proof}

To prove the lower bound, we make a connection to the Hamming metric.

\begin{lemma}  \label{lem:embedding}
Let $(X,d)$ be any metric space, and let $H$ denote the Hamming metric on the Boolean cube $\{0,1\}^n$.
If $S \subseteq X$ is a subset of even cardinality $2n$,
and $\eps$ is a lower bound on the distance between any two points
of $S$, then there is a mapping $f \,:\, \{0,1\}^n \rightarrow
\mathcal{P}_X$ such that for all $a,b \in \{0,1\}^n,$
\begin{align} \label{eq:earthmover-lb}
W_1(f(a),f(b)) \geq \tfrac{\eps}{n}\; H(a,b).
\end{align}
\end{lemma}
\begin{proof}
Group the points of $S$ arbitrarily into pairs $S_i = \{x_i,y_i\}$,
where $i=1,\ldots,n.$  For $a \in \{0,1\}^n$ and $1 \leq i \leq n$,
define $t_i(a) = x_i$ if $a_i=0$, and $t_i(a) = y_i$ otherwise.
Let $f(a)$ be the uniform distribution on the set
$\{t_1(a),\ldots,t_n(a)\}.$
To prove \eqref{eq:earthmover-lb}, note that if
$i$ is any index such that $a_i \neq b_i$ then
$f(a)$ assigns probability $\tfrac{1}{n}$ to $t_i(a)$ while
$f(b)$ assigns zero probability to the entire ball
of radius $\eps$ centered at $t_i(a).$  Consequently,
the $\tfrac{1}{n}$ units of probability at $t_i(a)$ have to
move a distance of at least $\eps$ when shifting
from distribution $f(a)$ to $f(b)$.  Summing over
all indices $i$ such that $a_i \neq b_i$, we obtain~\refeq{eq:earthmover-lb}.
\end{proof}

The following lemma, asserting the existence of
asymptotically good binary error-correcting codes,
is well known, e.g. see \citep{Gilbert,Varshamov}.

\begin{lemma} \label{lem:asympt-good}
Suppose $\delta, \rho$ are constants satisfying
$0 < \delta < \frac12$ and
$0 \leq \rho < 1 + \delta \log_2(\delta) + (1-\delta) \log_2(1-\delta).$
For every sufficiently large $n$, the Hamming cube
$\{0,1\}^n$ contains more than $2^{\rho n}$ points,
no two of which are nearer than distance $\delta n$
in the Hamming metric.
\end{lemma}

Combining these two lemmas, we obtain an easy proof for the lower bound in Theorem~\ref{thm:LCD-earthmover}.

\begin{proof}[(Theorem~\ref{thm:LCD-earthmover}: lower bound)]
Consider any $\gamma < \kappa$.  The hypothesis on the covering
dimension of $(X,d)$ implies that for all sufficiently small
$\eps$, there exists a set $S$ of cardinality $2n$ --- for some
$n > \eps^{-\gamma}$ --- such that the minimum distance between
two points of $S$ is at least $5 \eps.$  Now let $\mathcal{C}$
be a subset of $\{0,1\}^n$ having at least $2^{n/5}$ elements,
such that the Hamming distance between any two points of $\mathcal{C}$
is at least $n/5.$  Lemma~\ref{lem:asympt-good} implies that
such a set $\mathcal{C}$ exists, and we can then apply
Lemma~\ref{lem:embedding} to embed $\mathcal{C}$ in
$\mathcal{P}_X$, obtaining a subset of $\mathcal{P}_X$ whose
cardinality is at least $2^{\eps^{-\gamma} / 5}$, with
distance at least $\eps$ between every pair of points in the set.
Thus, any $\eps$-covering of $\mathcal{P}_X$ must contain at
least $2^{\eps^{-\gamma}/5}$ sets, implying that
$\LCD(\mathcal{P}_X,W_1) \geq \gamma.$  As $\gamma$ was an
arbitrary number less than $\kappa$, the proposition is proved.
\end{proof}

\end{document}